\documentclass[a4paper, UKenglish]{article}

\usepackage{latexsym}
\usepackage{amssymb}
\usepackage{amsmath,mathabx}
\usepackage{amsthm}
\usepackage{pstricks}
\usepackage{subfig}
\usepackage{algorithm}
\usepackage[noend]{algorithmic}
\usepackage{xspace,enumerate}
\usepackage{datetime}
\usepackage{tipa}
\usepackage{textcmds}

\usepackage{chngcntr}

\usepackage[normalem]{ulem}

\usepackage{tikz}
\usepackage{tikz-qtree}
\usepackage{chapterbib}
\usetikzlibrary{arrows,automata,backgrounds,calendar,chains,decorations,er,fadings,fit,matrix,mindmap,folding,patterns,petri,plothandlers,plotmarks,shadows,shapes,topaths,through,trees}

\theoremstyle{definition}
\newtheorem{theorem}{Theorem}
\newtheorem{definition}[theorem]{Definition}
\newtheorem{proposition}[theorem]{Proposition}

\newtheorem{lemma}[theorem]{Lemma}
\newtheorem{corollary}[theorem]{Corollary}

\graphicspath{{./}}

\begin{document}
	\newcommand{\nc}{\newcommand}


\newcommand  {\subs}     {\subseteq}
\renewcommand{\sups}{\supseteq}
\newcommand{\covers}{\sqsupseteq}
\newcommand{\covered}{\sqsubseteq}
\newcommand{\mydef}{\ensuremath{\mathrel{\smash{\stackrel{\scriptscriptstyle{
    \text{def}}}{=}}}}}

\newcommand{\ohne}{\ensuremath{\setminus}}
\newcommand{\Pot}{{\ensuremath{\mathcal{P}}}}

\newcommand{\eps}{{\ensuremath{\epsilon}}}
\newcommand{\pprime}{{\ensuremath{\prime \prime}}}

\newcommand {\nat}      {\mathbb{N}}

\newcommand{\calA}{{\ensuremath{\mathcal{A}}}\xspace}
\newcommand{\calB}{{\ensuremath{\mathcal{B}}}\xspace}
\newcommand{\calC}{\ensuremath{\mathcal{C}}\xspace}
\newcommand{\calD}{\ensuremath{\mathcal{D}}\xspace}
\newcommand{\calE}{\ensuremath{\mathcal{E}}\xspace}
\newcommand{\calG}{\ensuremath{\mathcal{G}}\xspace}
\newcommand{\calP}{\ensuremath{\mathcal{P}}\xspace}
\newcommand{\calQ}{\ensuremath{\mathcal{Q}}\xspace}
\newcommand{\calS}{\ensuremath{\mathcal{S}}\xspace}

\newcommand{\bigO}{\ensuremath{\mathcal{O}}}

\newcommand{\leftidx}[3]{{\vphantom{#2}}#1#2#3}
\newcommand{\lequiv}[3][]{\ensuremath{\tensor*[_{#2}]{\equiv}{^{#1}_{#3}}}}

\newcommand{\id}{\ensuremath{\text{id}}}

\newcommand{\hatE}{\ensuremath{\hat{E}}}
\newcommand{\hate}{\ensuremath{\hat{e}}}
\newcommand{\hatP}{\ensuremath{\hat{P}}}
\newcommand{\hatQ}{\ensuremath{\hat{Q}}}
\newcommand{\hatS}{\ensuremath{\hat{S}}}
\newcommand{\hatG}{\ensuremath{\hat{G}}}
\newcommand{\ol}[1]{\overline{#1}}

\newcommand{\posbool}{\ensuremath{\mathcal{B}^+}}

\newcommand{\Sigmat}{\tilde{\Sigma}}

\newcommand{\smin}[1]{\ensuremath{\left[#1\right]_{\text{min}}}}


\def\pone{\textsc{Juliet}\xspace}
\def\ptwo{\textsc{Romeo}\xspace}
\def\pones{\textsc{Juliet's}\xspace} 
\def\ptwos{\textsc{Romeo's}\xspace}
\def\ponea{\textsc{J}\xspace} 
\def\ptwoa{\textsc{R}\xspace}

\def\playerA{\textsc{Adam}\xspace}
\def\playerE{\textsc{Eve}\xspace}
\def\playerAs{\textsc{Adam's}\xspace}
\def\playerEs{\textsc{Eve's}\xspace}

\newcommand{\funcsymb}{\Gamma}
\newcommand{\ponesymb}{\funcsymb_{\ponea}}
\newcommand{\ptwosymb}{\funcsymb_{\ptwoa}}

\newcommand{\movesingle}[1]{\ensuremath{\stackrel{#1}{\longrightarrow}}}
\newcommand{\movesingles}{\movesingle{\sigma}}
\newcommand{\move}[2]{\ensuremath{\stackrel{#1,#2}{\longrightarrow}}}
\newcommand{\moves}[1]{\move{\sigma}{#1}}
\newcommand{\movest}{\move{\sigma}{\tau}}

\newcommand{\tmovesingle}[1]{\ensuremath{\stackrel{#1}{{\longrightarrow}^*}}}
\newcommand{\tmovesingles}{\tmovesingle{\sigma}}
\newcommand{\tmove}[2]{\ensuremath{\stackrel{#1,#2}{\longrightarrow}^*}}
\newcommand{\tmoves}[1]{\tmove{\sigma}{#1}}
\newcommand{\tmovest}{\tmove{\sigma}{\tau}}

\newcommand {\play}      {{\ensuremath{\Pi}}\xspace}

\newcommand{\passes}{\ensuremath{\text{LS}}\xspace}

\newcommand{\selG}{\ensuremath{\tilde{G}}\xspace}
\newcommand{\Glr}{\ensuremath{G_{\text{L2R}}}\xspace}

\newcommand{\Depth}[1][G]{\ensuremath{\text{Depth}^{#1}}}


\newcommand{\algprobname}[1]{\ensuremath{\text{\sc{#1}}}\xspace}
\newcommand{\lrall}{\algprobname{L2RAll}\xspace}
\newcommand{\corridor}{\algprobname{TwoPlayer2ExpCorridorTiling}\xspace}

\providecommand  {\myclass} [1]  {\ensuremath{\mbox{\sc #1}}\xspace}

\newcommand{\PTIME}{\myclass{PTIME}}
\newcommand{\PSPACE}{\myclass{PSPACE}}
\newcommand{\LOGSPACE}{\myclass{LOGSPACE}}
\newcommand     {\EXPTIME}  {\myclass{EXPTIME}}
\newcommand     {\iiEXPTIME}  {\myclass{2-EXPTIME}}
\newcommand     {\EXPSPACE}  {\myclass{EXPSPACE}}
\newcommand     {\NEXPSPACE}  {\myclass{NEXPSPACE}}
\newcommand     {\NEXPTIME}  {\myclass{NEXPTIME}}
\newcommand     {\coNEXPTIME}  {\myclass{co-NEXPTIME}}
\newcommand     {\iiiEXPTIME}  {\myclass{3-EXPTIME}}

\newcommand{\QBF}{\algprobname{QBF}}

\newcommand{\icup}{\ensuremath{\oplus}} 
\newcommand{\bigicup}{\ensuremath{\bigoplus}} 


\nc{\win}{JWin}

\newcommand{\safeu}{\ensuremath{\text{\win}}}
\newcommand{\safelrrf}{\ensuremath{\text{\win}_{L2R}^{\text{rf}}}}
\newcommand{\safelrol}[1]{\mbox{\ensuremath{\Sigma^*\setminus\safelr(#1)}}}
\newcommand{\safelrp}{\ensuremath{\text{\win}_{\text{L2R}^+}}}
\newcommand{\safelsi}{\ensuremath{\text{\win}_{LS1}}}
\newcommand{\safeop}[1][3]{\ensuremath{\text{\win}_{1P(#1)}}}
\newcommand{\safeR}{\safeop[4]}


\newcommand{\ltrp}{\ensuremath{\text{L2R}^+}}
\newcommand{\stratu}{\ensuremath{\text{STRAT}}}
\newcommand{\stratlr}{\ensuremath{\text{STRAT}}} 
\newcommand{\stratlrrf}{\ensuremath{\text{STRAT}_{L2R}^{rf}}}
\newcommand{\stratlrpc}{\ensuremath{\text{STRAT}_{L2R}^{pc}}}
\newcommand{\stratrf}{\ensuremath{\text{STRAT}^\text{rf}_{L2R}}}
\newcommand{\stratls}{\ensuremath{\text{STRAT}_{LS1}}}
\newcommand{\stratp}[1][3]{\ensuremath{\text{STRAT}_{1P(#1)}}}
\newcommand{\strata}{\ensuremath{\text{STRAT}_{\ponea}}}
\newcommand{\stratac}{\ensuremath{\text{STRAT}_{\ponea,\Call}}}
\newcommand{\stratb}{\ensuremath{\text{STRAT}_{\ptwoa}}}
\newcommand{\stratr}[1][G]{\ensuremath{\stratb(#1)}\xspace}

\newcommand{\Astrat}{{\ensuremath{\sigma}}}
\newcommand{\Bstrat}{{\ensuremath{{\tau}}}}

\newcommand{\ov}[1]{\ensuremath{\overline{#1}}}
\newcommand{\ovS}{\ensuremath{\ov{S}}}
\newcommand{\ovAstrat}{\ensuremath{\ov{\Astrat}}}
\newcommand{\ovBstrat}{\ensuremath{\ov{\Bstrat}}}

\newcommand{\func}[1][]{\ensuremath{\text{func}^{#1}}}
\newcommand{\funcs}[1][]{\ensuremath{\text{funcs}^{#1}}}
\newcommand{\funcslr}[1][G]{\ensuremath{s_\text{L2R}^{#1}}}
\newcommand{\word}[1][G]{\ensuremath{\text{word}_{#1}}}
\newcommand{\words}[1][G]{\ensuremath{\text{words}_{#1}}}
\newcommand{\state}[1][G]{\ensuremath{\text{state}_{#1}}}
\newcommand{\states}[1][G]{\ensuremath{\text{states}_{#1}}}
\newcommand{\suf}[1][G]{\ensuremath{\text{suf}^{#1}}}
\newcommand{\sufstates}[1][G]{\ensuremath{\text{sufstates}^{#1}}}
\newcommand{\Bstates}[1][]{\ensuremath{\text{2-states}^{#1}}}
\newcommand{\transit}{\ensuremath{\text{transit}}}

\newcommand{\cstate}[1][C]{\ensuremath{\text{state}_{\game{#1}}}}
\newcommand{\cstates}[1][C]{\ensuremath{\text{states}_{\game{#1}}}}


\newcommand{\wset}{\ensuremath{\mathcal{L}}}
\newcommand{\wsetRead}{\ensuremath{\mathcal{L}_{Read}}}
\newcommand{\wsetCall}{\ensuremath{\mathcal{L}_{Call}}}

\nc{\Call}{\ensuremath{\text{Call}}\xspace}
\nc{\Read}{\ensuremath{\text{Read}}\xspace}
\nc{\LS}{\ensuremath{\text{LS}}\xspace}
\nc{\Stop}{\ensuremath{\text{Stop}}\xspace}

\newcommand{\deriv}{\ensuremath{\Rightarrow}}


\newcommand{\effect}[2][]{\ensuremath{\calE^{#1}[#2]}\xspace} 
\newcommand{\geffect}[2][]{\ensuremath{\calE^{#1}[G,#2]}\xspace} 
\newcommand{\ceffect}[1][G]{\ensuremath{\calC[#1]}\xspace} 
\newcommand{\cneffect}[2][G]{\ensuremath{\calC^{#2}[#1]}\xspace} 
\newcommand{\game}[1]{\ensuremath{\calG(#1)}\xspace} 

\newcommand{\safelr}[1][]{\ensuremath{\text{\win}^{#1}}}
\newcommand{\Safelr}[1][]{\ensuremath{\textsc{JWin}^{#1}}}

\newcommand{\Gloc}{\ensuremath{\calG_{\text{loc}}}}
\newcommand{\Gnloc}{\ensuremath{\calG_{\text{nd,loc}}}}
\newcommand{\Gxsc}{\ensuremath{\calG_{\text{xsc}}}}
\newcommand{\Gall}{\ensuremath{\calG_{\text{all}}}}
\newcommand{\Gnall}{\ensuremath{\calG_{\text{nd,all}}}}
\newcommand{\Gfin}{\ensuremath{\calG_{\text{fin}}}}
\newcommand{\Gnd}{\ensuremath{\calG_{\text{nd}}}}
\newcommand{\Gndf}{\ensuremath{\calG_{\text{ndf}}}}
\newcommand{\Gins}{\ensuremath{\calG_{\text{ins}}}}
\newcommand{\Gval}[1][]{\ensuremath{\calG_{\text{val}}^{#1}}}
\newcommand{\Gflat}{\ensuremath{\calG_{\text{flat}}}}

\newcommand{\Floc}{\ensuremath{F_{\text{loc}}}}

\newcommand    {\algproblem} [3] 
{\fbox{\parbox[t]{.65\textwidth}{\centerline{#1}\begin{tabular}{lp{.45\textwidth}} Given: & #2\\ Question: & #3\end{tabular}}}}


\newcommand{\anRC}[1]{\ensuremath{{#1}_{\text{RC}}}}
\newcommand{\anR}[1]{\ensuremath{{#1}_{\text{R}}}}
\newcommand{\anC}[1]{\ensuremath{{#1}_{\text{C}}}}
\newcommand{\anNull}[1]{\ensuremath{{#1}_{\text{0}}}}
\newcommand{\anRandC}[1]{\ensuremath{{#1}_{\text{R$\land$C}}}}
\newcommand{\anRorC}[1]{\ensuremath{{#1}_{\text{R$\lor$C}}}}

\newcommand{\unan}[1]{\ensuremath{\natural({#1})}}

\newcommand{\Contract}{\ensuremath{\text{Contract}}}
\newcommand{\Root}{\ensuremath{\text{root}}}

\newcommand{\linstate}[1]{\ensuremath{\overline{#1}}}


\newcommand{\calls}[1][G]{\ensuremath{S^{#1}_{\text{call}}}}
\newcommand{\summary}{s}
\newcommand{\seffect}[2][1]{\ensuremath{E^{#1}(#2)}} 
\newcommand{\teffect}{\ensuremath{\tilde{E}}} 
\newcommand{\releffect}[2][]{\ensuremath{e^{#1}(#2)}\xspace} 
\newcommand{\deffect}[2][]{\ensuremath{\hat{E}^{#1}[#2]}\xspace}
\newcommand{\fctE}{\ensuremath{E_{\text{fct}}}}
\newcommand{\rveffect}{\ensuremath{e_{v}}}
\newcommand{\veffect}{\ensuremath{E_{v}}}
\newcommand{\finE}{\ensuremath{E^{<\infty}}}

\newcommand{\dreleffect}[1]{\ensuremath{\hat{e}(#1)}\xspace} 
\newcommand{\ehat}{\ensuremath{\hat{e}}\xspace} 

\newcommand{\uhat}{\ensuremath{\hat{u}}\xspace}

\newcommand{\effectset}[2][2]{\ensuremath{\mathcal{E}_{#1}(#2)}\xspace}

\newcommand{\Norm}{\ensuremath{\textsc{Norm}}\xspace}
\newcommand{\Mix}{\ensuremath{\textsc{Mix}}\xspace}
\newcommand{\SMix}{\ensuremath{\textsc{SMix}}\xspace}
\newcommand{\Pick}{\ensuremath{\textsc{Shuffle}}\xspace}
\newcommand{\Shuffle}{\ensuremath{\textsc{Shuffle}}\xspace}

\newcommand{\comp}{\ensuremath{\text{comp}}}
\newcommand{\decomp}{\ensuremath{\text{decomp}}}


\newcommand{\Alr}{{\ensuremath{A_{\text{L2R}}}}}
\newcommand{\Als}{{\ensuremath{A_{\text{LS1}}}}}

\newcommand{\deltalr}{{\ensuremath{\delta_{\text{L2R}}}}\xspace}
\newcommand{\deltaslr}{{\ensuremath{\delta^*_{\text{L2R}}}}\xspace}
\newcommand{\deltalrhat}{{\ensuremath{\hat{\delta}_{\text{L2R}}}}\xspace}
\newcommand{\deltaslrhat}{{\ensuremath{\hat{\delta}^*_{\text{L2R}}}}\xspace}

\newcommand{\dAlr}{{\ensuremath{\widehat{A}_{\text{L2R}}}}}

\newcommand{\Asep}{{\ensuremath{A_{\text{sep}}}}}
\newcommand{\Aut}[2][]{{\ensuremath{A_{#1 #2}}}}

\newcommand{\ovdelta}{{\ensuremath{\overline{\delta}}}}
\newcommand{\tdelta}{{\ensuremath{\tilde{\delta}}}}

\newcommand{\reach}[1][]{\ensuremath{\stackrel{#1}{\leadsto}}}

\newcommand{\Yield}{\ensuremath{\text{Yield}}}

\newcommand{\layer}[1][]{\ensuremath{#1\text{-layer}}}


\newcommand{\wf}[1][\Sigma]{\ensuremath{\text{WF}(#1)}}
\newcommand{\swf}[1][\Sigma]{\ensuremath{\text{sWF}(#1)}}
\newcommand{\op}[1]{\ensuremath{\text{\small\textlangle}{#1}\text{\small\textrangle}}}
\newcommand{\cl}[1]{\ensuremath{\text{\small\textlangle}/{#1}\text{\small\textrangle}}}

\newcommand{\nw}[1]{\ensuremath{\textroundcap{#1}}}

\newcommand{\true}{\ensuremath{\text{true}}}
\newcommand{\false}{\ensuremath{\text{false}}}

\newcommand{\Aug}{\ensuremath{\text{Aug}}}
\newcommand{\Enc}{\ensuremath{\text{Enc}}}

\newcommand{\Strip}{\ensuremath{\text{Strip}}}
\newcommand{\Val}{\ensuremath{\text{Val}}}


\newcommand{\We}{\ensuremath{\mathcal{W}}}
\newcommand{\Se}{\ensuremath{\mathcal{S}}}
\newcommand{\Te}{\ensuremath{\mathcal{T}}}
\newcommand{\Ge}{\ensuremath{\mathcal{G}}}
\newcommand{\He}{\ensuremath{\mathcal{H}}}

\newcommand{\Hh}{\ensuremath{\mathcal{H}}}
\newcommand{\Vv}{\ensuremath{\mathcal{V}}}

\newcommand{\ovWe}{\ensuremath{\overline{\mathcal{W}}}}
\newcommand{\ovTe}{\ensuremath{\overline{\mathcal{T}}}}
\newcommand{\ovGe}{\ensuremath{\overline{\mathcal{G}}}}

\newcommand{\TT}{\ensuremath{\mathbb{T}}}
\newcommand{\GG}{\ensuremath{\mathbb{G}}}
\newcommand{\NN}{\ensuremath{\mathbb{N}}}


\newcommand{\added}[1]{\color{blue}#1\color{black}}

\newcommand{\mcomment}[2]{{\footnotesize\color{blue}(#1)}\footnote{\color{blue}#1: #2}} 

\newcommand{\tsm}[1]{\mcomment{TS}{#1}}
\newcommand{\msm}[1]{\mcomment{MS}{#1}}

\newcommand{\thomas}[1]{\ \\ \fbox{\parbox{\linewidth}{{\sc Thomas}:\\
      #1}}}
\nc{\thomasm}{\tsm}

\newcommand{\martin}[1]{\ \\ \fbox{\parbox{\linewidth}{{\sc Martin}:\\ #1}}}
\nc{\martinm}{\msm}


\newcommand{\bild}[4]{        
\begin{figure}[h]
	\begin{center}
		\scalebox{#1}{\input{Bilder/#2}} \caption{#3} \label{#4}
	\end{center}          
\end{figure}           
}

\newcommand{\ignore}[1]{}

\newcommand{\skipproof}[1]{}

\newenvironment{restate}[1]
{\noindent \textbf{#1 \itshape (restated).} \itshape} 
{}

\newenvironment{proofof}[1]
{\noindent \textbf{Proof of #1.}} 
{\qed}



	\title{Games for Active XML Revisited}
	\author{Martin Schuster \and Thomas Schwentick\\ \and \small TU Dortmund University}
	\maketitle

	\begin{abstract}
		The paper studies the rewriting mechanisms for intensional documents in the Active XML framework, abstracted in the form of \emph{active context-free games}. The \emph{safe rewriting} problem studied in this paper is to decide whether the first player, \pone, has a winning strategy for a given game and (nested) word; this corresponds to a successful rewriting strategy for a given intensional document. The paper examines several extensions to active context-free games.

The primary extension allows more expressive schemas (namely XML schemas and regular nested word languages) for both target and replacement languages and has the effect that games are played on nested words instead of (flat) words as in previous studies. Other extensions consider validation of input parameters of web services, and an alternative semantics based on insertion of service call results.

In general, the complexity of the safe rewriting problem is highly intractable (doubly exponential time), but the paper identifies interesting tractable cases.
	\end{abstract}

		\section{Introduction}

\paragraph*{Scientific context} 
This paper contributes to the theoretical foundations of intensional
documents, in the framework of \emph{Active XML}
\cite{AbiteboulBM08}. It studies game-based abstractions of the
mechanism transforming intensional documents into documents of a
desired form by calling web services. One form of such games has been
introduced under the name \emph{active context-free games} in
\cite{MuschollSS06} as an abstraction of a problem studied in
\cite{MiloAABN05}.\footnote{Actually, the two notions were introduced in  the respective conference papers.} The setting in \cite{MiloAABN05} is as follows: an \emph{Active XML} document is given, where some elements consist of functions representing web services that can be called. The goal is to rewrite the document by a series of web service calls into a document matching a given \emph{target schema}. 

Towards an intuition of Active XML document rewriting, consider the example in Figure \ref{fig:rewriting} of an online local news site dynamically loading information about weather and local events (adapted from \cite{MiloAABN05} and \cite{MuschollSS06}). Figure \ref{subfig:before} shows the initial Active XML document for such a site, containing \emph{function nodes} which refer to a weather and an event service, respectively, instead of concrete weather and event data. After a single function call to each of these services has been materialised, the resulting document may look like the one depicted in Figure \ref{subfig:after}. Note that the rewritten document now contains new function nodes; further rewriting might be necessary to reach a document in a given target schema (which could, for instance, require that the document contains at least one indoor event if the weather is rainy). 

\begin{figure}[h]
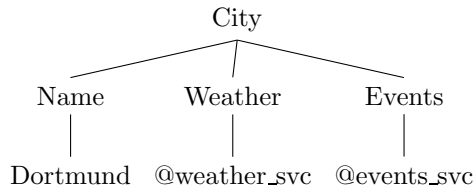
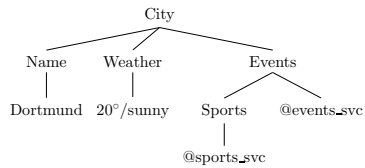

	\centering
	\subfloat[c][Example document before rewriting.]{
		\label{subfig:before}
		\Tree [.City [.Name Dortmund ] [.Weather @weather\_svc ] [.Events @events\_svc ] ]
	}
	\qquad
	\subfloat[c][Same document after function calls.]{
		\label{subfig:after}
		\resizebox{.4\textwidth}{!}{		
			\Tree [.City [.Name Dortmund ] [.Weather 20$^\circ$/sunny ] [.Events [.Sports @sports\_svc ] @events\_svc ] ]	
		}	
	}\\
	\caption{Example of Active XML rewriting.}
	\label{fig:rewriting}
\end{figure}

Modelling this rewriting problem as a game follows the approach of dealing with uncertainty by playing a ``game against nature'': We model the process intended to rewrite a given document into a target schema by performing function calls as a player (\pone). As her moves, she chooses which function nodes to call, and her goal is to reach a document in the target schema. Returns of function calls, on the other hand, are chosen (in accordance with some schema for each called service) by an antagonistic second player (\ptwo), whose goal is to foil \pone. The question whether a given document can always be rewritten into the target schema may then be solved by deciding whether \pone has a winning strategy. More specifically, given an input document, target schema and return schemas for function calls, there should exist a \emph{safe rewriting} algorithm that always rewrites the input document into the target schema, no matter the concrete returns of function calls, if and only if \pone has a winning strategy in the corresponding game.\footnote{It is hard to give a precise statement of \emph{safe rewriting} that does not already involve games, but we hope that the general idea of this statement becomes sufficiently clear.}

In \cite{MiloAABN05}, the target schema is represented by an XML
document type definition (DTD). It was argued that, due to
the restricted nature of DTDs, the problem can
be reduced to a rewriting game on strings where, in each move a single
symbol is replaced by a string, the
set of allowed replacement strings for each symbol is a regular language and the
target language is regular\footnote{More precisely, it should be given by  a \emph{deterministic} regular expression.}, 
as well.  

In \cite{MuschollSS06} the
complexity of the problem to determine the winner in such games
(mainly with finite replacement languages) was studied. Whereas this
problem is undecidable in general, there are important cases in which
it can be solved, particularly if \pone chooses the symbols to be replaced in a
left-to-right fashion.  In and after \cite{MuschollSS06,MiloAABN05}, research very much concentrated on games on
strings (and thus on the setting with DTDs). Furthermore, to achieve tractability, a special emphasis was
given to the restriction to \emph{bounded} strategies, in which  the
recursion depth with respect to web service calls is bounded by some
constant.

\paragraph*{Our approach}

The aim of this paper is to broaden the scope and extend the investigation of games for Active XML in several aspects. First of
all, we consider stronger schema languages (compared to DTDs) such as XML Schema and  Relax NG, due to their practical importance. To allow for this extension, our games are played on nested words \cite{AlurM09}.
\footnote{More precisely: word encodings of nested words in the sense of \cite{AlurM09}.}

Furthermore, we study the impact of the validation of input parameters for web service calls (partly
  considered already in \cite{MiloAABN05}), and
investigate an alternative semantics, where results of web service calls are
  inserted next to the node
  representing the web service, as opposed to replacing that node.

As we are particularly interested in the identification of tractable
cases, we follow the previous line of research by concentrating on
strategies in document order (left-to-right strategies) and by
considering bounded strategies (\emph{bounded replay}) and strategies in which no calls in results from previous
web service calls are allowed (\emph{no replay}). However, we also
pinpoint the complexity of the general setting. 

As a basic intuition for the concept of replay, consider again the online news site example from Figure \ref{fig:rewriting}, and assume that the schema for the event service's returns is (partially) given by $@\text{event\_svc} \rightarrow (\text{Sports} | \text{Movie}) @\text{event\_svc}$, i.e.~the event service allows for dynamic loading of additional results. A strategy with no replay would not be allowed to fetch any additional results in the situation of Figure \ref{subfig:after}, while a strategy with bounded replay $k$ (for some constant $k$) could load up to $k$ more events after the first. A strategy with unbounded replay would be able to fetch an arbitrary number of results, but might lead to a rewriting process that does not terminate if unsuccessful.

\paragraph*{Our contributions}

Our complexity results with respect to stronger schema languages are
summarised in \mbox{Table \ref{tab:results}}. In the general setting, the
complexity is very bad: doubly exponential time. However, there are
tractable cases for XML Schema: replay-free strategies in general and
strategies with bounded
replay in the case of finite replacement languages (that is, when there are
only finitely many possible answers, for each web service).
It should be noted that the \PSPACE-hardness result for the case
with DTDs, bounded replay and infinite replacement languages indicates that the
respective \PTIME claim in \cite{MiloAABN05} is wrong. 

\begin{table}[t]
\footnotesize
  \centering
          \begin{tabular}{|p{5cm}|c|c|c|}
            \hline
           & No replay & Bounded & Unbounded\\\hline
            \multicolumn{4}{|l|}{\textbf{Regular target language}}  \\\hline
Regular replacement & $\PSPACE$ & $\iiEXPTIME$& $\iiEXPTIME$\\\hline
Finite replacement & $\PSPACE$ & $\PSPACE$& $\EXPTIME$\\\hline
            \multicolumn{4}{|l|}{\textbf{DTD or XML Schema target language}} \\\hline
Regular replacement & $\PTIME$&  $\PSPACE$ & $\EXPTIME$\\\hline
Finite replacement & $\PTIME$&  $\PTIME$ & $\EXPTIME$\\\hline
          \end{tabular}
        \caption{Summary of complexity results. All results are completeness results.}
    \label{tab:results}
  \end{table}

In the setting where web services come with an input schema that
restricts the parameters of web service calls, we only study
replay-free strategies. It turns out that this case is tractable
if all schemas are specified by DTDs and the number of web services is
bounded. On the other hand, if the desired document structure is specified by an XML
Schema or the number of function symbols is unbounded, the task
becomes \PSPACE-hard.

For insertion-based semantics, we identify an undecidable
setting and establish a correspondence with the standard ``replacement''
semantics, otherwise. 

As a side result of independent interest, we show that the
word problem for alternating nested word automata is \PSPACE-complete.

\paragraph*{Related Work}
We note that the results on flat strings in this paper do not directly follow from the results in 
\cite{MuschollSS06}, as \cite{MuschollSS06} assumed target languages given by DFAs as opposed to \emph{deterministic} regular expressions, which are integral to both DTDs and more expressive XML schema languages. However, the techniques from \cite{MuschollSS06} can be adapted.

More related work for active context-free games than the papers
mentioned so far is discussed in
\cite{MuschollSS06}. Further results on active context-free games in
the ``flat strings'' setting can
be found in \cite{AbiteboulMB05,BjorklundSSK13}.
A different form of 2-player rewrite games are studied in
\cite{Waldmann02}.   More general \emph{structure rewriting games} are
defined in \cite{Kaiser09}.

\paragraph*{Organisation}

We give basic definitions in Section \ref{sec:prelim}. Games with
regular schema languages (given by nested word automata) are studied
in Section \ref{sec:general}, games in which the schemas are given as
DTDs or XML Schemas are investigated in Section
\ref{sec:simple}. Validation of parameters and insertion of web
service results are considered in Section \ref{sec:other}.
Most proofs are delegated to the appendix for brevity.

\paragraph*{Acknowledgements}
 We would like to thank the anonymous reviewers for their insightful and constructive comments. We are grateful to Nils Vortmeier and Thomas Zeume for careful proof reading, and to Krystian Kensy for checking our proof of Proposition \ref{prop:lowerdtd} (b) and for pinpointing the problems in the algorithm of \cite{MiloAABN05}  as part of his Master's thesis.


		\section{Preliminaries} \label{sec:prelim}

For any natural number $n \in \mathbb{N}$, we denote by $[n]$ the set $\{1, \ldots, n\}$.
Where $M$ is a (finite) set, $\Pot(M)$ denotes the powerset of $M$, i.e. the set of all subsets of $M$.
For an alphabet $\Sigma$, we denote the set of finite strings over $\Sigma$ by $\Sigma^*$ and  $\epsilon$ denotes the empty string.

\paragraph*{Nested words}
We use nested words\footnote{Our definition of nested words corresponds to word encodings of well-matched nested words in  \cite{AlurM09}.} as an abstraction of  XML documents \cite{AlurM09}. 
For a finite alphabet $\Sigma$, $\op{\Sigma} \mydef \{\op{a} \mid a \in \Sigma\}$ denotes the set of all \emph{opening $\Sigma$-tags} and $\cl{\Sigma} \mydef \{\cl{a} \mid a \in \Sigma\}$ the set of all \emph{closing $\Sigma$-tags}. The set $\wf \subs (\op{\Sigma} \cup \cl{\Sigma})^*$ of \emph{(well-)nested words over $\Sigma$} is the smallest set such that $\epsilon \in \wf$, and if $u,v \in \wf$ and $a \in \Sigma$, then also $u \op{a}v\cl{a} \in \wf$. We (informally) associate with every nested word $w$ its \emph{canonical forest representation}, such that words  $\op{a}\cl{a}$, $\op{a}v\cl{a}$ and $uv$ correspond to an $a$-labelled leaf, a tree with root $a$ (and subforest corresponding to $v$), and the forest of $u$ followed by the forest of $v$, respectively.  
A nested string $w$ is \emph{rooted}, if its corresponding forest is a tree. In a nested string $w = w_1 \ldots w_n \in \wf$, two tags $w_i \in \op{\Sigma}$ and $w_j \in \cl{\Sigma}$ with $i<j$ are \emph{associated} if the substring $w_i \ldots w_j$ of $w$ is  rooted.
To stress the distinction from nested strings in $\wf$, we refer to strings in $\Sigma^*$ as \emph{flat strings} (over $\Sigma$).

What we describe as opening and closing tags is often referred to as \emph{call symbols} and \emph{return symbols} in the literature on nested words; we avoid these terms to avoid confusion with \Read and \Call moves used in context-free games (see below).

\paragraph*{Context-free games}

A \emph{context-free game on nested words (cfG)} $G=(\Sigma,\Gamma,R,T)$ consists\footnote{Some of the following definitions are taken from \cite{BjorklundSSK13}.}
of a finite alphabet $\Sigma$, a set $\Gamma\subs\Sigma$ of \emph{function symbols}, a \emph{rule set} $R\subs\Gamma\times \wf$ and a \emph{target language} $T\subs\wf$. 
We will only consider the case where $T$ and, for each symbol $a\in \Gamma$, the set $R_a\mydef \{u\mid (a,u)\in R\}$ is a non-empty regular nested word language, to be defined in the next subsection.

A play of $G$ is played by two players, \pone\/ and \ptwo, on a word $w\in \wf$. In a nutshell, \pone moves the focus along $w$ in a left-to-right manner and decides, for every closing tag\footnote{It is easy to see that the winning chances of the game do not change if we allow \pone to play \Call moves at opening tags: if \pone wants to play \Call at an opening tag she can simply play \Read until the focus reaches the corresponding closing tag and play \Call then. On the other hand, if she can win a game by calling a closing tag, she can also win it by calling the corresponding opening tag, thanks to the fact that she has full information.} $\cl{a}$  whether she plays a \emph{Read} or, in case $a\in\Gamma$,  a \emph{Call} move. In the latter case, \ptwo then replaces the rooted word ending at the position of $\cl{a}$ with some word  $v\in R_a$ and the focus is set on the first symbol of $v$.
 In case of a \Read move (or an opening tag) the focus just moves further on. \pone wins a play if the word obtained at its end is in $T$.

Towards a  formal definition, a \emph{configuration} is a tuple $\kappa=(p,u,v)\in\{\ponea,\ptwoa\}\times (\op{\Sigma} \cup \cl{\Sigma})^*\times(\op{\Sigma} \cup \cl{\Sigma})^*$ where $p$ is the player to move, $uv \in \wf$ is the \emph{current word}, and the first symbol of $v$ is \emph{the current position}. A \emph{winning configuration} for \pone is a configuration $\kappa=(\ponea,u,\epsilon)$ with $u\in T$. 
The configuration $\kappa'=(p',u',v')$ is a \emph{successor configuration} of $\kappa=(p,u,v)$ (Notation: $\kappa\to \kappa'$) if one of the following holds:
 \begin{enumerate}[(1)]
 \item $p'=p=\ponea$, $u' = us$, and $sv' = v$ for some $s\in \op{\Sigma} \cup \cl{\Sigma}$ (\pone plays \Read);
  \item $p=\ponea$, $p'=\ptwoa$, $u = u'$, $v=v' = \cl{a} z$ for $z \in (\op{\Sigma} \cup \cl{\Sigma})^*$, $a\in\Gamma$, (\pone plays \Call);
  \item $p=\ptwoa$, $p'=\ponea$, $u=x \op{a} y$, $v = \cl{a} z$ for $x,z \in (\op{\Sigma} \cup \cl{\Sigma})^*$, $y \in \wf$, $u'=x$ and $v'=y'z$ for some $y' \in R_a$ (\ptwo plays $y'$);\footnote{We note that a \Call move on $\cl{a}$ in a substring of the form $\op{a}y\cl{a}$ actually deletes the substring $y$ along with the opening and closing $a$-tags. This is consistent with the AXML intuition of the subtree rooted at a function node getting replaced when the function node is called.}
 \end{enumerate}

The \emph{initial configuration} of game $G$ for string $w$ is $\kappa_0(w)\mydef (\ponea,\epsilon,w)$. 
A \emph{play} of $G$ is either an infinite sequence $\play=\kappa_0,\kappa_1,\ldots$ 
or a finite sequence $\play=\kappa_0,\kappa_1,\ldots,\kappa_k$ of
configurations, where, for each $i>0$, $\kappa_{i-1} \to \kappa_{i}$ and, in the finite case, $\kappa_k$ has no successor configuration. In the latter case, \pone
\emph{wins} the play if $\kappa_k$ is of the form $(\ponea,u,\epsilon)$ with $u\in T$, in all other cases, \ptwo wins.

\paragraph*{Strategies}
A \emph{strategy} for player $p\in\{\ponea,\ptwoa\}$ maps prefixes
$\kappa_0,\kappa_1,\ldots,\kappa_k$ of plays, where
 $\kappa_k$ is a $p$-configuration, to allowed moves.
We denote strategies for \pone by $\Astrat,\Astrat',\Astrat_1,\ldots$ and
strategies for \ptwo by $\Bstrat,\Bstrat',\Bstrat_1,\ldots$. 

A strategy
$\Astrat$ is \emph{memoryless} if, for every prefix
$\kappa_0,\kappa_1,\ldots,\kappa_k$ of a play, the selected move $\Astrat(\kappa_0,\kappa_1,\ldots,\kappa_k)$ only
depends on $\kappa_k$. As context-free games are reachability games we only need to consider memoryless games; see, e.g., \cite{GraedelTW02}.
\begin{proposition}
  Let $G$ be a context-free game, and $w$ a string. Then either \pone or \ptwo has a winning strategy on $w$, which is actually memoryless.
\end{proposition}

Therefore, in the following, strategies $\Astrat$ for \pone map configurations $\kappa$ to moves $\Astrat(\kappa)\in\{\Call,\Read\}$ and strategies $\Bstrat$ for \ptwo map configurations $\kappa$ to moves $\Bstrat(\kappa)\in \wf$. 

For configurations $\kappa,\kappa'$ and strategies $\Astrat,\Bstrat$  we write $\kappa\movest
\kappa'$ if $\kappa'$ is the unique successor configuration of $\kappa$ determined by strategies $\Astrat$ and $\Bstrat$. Given an initial word $w$ and
strategies  $\Astrat,\Bstrat$ the play\footnote{As the underlying game $G$
will always be clear from the context, our notation does not
mention $G$ explicitly.}
 $\play(\Astrat,\Bstrat,w)\mydef \kappa_0(w)\movest
\kappa_1\movest 
\cdots$ is uniquely determined. If $\play(\Astrat,\Bstrat,w)$ is finite, we denote the word represented by its final configuration by $\word(w,\Astrat,\Bstrat)$.

A strategy $\Astrat$ for \pone is \emph{finite} on string $w$ if the play $\play(\Astrat,\Bstrat,w)$ is 
finite for every strategy $\Bstrat$ of \ptwo. It is a \emph{winning strategy} on
$w$ if \pone wins the play $\play(\Astrat,\Bstrat,w)$, for every $\Bstrat$ of \ptwo. A strategy $\Bstrat$ 
for \ptwo is a \emph{winning strategy} for 
$w$ if \ptwo wins $\play(\Astrat,\Bstrat,w)$, for every strategy $\Astrat$ of
\pone. We only consider finite strategies for \pone, due to \pones  winning condition.
We denote the set of all finite strategies for \pone in the game $G$ by $\strata(G)$, and the set of all strategies for \ptwo by $\stratb(G)$.

The \emph{\Call depth} of a play $\play$ is the maximum nesting depth of \Call moves in $\play$, 
if this maximum exists. That is, the \Call depth of a play is zero, if no \Call is played at all, and one, if no \Call is played inside a string yielded by a replacement move.
For a  strategy $\Astrat$ of \pone and a string $w\in\wf$, the \emph{\Call depth} $\Depth(\Astrat,w)$ of $\Astrat$ on $w$ is the maximum \Call depth in any play $\play(\Astrat,\Bstrat,w)$. A strategy $\Astrat$ has \emph{$k$-bounded \Call depth} if $\Depth(\Astrat, w) \leq k$ for all $w \in \wf$. We denote by $\strata^k(G)$ the set of all strategies with $k$-bounded \Call depth for \pone on $G$. As a more intuitive formulation, we use the concept of \emph{replay}, which is defined as \Call depth (if it exists) minus one: Strategies for \pone  of \Call depth one are called \emph{replay-free}, and strategies of $k$-bounded \Call depth, for any $k$, have \emph{bounded replay}. For technical reasons, we need to use \Call depth for some formal proofs and definitions, but we will stick with the more intuitive concept of replay wherever possible.

By $\safelr(G)$ we denote the set of all words for which \pone has a
winning strategy in $\strata(G)$ (likewise for $\safelr[k](G)$ and $\strata^k(G)$).

\paragraph*{Nested word automata}
A \emph{nested word automaton (NWA)} $A = (Q, \Sigma, \delta, q_0, F)$ \cite{AlurM09} is basically a pushdown automaton which performs a push operation on every opening tag and a pop operation on every closing tag, and in which the pushdown symbols are just states.
More formally, $A$
consists of a set $Q$ of \emph{states}, an alphabet $\Sigma$, a \emph{transition function} $\delta$, an \emph{initial state} $q_0 \in Q$ and a set  $F \subs Q$ of \emph{accepting states}.
The function $\delta$ is the union of a function $(Q \times \op{\Sigma}) \rightarrow \Pot(Q \times Q)$ and a function $(Q \times Q \times \cl{\Sigma}) \rightarrow \Pot(Q)$.

A \emph{configuration} $\kappa$ of $A$ is a tuple $(q, \alpha) \in Q \times Q^*$, with a  \emph{linear state} $q$ and a sequence $\alpha$ of \emph{hierarchical states}, reflecting the pushdown store. A \emph{run of $A$ on $w=w_1 \ldots w_n \in \wf$} is a sequence $\kappa_0, \ldots, \kappa_n$ of configurations $\kappa_i = (q_i, \alpha_i)$ of $A$ such that for each $i \in [n]$ and $a\in\Sigma$ it holds that
\begin{itemize}
	\item if $w_i = \op{a}$, $ (q_i, p) \in \delta(q_{i-1}, \op{a})$ (for some $p \in Q$),  and $\alpha_i = p \alpha_{i-1}$, or
	\item if $w_i = \cl{a}$, $q_i \in \delta(q_{i-1}, p, \cl{a})$ (for some $p \in Q$), and $p \alpha_i = \alpha_{i-1}$. 
\end{itemize}
	In this case, we also write $\kappa_0 \reach[w]_{A} \kappa_n$.
We say that $A$ \emph{accepts} $w$ if $(q_0, \epsilon) \reach[w]_A (q', \epsilon)$ for some $q' \in F$. The language $L(A) \subs \wf$ is defined as the set of all strings accepted by $A$ and is called a \emph{regular language} (of nested words).

An NWA is \emph{deterministic} (or DNWA) if $|\delta(q, \op{a})| = 1 = |\delta(q,p,\cl{a})|$ for all $p,q \in Q$ and $a \in \Sigma$. In this case, we simply write $\delta(q, \op{a}) = (q', p')$ instead of $\delta(q, \op{a}) = \{(q', p')\}$ (and accordingly for $\delta(q,p,\cl{a})$), and $\delta^*(p, w) = q$ if $q$ is the unique state, for which $(p, \epsilon) \reach[w]_{A} (q, \epsilon)$.

An NWA is in \emph{normal form} if every transition function $\delta(p,\op{a})$ only uses pairs of the form $(q,p)$.
Informally, when $A$ reads an opening tag it  always pushes its current state (before the opening tag) and therefore can see this state when it reads the corresponding closing tag.
As in this case the hierarchical state is just the origin state $p$ of the transition, we write  $\delta(p,\op{a})=q$ as an abbreviation of  $\delta(p,\op{a})=(q,p)$, for DNWAs in normal form.

\begin{lemma}\label{lem:normal}
 There is a polynomial-time algorithm that computes for every deterministic NWA an equivalent deterministic NWA in normal form. \end{lemma}

\paragraph*{Algorithmic Problems}

In this paper, we study the following algorithmic problem $\Safelr(\calG)$ for various
classes $\calG$ of context-free games.

\begin{centering}
  \algproblem{$\Safelr(\calG)$}{A context-free game $G\in\calG$ and a string
    $w$.}{Is $w \in \safelr(G)$?}\\
\end{centering}

A class $\calG$ of context-free games in $\Safelr(\calG)$ comes with three parameters:
\begin{itemize}
\item the representation of the target language $T$,
\item the representation of the replacement languages $R_a$, and
\item to which extent replay is restricted.
\end{itemize}
It is a fair assumption that the representations of the target language and the replacement languages are of the same kind, but we will always discuss the impact of the replacement language representations separately. 
In our most general setting, investigated in Section \ref{sec:general}, target languages are represented by  deterministic nested word automata, and replacement languages by (not necessarily deterministic) nested word automata. We do not consider the representation of target languages by non-deterministic NWAs, as (1) already for DNWAs the complexity is very high in general, and (2) we can show that even in the replay-free case the complexity would become \EXPTIME-complete. We usually denote the automata representing the target and replacement languages by $A(T)$ and $A(R_a)$, respectively.
 
In Section \ref{sec:simple} we study the cases where $T$ is given as an XML Schema or a DTD. 
In each setting, we consider the cases of unrestricted replay, bounded replay (\Call depth $k$, for some $k$), and no  replay (\Call depth $1$). We note that replay depth is formally not an actual game parameter, but the algorithmic problem can be restricted to strategies of \pone of the stated kind.

If the class $\calG$  of games is clear from the context, we often simply write $\Safelr$ instead of $\Safelr(\calG)$.

We denote by $|R|$ the combined size of all $A(R_a)$, $a \in \Gamma$, and by $|G|$ the size of (a sensible representation of) $G$, i.e. $|G| = |\Sigma| + |R| + |A(T)|$.
		
		\section{Games with regular target languages} \label{sec:general}

We first consider our most general case, where target languages are
given by DNWAs, replacement languages by NWAs and replay is
unrestricted, because the algorithm that we develop for this case can
be adapted (and sped up) for many of the more restricted
cases. It is important to note that our results do not
  \emph{rely} on the presentation of schemas as nested word
  automata. In fact, in Section \ref{sec:simple}, we will assume that
  the target schema is given as an XML Schema or a DTD. However, for our algorithms
nested word automata are handy  to represent
(linearisations of) regular tree languages and therefore in \emph{this} section
target languages are represented by NWAs. We emphasize that deterministic bottom-up tree automata can be translated into deterministic NWAs in polynomial time \cite{AlurM09}.

This generic algorithm works in two main stages for a given cfG $G$ and
word $w$. It first analyses the game $G$ and aggregates all necessary
information in a so-called \emph{call effect} $C$. Then it uses $C$ to
decide whether \pone has a winning strategy in the game $G$ on $w$. 

The call effect $C$ only depends on $G$ and contains,
for every function symbol $f$ and every state $q$ of the $A(T)$, all possible effects of the subgame starting with a \Call move of \pone on some symbol $\cl{f}$ on the target language $T$, under the assumption that the sub-computation of $A(T)$ on the word yielded by the game from $\cl{f}$ starts in state $q$. More precisely, it summarises which sets $S$ of states \pone can enforce by some strategy $\sigma$,
where each $S$ is a set of states of $A(T)$ that \ptwo might enforce with a
counter strategy against $\sigma$. 

The first stage of the algorithm consists of an inductive computation in which
successive approximations $C^1,C^2,\ldots$ of $C$ are computed, where
$C^i$ is the restriction of $C$ to strategies of \pone of \Call depth
$i$. The size of call effects and the number of iterations are at most exponential in $|G|$. However, the
first stage can not be performed in exponential time as a single
iteration might take doubly exponential time in $|G|$. It turns out
through our corresponding lower bound that single iterations can not
be done faster.

At the end of the first stage, the algorithm computes an alternating
NWA $A_G$ (of exponential size) from $C$ that decides the set $\safelr(G)$.
In the second stage, $A_G$ is evaluated on $w$, taking at most polynomial
space in $|A_G|$ and $|w|$.

A restriction of games to bounded replay does not improve the general
complexity of the problem, as this is dominated by the doubly
exponential effort of a single iteration. However, for replay-free
games, no iterations are needed, the initial call effect $C^1$ is of
polynomial size and can easily be computed and therefore, in this
case, the overall complexity is dominated by the second stage,
yielding a polynomial-space algorithm. 

Altogether we prove the following theorem in this section.

\begin{theorem}\label{theo:general}\mbox{}
For the class of unrestricted games $\Safelr(\calG)$ is
  \begin{enumerate}[(a)]
  \item $\iiEXPTIME$-complete with
    unbounded replay,
  \item $\iiEXPTIME$-complete with 
    bounded replay, and
  \item $\PSPACE$-complete without replay.
  \end{enumerate}
\end{theorem}
The rest of this section gives a proof sketch for Theorem \ref{theo:general}.

Before we describe the generic algorithm in more detail, we discuss the
very natural and more direct approach by alternating
algorithms, in which a strategy for \pone is nondeterministically
guessed and the possible moves of \ptwo are taken care of by universal
branching. In our setting of context-free games, there are the following obstacles to this approach: (1) \ptwo
can, in general, choose from an infinite number of (and thus arbitrarily long) strings
in $R_a$, for the current $a$, and (2) it is not a priori clear that
such algorithms terminate on all branches.  Whereas the latter
obstacle is not too serious (if \pone has a winning strategy,
termination on all branches is guaranteed), the former
requires a more refined approach. We basically deal with it in two
ways: in some cases it is possible to show that it does not help \ptwo
to choose strings of length beyond some bound; in the remaining cases
(in particular in those cases considered in this section), the algorithms use abstracted moves instead of the actual replacement moves of the
game. The two stages that were sketched above, then come very
naturally: first, the abstraction has to be computed, then it can be
used for the actual alternating computation.

Our abstraction from actual cfGs is based on the simple observation that instead of knowing the final word $\word(w,\Astrat,\Bstrat)$ that is reached in a play $\play(\Astrat,\Bstrat,w)$, it suffices to  know whether $\delta^*(q_0,\word(w,\Astrat,\Bstrat))\in F$ to tell the winner. 
If we fix a strategy $\Astrat$ of \pone in a game on $w$, the possible
outcomes of the game (for the different strategies of \ptwo) can thus
be summarised by  $\states(q_0,w,\Astrat) \mydef
\{\delta^*(q_0,\word(w,\Astrat,\Bstrat))\mid \Bstrat \in
\stratb(G)\}$.

To this end, it will be particularly useful to study the (abstractions of) possible outcomes of
\emph{subgames} that start from a \Call move on some tag $\cl{a}$ until the
focus moves to the symbol after $\cl{a}$. 
\begin{definition}\label{def:calleffects}
For a cfG $G=(\Sigma,\Gamma,R,T)$ with a deterministic  target NWA $A(T) = (Q, \Sigma, \delta, q_0, F)$, the \emph{call effect} $\ceffect:\Gamma\times Q\to \Pot(\Pot(Q))$ is defined, for every $a\in\Gamma$, $q\in Q$, by
  \[
  \ceffect(a,q) \mydef \smin{\{\states(q, \op{a}\cl{a},\Astrat) \mid \Astrat
    \in \stratac(G)\}},
  \]
where $\stratac(G)$ contains all strategies of \pone that start
by playing \Read on $\op{a}$ and \Call on $\cl{a}$, and the operator $\smin{\cdot}$ removes all non-minimal sets from
  a set of sets.  
\end{definition}

We next describe how  to compute $\ceffect$ from a given cfG $G$.   As already mentioned, our algorithm follows a 
fixpoint-based approach. It computes inductively, for $k=1,2,\ldots$
the call effect of the restricted game of maximum \Call depth $k$. We
show that the fixpoint reached by this process is the actual call
effect $\ceffect$.

To this end, let, for every cfG $G$, $a\in\Sigma$, $q\in Q$,  and $k\ge 1$,
\[
\cneffect{k}(a,q) \mydef \smin{\{\states(q, \op{a}\cl{a},\Astrat) \mid \Astrat
    \in \stratac^k(G)\}}.
\]
As an important special case, the call effect of replay-free games ---  the basis for the inductive computation --- consists of only
one set.
\begin{lemma} \label{lemma:rfcalls} 
For every $q \in Q$ and $a\in \Sigma$, it holds that
\[
\cneffect{1}(a,q) =
  \{ \{\delta^*(q, v) \mid v \in R_a \} \}.
\]
 In particular, $\cneffect{1}$ can be computed from $G$ in polynomial time.
\end{lemma}
This just follows from the definitions, as \ptwo can choose
any string from $R_a$.

We next describe how each $\cneffect[G]{k+1}$ can be computed from $\cneffect[G]{k}$. The algorithm uses alternating nested word automata (ANWAs) which we will now define. 

An \emph{alternating nested word automaton (ANWA)} $A=(Q, \Sigma,
\delta, q_0, F)$ is defined like an NWA, except that the two parts of
$\delta$ map $(Q \times \op{\Sigma})$ into $\posbool(Q \times Q)$ and
$(Q \times Q \times \cl{\Sigma})$ into  $\posbool(Q)$, respectively, where $\posbool(Q)$  denotes the set of all positive
boolean combinations over elements of $Q$ using the binary operators
$\land$ and $\lor$ (and likewise for $\posbool(Q \times
Q)$).

The semantics of ANWA is defined via \emph{runs}, which require the notion of \emph{tree domains}. A tree domain is a prefix-closed language $D \subs \mathbb{N}^*$ of words over $\mathbb{N}$ such that, if $wk \in D$ for some $w \in D, k \in \mathbb{N}$, then also $wj \in D$ for all $j < k$. Strings in a tree domain are interpreted as node addresses for ordered trees in the standard way: $\epsilon$ addresses the root, and if $w \in D$ addresses some node $v$ with $k$ children, then $w1, \ldots, wk \in D$ address those children.

For any function $\lambda: D\to (Q \cup (Q\times Q))$ and node address $x \in D$, we denote by $\linstate{\lambda(x)}$ the linear state component of $\lambda(x)$, i.e. if $\lambda(x) = q$ or $\lambda(x) = (q,p)$ for some $p,q \in Q$, then $\linstate{\lambda(x)} = q$.

A \emph{run} $r=(D,\lambda)$ of an ANWA $A$ over a nested word $w =w_1 \ldots w_n$ is a finite tree of depth $n$, represented by a tree domain
$D$ and a labelling function $\lambda: D\to (Q \cup (Q\times Q))$ such that $\lambda(\epsilon)=q_0$ and, for every $x\in D$ of length $i$ with $\ell$ children, it holds that
\begin{itemize}
\item if $w_{i+1} \in \op{\Sigma}$, then $\{\lambda(x\cdot 1),\ldots,\lambda(x\cdot\ell)\} \models \delta(\linstate{\lambda(x)},w_{i+1})$, and
\item if $w_{i+1} \in \cl{\Sigma}$ with associated opening tag $w_j$, 
 and $\lambda(y) = (q,p)$ for some $p,q\in Q$ (where $y$ is the prefix of $x$ of length $j$), then $\{\lambda(x\cdot 1),\ldots,\lambda(x\cdot\ell)\} \models \delta(\linstate{\lambda(x)},p,w_{i+1})$.
\end{itemize}
An  ANWA $A$ \emph{accepts}  a nested word $w$ if there is a run
$(D,\lambda)$ over $w$ such that $\lambda(x)\in F$, for every $x\in D$
of length $|w|$.

ANWAs are used twice in the generic algorithm, first, to inductively compute
$\cneffect[G]{k+1}$ from $\cneffect[G]{k}$, second to actually decide  $\safelr(G)$, given $\ceffect$.
The following proposition will be crucial, in both cases.

\begin{proposition}\label{prop:calltoanwa}\mbox{ }
 There is an algorithm that computes from the call effect
    $\ceffect$ of a game $G$ in polynomial time in $|\ceffect|$ and $|G|$ an ANWA $A_{\ceffect}$ such that $L(A_{\ceffect})= \safelr(G)$.
 \end{proposition}
The computation of $\cneffect[G]{k+1}$ from $\cneffect[G]{k}$ involves
a non-emptiness test for ANWAs, the second stage a test whether $w\in
L(A_{\ceffect})$. Therefore, both of the following complexity results for ANWAs  influence the
complexity of our algorithms.

\begin{proposition}\label{prop:anwacomplexity}\mbox{}
  \begin{enumerate}[(a)]
  \item Non-emptiness for ANWAs is $\iiEXPTIME$-complete.
  \item The membership problem for ANWAs is $\PSPACE$-complete. 
  \end{enumerate}
\end{proposition}
Statement (a)  follows immediately from the corresponding result for
visibly pushdown automata in  \cite{Bozzelli07}, statement (b) is new, to
the best of our knowledge, and seems to be interesting in its own
right. It is shown in the appendix.

Now we continue describing the ingredients of the first stage of the generic algorithm.

\begin{lemma}\label{lemma:anwaeffects}
	Given a state $q \in Q$, an alphabet symbol $a \in \Gamma$, and $\cneffect[G]{k}$, for some $k\ge 1$, the call effect 
        $\cneffect[G]{k+1}(a,q)$ can be computed  in
        doubly exponential time in $|G|$.
\end{lemma}

By Lemmas \ref{lemma:rfcalls} and \ref{lemma:anwaeffects}, one can compute $\cneffect{k}$
inductively, for every $k\ge 1$. By definition it holds, for every $q$
and $a$, that $\cneffect{k}(a,q)$ is contained in the closure of
$\cneffect{k+1}(q,a)$ under supersets. As there are 
 $\le 2^{|Q|}$ sets in each $\cneffect{k}(a,q)$ (for $a \in \Gamma, q \in Q$),
the computation reaches a fixed point after at most
exponentially many iterations. We denote this fixed point
by $\cneffect{*}$, that is, we define, for every $a\in\Sigma$,
$q\in Q$:\\
\[
\cneffect{*}(a,q) \displaystyle\mydef\smin{\bigcup_{k=1}^\infty \cneffect{k}(a,q).}
\]

In particular, for each game $G$, there is a number $\ell\le
|\Gamma|\times|Q|\times 2^{|Q|}$ such that $\cneffect{*} =
\cneffect{\ell}$ and $\cneffect{m}=\cneffect{\ell}$, for every $m\ge \ell$.
However, it is not self evident that this process actually constructs
$\ceffect$, i.e., that  $\cneffect{*}=\ceffect$. The following
result shows that this is actually the case.

\begin{proposition}\label{prop:finitedepth}
For every cfG $G$ it holds:  \mbox{$\cneffect{*}=\ceffect$}.
\end{proposition}

Now we can give a (high-level) proof for Theorem \ref{theo:general}.

\begin{proofof}{Theorem \ref{theo:general}}
We first justify the upper bounds. Let $G$ be a cfG and $w$ a word. By Lemma \ref{lemma:rfcalls},
$\cneffect{1}$ can be computed in polynomial time from $G$. For the replay-free
case, we can immediately construct an ANWA for $\safelr(G)$ and
evaluate it on $w$, yielding a \PSPACE upper bound by Proposition
\ref{prop:anwacomplexity}.

For (a) and (b), $\ceffect$ ($\cneffect{k}$, respectively) can be computed in doubly
  exponential time, $A_C$ can be computed in
  exponential time (in the size of $G$), and  whether $w\in L(A_C)$ can then be tested in
  polynomial space in $|A_C|$ and $|w|$, that is, in at most
  exponential space in $|G|$ and $|w|$.

  That these upper bounds can not be considerably improved, is stated
  in the following proposition, thereby completing the proof of Theorem
  \ref{theo:general}.
\end{proofof}

\begin{proposition} \label{prop:lowergeneral}\mbox{}
For the class of unrestricted games $\Safelr$ is
  \begin{enumerate}[(a)]
  \item $\iiEXPTIME$-hard with
    bounded replay, and
  \item  $\PSPACE$-hard with
    no replay.
  \end{enumerate}
\end{proposition}

Claims (a) and (b) of Proposition \ref{prop:lowergeneral} follow from the corresponding parts of \mbox{Proposition \ref{prop:anwacomplexity}}; in the proof, we construct from an ANWA $A$ a replay-free cfG simulating $A$ on any input word $w$ (yielding claim (b)) and explain how replay can be added to that game to find and verify a witness for the non-emptiness of $A$, if one exists (yielding claim (a)).

For finite (and explicitly given) replacement languages the
complexity changes considerably in the cases with replay, but not in
the replay-free case.

\begin{proposition}\label{prop:finitegeneral}
For the class of unrestricted games with finite replacement languages,  $\Safelr(\calG)$ is
  \begin{enumerate}[(a)]
  \item $\EXPTIME$-complete with    unbounded replay, and
  \item $\PSPACE$-complete with  bounded or without replay.
  \end{enumerate}
\end{proposition}

 The upper bound in (a) follows as for finite
replacement languages  $\cneffect[G]{k+1}(a,q)$ can be computed from
$\cneffect[G]{k}(a,q)$ in polynomial space\footnote{It is worth noting
that this upper bound even holds if the finite replacement language is
not explicitly given, but represented by NWAs.}. 
The \PSPACE upper bound in (b) can then be achieved by the usual ``recomputation
technique'' of space-bounded computations.   

The lower bound in (a) already holds for flat words (see Theorem 4.3 in \cite{MuschollSS06}). The lower bound in (b) follows as the proof of Proposition \ref{prop:lowergeneral} only uses finite replacement languages. 

As our algorithms generally construct ANWAs deciding $\safelr(G)$, the data complexity for $\Safelr$ is in $\PSPACE$ for all cases considered in this section due to Proposition \ref{prop:anwacomplexity}.

		\section{Games with XML Schema target languages}\label{sec:simple}

The results of Section \ref{sec:general} provide a solid foundation for our further studies, but the setting studied there suffers from two problems: (1) the complexities are far too high (at least for games with replay) and (2) the assumption that target and replacement languages are specified by (D)NWAs is not very realistic. In this section, we address both issues at the same time: when we require that target languages are specified by typical XML schema languages (DTD or XML Schema), we get considerably better complexities. 

The better complexities basically all have the same reason: XML Schema
target languages can be described by a restriction of nested word
automata, which we call \emph{simple} below. This restriction
translates to the alternating NWAs corresponding to call effects. 
For simple ANWAs, however, the two basic algorithmic
problems, Non-emptiness and Membership have dramatically better
complexities: \PSPACE and \PTIME as opposed to $\iiEXPTIME$ and
$\PSPACE$, respectively. We emphasise that, in accordance with the
official standards, our definitions for DTDs and XML Schema require \emph{deterministic}
regular expressions.

Altogether, we prove the following complexity results.

\begin{theorem}\label{theo:simple}\mbox{}
For classes of games with XML Schemas or DTDs, respectively, $\Safelr$ is
  \begin{enumerate}[(a)]
  \item $\EXPTIME$-complete for 
    unbounded replay,
  \item $\PSPACE$-complete for 
    bounded replay, and
  \item $\PTIME$-complete (under logspace-reductions) without replay.
  \end{enumerate}
\end{theorem}
Here, the lower bounds are proven for DTDs, and the upper bounds for XML Schemas.

The lower bound in Theorem \ref{theo:simple} (b) for the case of games
with DTDs contradicts the statement of a \PTIME algorithm in Section 4.3 of
\cite{MiloAABN05}  (unless $\PTIME =\PSPACE$).\footnote{A close inspection of the 
  construction in the proof in
  \cite{MiloAABN05} reveals that the automaton constructed there does not deal correctly
  with the alternation between the choices of \ptwo and \pone. More
  precisely, the automaton allows \ptwo to let the suffix of a
  replacement string depend on the choices of \pone on its prefix.}

Before we describe the proof of  Theorem \ref{theo:simple}, we first define
\emph{single-type tree grammars} and \emph{local tree  grammars} as
well-established abstractions of XML Schema and DTDs,
respectively (see, e.g.,
\cite{MurataLMK05}). However, we will refer to  grammars of these types as XML
Schemas and DTDs, respectively.

\begin{definition}
	A \emph{(regular) tree grammar} is a tuple $T = (\Sigma, \Delta, S, P, \lambda)$, where
	\begin{itemize}
	\item $\Sigma$ is a finite alphabet of \emph{labels},
	\item $\Delta$ is a finite alphabet of \emph{types},
	\item $S \in \Delta$ is the \emph{root} or \emph{starting type},
	\item $P$ is a set of \emph{productions} of the form $X
          \rightarrow r_X$ mapping each type $X \in \Delta$ to a deterministic regular expression $r_X$ over $\Delta$, called the \emph{content model} of $X$, and
	\item $\lambda: \Delta \rightarrow \Sigma$ is a \emph{labelling function} assigning a label from $\Sigma$ to each type in $\Delta$.
	\end{itemize}
	$T$ is \emph{single-type} if for each $X \in \Delta$, the
        content model $r_X$ contains no competing types, i.e. if $r_X$
        contains no two types $Y \neq Z$ with $\lambda(Y) =
        \lambda(Z)$. $T$ is \emph{local}, if it has exactly one type
        for every label.
\end{definition}

We omit the definition of the formal semantics of regular tree grammars. The nested word language $L(T)$ described by $T$ is just the set of linearisations of trees of the tree language that is defined in the standard way.

We next define \emph{simple}  DNWAs, a restriction of DNWAs that captures all languages  specified by single-type tree
grammars. In simple DNWAs, states are typed, i.e. each state has a component in some type alphabet $\Delta$.
Informally, when a simple DNWA $A$ reads a subword $w=\op{a}v\cl{a}$ in state $q$, it determines already on reading $\op{a}$ which state $q'$ it will take after processing $w$, and this state will be of the same type as $q$. After reading $\op{a}$, the linear state of $A$ only depends on the \emph{type} of $q$, not the exact state; this models the single-type restriction. After reading $\op{a}$, $A$ goes on to validate $v$, and if this validation fails, $A$ enters a failure state $\bot$ instead of $q'$. Thus, the state of $A$ at a position basically only depends on its ancestor positions (in the tree view of the document) and their left siblings. The only way in which other nodes in subtrees of these nodes can influence the state is by assuming the sink state $\bot$. Thus, in the spirit of
\cite{MartensNSB06}, we could call such DNWAs \emph{ancestor-sibling-based} but we prefer the term \emph{simple} for simplicity.
 
\begin{definition}
	A deterministic NWA $A(T) = (Q, \Sigma, \delta, q_0, F)$ in
        normal form
 is      \emph{simple} (SNWA) if there exist a \emph{type alphabet} $\Delta$ and \emph{state set} $P$ with $Q \subs P \times \Delta$, a \emph{local acceptance function} $\Floc: \Sigma \rightarrow \Pot(Q)$, a \emph{target state function} $t: Q \times \Sigma \rightarrow Q$ and a \emph{failure state} $\bot \in Q \setminus F$, such that 
 the following conditions are satisfied for every $a\in\Sigma$:
  \begin{itemize}
  \item for every $p,p' \in P, X \in \Delta$: $\delta((p, X), \op{a}) = \delta((p',X),\op{a})$;
  \item for every $q \in \Floc(a)$: $\delta(q,p,\cl{a}) = t(p,a)$;
  \item for every $q \in Q\setminus \Floc(a)$: $\delta(q,p,\cl{a}) = \bot$ and
  \item for every $q \in Q$: $\delta(\bot, \op{a}) = \delta(\bot, q,\cl{a}) = \bot$.
  \item for every $(p, X) \in Q$: $t((p, X),a) = (p', X )$ for some $p' \in P$.
  \end{itemize}

A cfG is called \emph{simple} if its target DNWA is
simple. 
\end{definition}

\begin{proposition}\label{prop:dtdsimple}
From every  single-type tree grammar $T$, a simple DNWA $A$ can be
computed in polynomial time, such that $L(A)=L(T)$.
\end{proposition}
The following adaptation of the notion of \emph{simplicity} to ANWAs is a bit technical. It will guarantee however that the ANWAs obtained from simple games are simple and have reasonable complexity properties. 
\begin{definition}
	An ANWA $A = (Q, \Sigma, \delta, q_0, F)$ with $Q \subs P \times \Delta$ (for some state set $P$ and type alphabet $\Delta$)  is \emph{simple} (SANWA), if it has the following two properties.
\begin{itemize}
\item \emph{(Horizontal simplicity)}
        There are 
a \emph{local acceptance function} $\Floc: \Sigma
  \rightarrow \Pot(Q)$,
       a \emph{test state} $q_?\in Q$, and
a \emph{target state function} $t: Q \times \Sigma \rightarrow Q$,
	such that the transition function $\delta$ of $A$ satisfies the following conditions:
	\begin{itemize}
		\item $\delta(q,q',\cl{a}) = t(q',a)$ for all $q \in Q$
                  and $q'\not=q_?$;
		\item $\delta(q,q_?,\cl{a}) =
                  \begin{cases}
                    \text{true}, & \text{if $q\in \Floc(a)$}\\
                    \text{false}, & \text{if $q\not\in \Floc(a)$}\\
                  \end{cases}$
	\end{itemize}
	Furthermore, for each $(p, X) \in Q$ and $a \in \Sigma$, it holds that $t((p, X),a) = (p', X)$ for some $p' \in P$.
\item \emph{(Vertical Simplicity)} For each $X \in \Delta$ and $a \in \Sigma$, there is a $q \in Q$ such that for all $p \in P$ it holds that $\delta((p,X), \op{a}) \in \posbool(\{q\} \times ((P \times \{X\}) \cup \{q_?\}))$.
\end{itemize}
\end{definition}

Essentially,  horizontal simplicity states that $A$ has two kinds of computations on a well-nested subword: (1) computations starting from a pair $(q,q_?)$ test a property of the subword and can either succeed or fail at the end of the subword (and thus influence the overall computation); (2) computations starting from a pair $(q,q')$ for $q'\not=q_?$ basically ignore the subword. Even though they may branch in an alternating fashion, the state after the closing tag $\cl{a}$ is the same  in all subruns, is determined by $t(q',a)$ and has the same type as $q'$.
  
  Vertical simplicity, on the other hand, states that all alternation in $A$ happens in the choice of hierarchical states -- while, on an opening tag, $A$ may branch into sub-runs pushing different hierarchical states onto the stack, the choice of linear follow-up state is ``locally deterministic'', depending only the type of the previous state of $A$ and the label of the tag being read, and the current type is preserved in all hierarchical states except for $q_?$. Together, these two conditions also guarantee that SNWAs may also be interpreted as SANWAs.

\begin{proposition} \label{prop:lanwacomplexity}\mbox{}
  \begin{enumerate}[(a)]
  \item Non-emptiness for SANWA is $\PSPACE$-complete.
  \item The membership problem for SANWA is decidable in polynomial
    time.
  \end{enumerate}
\end{proposition}

\begin{proofof}{Theorem \ref{theo:simple}}
  The generic algorithm from the previous section can be adapted for simple cfGs, but with better complexity thanks to Proposition \ref{prop:lanwacomplexity}, to yield the upper bounds stated in \mbox{Theorem \ref{theo:simple}}.

More precisely, Proposition \ref{prop:lanwacomplexity} (b) and Lemma
\ref{lemma:rfcalls} yield a polynomial time bound for replay-free games.
Proposition \ref{prop:lanwacomplexity} (a) guarantees that the
inductive step in the computation of $\ceffect$ can be carried out in polynomial space
(as opposed to doubly exponential time).\footnote{We actually use a
  slightly stronger result than Proposition \ref{prop:lanwacomplexity}
  (a): deciding whether, for an NWA $A_1$ and
  a SANWA $A_2$, it holds $L(A_1)\cap L(A_2)\not=\emptyset$, is
  complete for $\PSPACE$.} The upper bounds for games with
unrestricted replay follows immediately and the upper bound for
bounded replay can be shown similarly as in
Proposition \ref{prop:finitegeneral} (b).

The lower bounds are given by the following proposition. They mostly follow from careful adaptation of lower bound proofs of \cite{MuschollSS06} for games on flat strings.
\end{proofof}

\begin{proposition} \label{prop:lowerdtd}\mbox{}
For the class of games with target languages specified by DTDs, \mbox{$\Safelr$ is}
  \begin{enumerate}[(a)]
  \item $\EXPTIME$-hard with
    unrestricted replay,
  \item  $\PSPACE$-hard with
    bounded replay, and
  \item $\PTIME$-hard  (under logspace-reductions) without replay
  \end{enumerate}
\end{proposition}

For finite (and explicitly given) replacement languages we get
feasibility even for bounded replay, but no improvement for unbounded replay.
\begin{proposition} \label{prop:simplefinite}\mbox{}
For the class of games with target languages specified by XML Schemas and
explicitly enumerated finite replacement languages, $\Safelr$ is
  \begin{enumerate}[(a)]
  \item $\EXPTIME$-complete with
    unrestricted replay, and
  \item $\PTIME$-complete  (under logspace-reductions)  with bounded replay
    or without  replay.
  \end{enumerate}
The same results hold for DTDs in place of XML Schemas.
\end{proposition}

Once again, as our algorithm generally computes a SANWA deciding $\safelr(G)$, the data complexity for $\Safelr$ is in $\PTIME$ for all cases considered here, due to  Proposition \ref{prop:lanwacomplexity}.
		
			\section{Validation of parameters and Insertion}\label{sec:other}

In this section, we focus on two features that have not been addressed in the previous two sections: validation of the parameters of a function call with respect to a given schema, and a semantics which allows that returned trees do not \emph{replace} their call nodes but are inserted next to them.

	\subsection{Validation of parameters}
As pointed out in \cite{MiloAABN05}, in Active XML, parameters of function calls should be valid with respect to some schema. Transferred to the setting of cfGs this means that \pone should only be able to play a \Call move in a configuration $(\ponea,u\op{a}v,\cl{a}w)$ if $\op{a}v\cl{a}$ is in $V_a$ for some set $V_a$ of words that are valid for calls of $\cl{a}$. Our definition of cfGs and the previous ones studied in the literature mostly ignore this aspect.\footnote{Actually, \cite{MiloAABN05} takes validation into account but the precise way in which parameters are specified and tested is not explained in full detail.} We do not investigate all possible game types in combination with parameter validation but rather concentrate on the most promising setting with respect to tractable algorithms.
It turns out, that games without replay and with DTDs to specify target, replacement and validation languages have a tractable winning problem as long as the number of different validation DTDs is bounded by some constant.\footnote{Note that this implies a polynomial-time data complexity for arbitrary replay-free games with DTD target, replacement and validation languages.} It becomes intractable if the number of validation schemas can be unbounded and (already) with target and validation languages specified by XML Schemas, even with only one validation schema.  
 
More precisely, we prove the following results.

\begin{theorem}\label{theo:validationptime}
For the class of games with validation with a bounded number of validation DTDs and target languages specified by DTDs, $\Safelr$  is in  $\PTIME$ without replay.
\end{theorem}

The algorithm uses a bottom-up approach. The basic idea is that, starting from the leaves, at each level of the tree (that is for some node $v$ and its leaf children) all relevant information about the game in the subtree $t_v$ is computed with the help of flat replay-free games and aggregated in $v$. Then the children of $v$ are discarded and the algorithm continues until only the root remains.

The following result shows that for slightly stronger games, parameter validation worsens the complexity.\footnote{This is, of course not surprising. If any, the surprising result is Theorem \ref{theo:validationptime}.}

\begin{theorem}\label{theo:validationlower}
  For the class of games with validation, $\Safelr$ (without replay) is
  \begin{enumerate}[(a)]
  \item $\EXPTIME$-hard, if target and validation languages are specified by DNWAs (even with only one function symbol);
  \item $\PSPACE$-hard, for games with only one function symbol, if the  validation language is given by an XML schema, the  target language by a DTD and a finite replacement language; and
  \item $\PSPACE$-hard, for games with an unbounded number of validation DTDs and replacement and target languages specified by DTDs.
  \end{enumerate}
\end{theorem}

Part (a) is proven by reduction from the intersection emptiness problem for DNWAs, while parts (b) and (c) use similar reductions from the problem of determining whether a quantified Boolean formula in disjunctive normal form is true.

Due to time constraints and as we are mainly interested in finding tractable cases, we have not looked for matching upper bounds.

\subsection{Insertion rules}
In our definition of \Call moves, we define the successor configuration of a configuration $(\ptwoa,u\op{a}v,\cl{a}w)$ to be $(\ponea,u,v'w)$, that is,  $\op{a}v\cl{a}$ is \emph{replaced} by a string $v'\in R_a$. However, Active XML also offers an ``append'' option, where results of function calls are inserted as siblings after the calling function node (cf. \cite{AbiteboulBM08}). There are (at least) three possible semantics of a \Call move for insertion (as opposed to replacement) based games: the next configuration could be (1) 
$(\ponea,u,\op{a}v\cl{a}v'w)$, (2) $(\ponea,u\op{a}v\cl{a},v'w)$, or (3) $(\ponea,u\op{a}v\cl{a}v',w)$, depending on ``how much replay'' we allow for \pone. We consider  (1) as the general setting, (2) as the setting with \emph{weak replay} and (3) as the setting \emph{without replay}. It turns out that the weak replay setting basically corresponds to the (unrestricted) setting with replacement rules and that (3) corresponds to the replay-free setting with replacement rules.
Setting (1), however, gives \pone a lot of power and makes $\Safelr(\calG)$ undecidable.

\begin{theorem}\label{theo:insertion}\mbox{ }
  For the class of games with insertion semantics, target DNWAs and replacement NWAs, $\Safelr$ is
  \begin{enumerate}[(a)]
  \item undecidable in general;
  \item $\iiEXPTIME$-complete for games with weak replay; and
  \item $PSPACE$-complete for games without replay.
 \end{enumerate}
\end{theorem}

The proof idea for Theorem \ref{theo:insertion} is to simulate insertion-based games by replacement-based games and vice versa; part (a) additionally uses the undecidability of $\Safelr$ for arbitrary (i.e. not necessarily left-to-right) strategies on games with flat strings, which was proven to be undecidable in \cite{MuschollSS06}.

		\section{Conclusion}\label{sec:conclusion}

The complexity of context-free games on nested words
differs considerably from that on flat words
(\iiEXPTIME vs.\ \EXPTIME), but there are still interesting tractable
cases. One of the main insights of this
paper is that the main tractable cases remain tractable if one allows
XML Schema instead of DTDs for the specification of schemas. 

Another result is that adding validation of input
parameters can worsen the complexity, but tractability can be
maintained by a careful choice of the setting. However, here the step
from DTDs to XML Schema may considerably worsen the complexity.

Insertion semantics with
unlimited replay yields undecidability. 

We leave open some corresponding upper bounds in the 
setting with validation of input parameters.  In future work, we plan to study the impact of parameters of
function calls more thoroughly. 
		
		 \markboth{Games for Active XML Revisited}{Martin Schuster and Thomas Schwentick}

\counterwithin{theorem}{section}

\renewcommand{\thesection}{\Alph{section}}
\setcounter{section}{0}
\newpage

		\section{Appendix}
For easier reference, we restate the results that were already stated
in the body of the paper. Definitions and results not stated in the body can be
identified by their number of the type A.xxx. At the end of the
appendix there is another bibliography which contains references for
all work mentioned in the appendix.

		\subsection*{Proofs for Section \ref{sec:prelim}}

\begin{restate}{Lemma \ref{lem:normal}}
 There is a polynomial-time algorithm that computes for every deterministic NWA an equivalent deterministic NWA in normal form. 
\end{restate}
\begin{proof}
Let $A = (Q, \Sigma, \delta, q_0, F)$ and let $\delta_1$ and
$\delta_2$ the projections of $\delta$ to its first and second
component (for opening tags only), respectively, i.e.,
$\delta(p,\op{a})=(\delta_1(p,\op{a}),\delta_2(p,\op{a}))$. An
equivalent DNWA $A' = (Q, \Sigma, \delta', q_0, F)$ in normal form can
be constructed by letting $\delta'(p,\op{a})\mydef
(\delta_1(p,\op{a}),p)$ and $\delta'(q,p,\cl{a})\mydef
\delta(q, \delta_2(p,\op{a}), \cl{a})$.
\end{proof}
\newpage		\section*{Proofs for Section \ref{sec:general}}

In this section, we give proofs for the upper and lower bounds on the complexity of $\Safelr$ for unrestricted games stated in Section \ref{sec:general}. 

\subsection*{Upper bounds for Theorem \ref{theo:general}}

The proof of the upper bounds in Theorem \ref{theo:general} consists technically of three main parts:
\begin{itemize}
\item the first part describes how to compute an ANWA for a cfG from its call effect (Proposition \ref{prop:calltoanwa}),
\item the second part establishes the complexity of emptiness and membership for ANWAS (Proposition \ref{prop:anwacomplexity}), and
\item the third part shows that the fix point process sketched after Lemma \ref{lemma:anwaeffects} in Section \ref{sec:general} indeed computes the call effect of a game. 
\end{itemize}

\subsubsection*{Transforming call effects into ANWAs}

The proof of Proposition \ref{prop:calltoanwa} requires a considerable amount of preparation.

As mentioned in Section \ref{sec:general}, our main tool for proving upper bounds on general cfGs is abstracting from subgames to the effects they induce on the target automaton $A(T)$. To facilitate the proof of Proposition \ref{prop:calltoanwa}, we extend this abstraction from the \emph{call effects} of subgames on rooted strings as defined in Section \ref{sec:general} to effects of arbitrary nested strings.
Formally, a \emph{(word) effect} maps states $q$ of $A(T)$ to sets of sets of states of $A(T)$. The effect of a game $G$ on  a word $w$ relative to state $q$ is basically the set of all state sets $X$, for which \pone has a strategy that guarantees that every play on $w$ yields some word $v$ with $\delta^*(q,v)\in X$. For ease of reference, we restate some definitions from Section \ref{sec:prelim} needed for word effects.

In the following, we  sometimes consider \emph{subgames} on a certain part of a string and talk about strategies for subgames. From a configuration $(u,vw)$, \pone can use a strategy $\Astrat$ on the subgame on $v$. This means that she follows $\Astrat$ until a configuration $(uv',w)$ is reached.

\begin{definition}
For a cfG $G=(\Sigma,\funcsymb, R,T)$ with a deterministic  target NWA $A(T) = (Q, \Sigma, \delta, q_0, F)$, we define the following notation.
  \begin{itemize}
  \item $\word(w, \Astrat, \Bstrat)$ denotes the unique final word that is reached in the game on $w$ with strategies $\Astrat \in \strata(G)$ and $\Bstrat \in \stratb(G)$.
  \item $\words(w,\Astrat) \mydef \{\word(w, \Astrat, \Bstrat) \mid \Bstrat \in \stratb\}$ denotes the set of final words that can be reached through strategies of \ptwo, for a fixed strategy $\Astrat \in \strata(G)$.
  \item $\states(q,w,\Astrat) \mydef \{\delta^*(q,v) \mid v \in \words(w,\Astrat)\}$ denotes the set of states that $A(T)$ can take at the end of final words that can be reached through strategies of \ptwo, for a fixed strategy $\Astrat \in \strata(G)$.
  \end{itemize}

  Finally, we define the \emph{word effect, $\geffect{w}:Q\to\Pot(\Pot(Q))$, of $G$ on $w$} by
  \[
  \geffect{w}(q) \mydef \smin{\{\states(q, w,\Astrat) \mid \Astrat
    \in \strata(G)\}},
  \]
for every $q\in Q$, where the operator $\smin{\cdot}$ removes all non-minimal sets from  a set of sets as before.
\end{definition}

To simplify notation, the subscript $G$ will often be omitted if the game $G$ is clear from the context.

The intuition behind word effects is the following abstraction of cfGs into single-round games: On an input string $w$, \pone first chooses a strategy $\Astrat$, then \ptwo chooses a strategy $\Bstrat$; the outcome of the game on $w$ is uniquely determined by $\Astrat$ and $\Bstrat$. In terms of effects, this corresponds to \pone picking a set $X = \states(q_0,w,\Astrat) \in \geffect{w}(q_0)$ and \ptwo then choosing a final state $q = \delta^*(q_0, \word(w, \Astrat, \Bstrat)) \in X$. This intuition also explains our use of the $\smin{\cdot}$ operator, as it makes no sense for \pone to offer \ptwo a choice from a set $X \subs Q$ if she can instead offer him the more limited options in some $X' \subsetneq X$.\footnote{Minimisation in our model corresponds to the monotonicity of powers \cite{Benthem03} or effectivity functions \cite{PaulyP03a}. Using, as we do, an inclusion-minimal ``basis" instead of a monotonic ``upward closure" allows for a more succinct representation and lower complexity in some places.}

It is easy to see that \pone has a winning strategy in $G$ on $w$ if and only if there is some $X \in \geffect{w}(q_0)$ such that $X \subseteq F$; to determine whether \pone has a winning strategy  it therefore suffices to compute $ \geffect{w}$.

It is natural to reason about effects for nested words in an inductive fashion. We first consider sequential composition. From \pones point of view, the game on a nested word $uv$ (with $u,v \in \wf$) from a state $q$ on proceeds as follows. \pone fixes a strategy $\Astrat$ on $u$.  The set of states that \ptwo can reach at the end of the subgame on $u$ is just $\states(q,u,\Astrat)$. For each state $p\in   \states(q,u,\Astrat)$, \pone can choose a strategy $\Astrat_p$ for $v$ and the result set is then the union of all sets that can be reached by \ptwo against any $\Astrat_p$ on $v$. To express the set of all combinations of outcomes for the second part, we use the following operator.

\begin{definition}
	Let $\calD=\{D_1,\dots,D_n\}$ be a set of sets of sets. Then $\Mix(\calD)$ is the set 
	$$\smin{\{d_1\cup \dots \cup d_n \;|\; d_1\in D_1 \land \dots \land d_n\in D_n\}}.$$
\end{definition}
 In other words, the \Mix operation yields every way of taking the union of one element from each of $D_1,\dots,D_n$ and then removes non-minimal sets.

	Let $E_1, E_2$ be mappings from $Q$ into $\Pot(\Pot(Q)$. Then the \emph{composition of $E_1$ and $E_2$} is defined as the mapping $E_1 \circ E_2: Q \rightarrow \Pot(\Pot(Q))$ with
	$$(E_1 \circ E_2)(q) \mydef \smin{\bigcup_{X \in E_1(q)} \Mix(\{ E_2(q') | q' \in X \})}.$$

Not surprisingly, effect composition commutes with word concatenation.

\begin{lemma} \label{lemma:seqcomposition}
For every cfG $G=(\Sigma,\funcsymb,R,T)$ and $u, v \in \wf$ it holds
\[
\geffect{uv} = \geffect{u} \circ \geffect{v}.
\]
\end{lemma}

Before proving Lemma \ref{lemma:seqcomposition}, we give an auxiliary result that will greatly simplify proofs about effects and similar functions. To that end, we call a set $\calD$ of sets \emph{normalised} if it contains no two sets $X,Y$ such that $X \subsetneq Y$ (or, equivalently, if $\calD=\smin{\calD}$). For two sets of sets $E_1, E_2$, we write $E_1 \covers E_2$ if and only if every $X \in E_1$ has a subset in $E_2$.

\begin{lemma} \label{lemma:normalised}
	Let $E_1, E_2$ be two normalised sets of sets. If $E_1 \covers E_2$ and $E_1 \covered E_2$, then $E_1 = E_2$.
\end{lemma}

\begin{proof}
	We prove only $E_1 \subs E_2$; inclusion in the other direction then follows by symmetry.	Let $X_1 \in E_1$, and let $X_2 \in E_2$ with $X_2 \subs X_1$. By assumption, there also exists $X'_1 \in E_1$ with $X'_1 \subs X_2$, and therefore $X'_1 \subs X_2 \subs X_1$. Since both $X_1$ and $X'_1$ are in $E_1$, and $E_1$ is normalised by assumption, this inclusion cannot be proper, and it follows that $X'_1 = X_2 = X_1$, and therefore $X_1 = X_2$ and $X_1 \in E_2$.
\end{proof}

\begin{proofof}{Lemma \ref{lemma:seqcomposition}}
Let $q \in Q$. This proof uses Lemma \ref{lemma:normalised} to prove the equality of the two normalised sets $\geffect{uv}(q)$ and $(\geffect{u} \circ \geffect{v})(q)$.

($\covers$):
Let $X \in \geffect{uv}(q)$. Then there exists some strategy $\Astrat_{uv} \in \stratlr(G)$ such that $X = \states(q, uv, \Astrat_{uv})$. Let $\Astrat_u$ be the restriction of $\Astrat_{uv}$ to the subgame on $u$, let $X_u \in \geffect{u}(q)$ with $X_u \subs \states(q, u, \Astrat_u)$ and $\{p_1, \ldots, p_k\} = X_u$. For each $i \in [k]$, let $\Astrat_v^i$ be a restriction of $\Astrat_{uv}$ to the subgame on $v$ in case \ptwo chooses a strategy $\Bstrat$ with $\state(q, u, \Astrat, \Bstrat) = p_i$, and let $X^i_v \in \geffect{v}(p_i)$ such that $X^i_v \subs \states(p_i, v, \Astrat_v^i)$ for all $i \in [k]$. Let $X'= X_v^1 \cup \ldots X_v^k$.

	By definition of $\circ$, and because of normalisation, there exists some $X'' \in (\geffect{u} \circ \geffect{v})(q)$ with $X'' \subs X'$. So, to show the desired inclusion, it suffices to prove that $X' \subs X$.
	
	Let $p' \in X'$. Then, $p' \in X_v^i$  and therefore $X' \in \states(p_i, v, \Astrat_v^i)$ for some $i \in [k]$. Also, $p_i \in X_u \subs \states(q,u,\Astrat_u)$, i.e. $p_i = \state(q, u, \Astrat_u, \Bstrat)$ for some $\Bstrat \in \stratb(G)$. By the definition of $\Astrat_u$ and $\Astrat_v^i$, this implies that $p' \in \states(q, uv, \Astrat_{uv}) = X$.
	
($\covered$):
	Let $X \in (\geffect{u} \circ \geffect{v})(q)$. By definition of $\circ$, there are sets $X_u = \{q_1, \ldots, q_k\} \in \geffect{u}(q)$ and $X^i_v \in \geffect{v}(q_i)$ for each $i \in [k]$ such that  $X = X_v^1 \cup \ldots X_v^k$. By definition of $\geffect{\cdot}$, there are strategies $\Astrat_u, \Astrat^1_v, \ldots, \Astrat_v^k \in \stratlr(G)$ with $X_u = \states(q, u, \Astrat_u)$ and $X^i_v=\states(q_i, v, \Astrat^i_v)$.
	
	Define a strategy $\Astrat_{uv}$ on $uv$ as follows. On $u$, \pone plays according to $\Astrat_u$; if this play yields some string $u_i \in \words(u, \Astrat_u)$ with $\delta^*(q, u_i) = q_i$, \pone then plays according to $\Astrat^i_v$ on $v$.
	
	Denote  $\states(q, uv, \Astrat_{uv})$ by $X'$ for short. Due to normalisation, there exists some $X'' \in \geffect{uv}$ with $X'' \subs X'$. What needs to be shown is therefore only that $X' \subs X$. 
	
	Let $q' \in X'$. Then there exists a strategy $\Bstrat$ for \ptwo and strings $u',v' \in \wf$ such that $u'v' = \word(uv, \Astrat_{uv}, \Bstrat)$, $u' = \word(u,\Astrat_{uv}, \Bstrat)$ and $\delta^*(q,u'v') = q'$. Let $q_u = \delta^*(q, u')$; then, it holds that $q' = \delta^*(q_u, v')$. By the definition of $\Astrat_{uv}$, it follows that $q_u \in X_u$ and $q' \in X^i_v$ for some $i$, so $q' \in X$, which concludes the proof.
\end{proofof}

It follows directly from Lemma \ref{lemma:seqcomposition} that the sequential composition of effects is associative.

The word effect of a word of the form $\op{a} v \cl{a}$ is induced by the word effect of $v$ and the possible moves of the players on $\op{a}$ and $\cl{a}$. In particular, as \pone may choose \Call on $\cl{a}$, the possible outcomes of a subgame on a subword of the form $\op{a} \cl{a}$ become crucial. As in the main part of this paper, we summarise the possible outcomes of subgames on ``two-letter words'' of the form $\op{a} \cl{a}$ by \emph{call effects} as defined in Definition \ref{def:calleffects}. For ease of reference, we restate that 
  \[
  \ceffect(a,q) \mydef \smin{\{\states(q, \op{a}\cl{a},\Astrat) \mid \Astrat
    \in \stratac(G)\}},
  \]
for every $a\in\Sigma$ and $q\in Q$, where $\stratac(G)$ contains all strategies of \pone that start by playing \Read on $\op{a}$ and \Call on $\cl{a}$.  

To describe hierarchical composition of word effects we define, for every $a\in\Sigma$  the following operator $H_a:Q\to\Pot(\Pot(Q))$
For every two functions $E: Q \rightarrow \Pot(\Pot(Q))$ and $C: \Sigma \times Q \rightarrow \Pot(\Pot(Q))$ and $q,q'\in Q$ such that $\delta(q,\op{a})=q'$, let
\[
H_a[E,C](q) \mydef \smin{\bigcup_{X\in E(q')} \Mix(\{C(a,q)\cup\{\delta(r,q,\cl{a})\}\mid r\in X\})}.
\]
Informally, interpreting $E$ as a word effect and $C$ as a call effect, the first set inside the $\Mix$ operator accounts for \Call moves and the second for \Read moves of \pone. Now we can formulate how effects behave hierarchically.

\begin{lemma} \label{lemma:hiercomposition}
For every cfG $G=(\Sigma,\funcsymb,R,T)$, $v \in \wf$, and $a\in\Sigma$, it holds
\[
\geffect{\op{a}v\cl{a}} = H_a[\geffect{v},\ceffect].
\]
\end{lemma}

\begin{proof}
We show, once again using Lemma \ref{lemma:normalised}, that for every
$q \in Q$, it holds that
\[
\geffect{\op{a}v\cl{a}}(q) = H_a[\geffect{v},\ceffect](q).
\]

($\covers$):
	Let $X \in \geffect{\op{a}v\cl{a}}(q)$. Then there is some strategy $\Astrat \in \stratlr(G)$ such that $X = \states(q,\op{a}v\cl{a}, \Astrat)$. Let $\Astrat_v$ be the sub-strategy of $\Astrat$ on $v$, let $\delta(q, \op{a}) = (q_a,p)$ and let $X_v = \{q_1, \ldots, q_k\} = \states(q_a,v,\Astrat_v)$. For each $i \in [k]$, let $v_i \in \words(v, \Astrat_v)$ such that $\delta^*(q_a, v_i) = q_i$ and let $\Astrat_i$ be the sub-strategy of $\Astrat$ starting at $(\op{a}v_i, \cl{a})$. If \pones move on $\cl{a}$ according to $\Astrat_i$ is \Read, let $X_i=\{\delta(q_i,p,\cl{a}\}$, otherwise let $X_i = \states(q,\op{a}\cl{a},\Astrat_i)$\footnote{As, in this case, $\Astrat_i \in \stratac$ is a strategy playing \Call on $\cl{a}$, we can omit $v$ here.}. Clearly, $X' = X_1 \cup \ldots \cup X_k$ has a subset in $\Mix(\{\ceffect(a,q)\cup\{\delta(r,p,\cl{a})\}\mid r\in X\})$ and therefore in $H_a[\effect{v},\ceffect](q)$. It remains to be proven that $X' \subs X$. 
	
	Let $q' \in X'$. Then $q' \in X_i$ for some $i \in [k]$. If $X_i = \{\delta(q_i,p,\cl{a}\}$, then clearly $q' = \delta^*(q, \op{a}v_i\cl{a}) \in \states(q, \op{a}v\cl{a}, \Astrat)=X$. Otherwise, $q' \in \states(q,\op{a}\cl{a},\Astrat_i)$ and $\Astrat_i$ coincides on $\op{a}\cl{a}$ with $\Astrat$ on $\op{a}v_i\cl{a}$, it follows again that $q' \in X$.
	
($\covered$):
	Let $X \in H_a[\geffect{v},\ceffect](q)$ and let $\delta(q,\op{a}) = (q',p)$. Then, there exists some $X_v = \{q_1, \ldots q_k\} \in \geffect{v}(q')$ such that $X=X_1 \cup \ldots \cup X_k$, where each $X_i$ is either in $\ceffect(a,q)$ or of the form $\{\delta(q_i,p,\cl{a})\}$. By the definition of $\geffect{v}$, there exists some strategy $\Astrat_v \in \stratlr$ on $v$ such that $\states(q', v, \Astrat_v) = X_v$, and by the definition of $\ceffect(a,q)$, for each $i$ with $X_i \in \ceffect(a,q)$ there exists a strategy $\Astrat_i \in \stratac$ such that $X_i = \states(q,\op{a}\cl{a},\Astrat_i)$. We extend $\Astrat_v$ to a strategy $\Astrat$ on $\op{a}v\cl{a}$ as follows: \pone reads the initial $\op{a}$, then plays on $v$ according to $\Astrat_v$. The string $v'$ resulting from this play on $v$ has to fulfil $\delta^*(q', v') = q_i$ for some $i \in [k]$; if, for this $i$ it holds that $X_i = \{\delta(q_i,p,\cl{a})\}$, then \pone plays \Read on $\cl{a}$, otherwise she plays \Call on $\cl{a}$ and plays according to $\Astrat_i$ in the resulting sub-game. Let $X' = \states(q, \op{a}v\cl{a}, \Astrat)$; it is easy to see
that $X' \subs X$, and since $X'$ has a subset in $\geffect{\op{a}v\cl{a}}(q)$ by definition, this proves the claim.
\end{proof}

We are now ready to define the ANWA $A_C$ from Proposition \ref{prop:calltoanwa}. The intuition behind it is that $A_C$ uses alternation to guess strategy choices for \pone and \ptwo in the above abstraction of $G$ on $w$ using call effects and tracks a current state $q$ in the target language DNWA $A(T)$. On opening tags, as well as on closing tags for which $A_C$ existentially guesses \pones move to be \Read, $A_C$ simply simulates $A(T)$; on closing tags $\cl{a}$ where $A_C$ decides for \pone to play \Call, $A_C$ then chooses existentially  a set $X \in \ceffect(a,q)$ (corresponding to a substrategy for \pone after the \Call on $\cl{a}$) and branches universally into all states $q' \in X$ (corresponding to \ptwos choice of a counter-strategy and a corresponding resulting state).

Formally, $A_C= (Q, \Sigma, \delta_C, q_0, F)$ is an ANWA in normal form , where $\delta_C$ is defined as follows. (Recall that $A(T) = (Q, \Sigma, \delta, q_0, F)$ is the target language DNWA in normal form.)
	\begin{itemize}
	\item For $a\in\Sigma, q\in Q$: 
\[
\delta_C(q,\op{a}) \mydef \delta(q,\op{a}).
\]
	\item For $a\in\Sigma, q,p\in Q$: 
\[
 \delta_C(q,p, \cl{a}) \mydef \delta(q,p,\cl{a}) \lor \bigvee_{X \in \ceffect(a,p)} \bigwedge_{r \in X} r.
\]
	\end{itemize}

We go on to prove the correctness of $A_C$. To that end, we call a run $\rho$ of an ANWA $A$ on a string $w$ \emph{minimal} if no proper subtree of $\rho$ is a run of $A$ on $w$ (i.e. if each set of states chosen to follow up some state on reading some symbol is inclusion-minimal among the sets of states fulfilling the corresponding transition formula).

\begin{lemma}\label{lemma:anwaconstruction}
	Let $q \in Q$, $w \in \wf$ and $X \subs Q$. Then, $X \in \geffect{w}(q)$ if and only if there is a minimal run of $A_C$ on $w$ starting at $q$ and ending in states from $X$.
\end{lemma}

\begin{proof}
	Let $q \in Q$, $X \subs Q$ and $w \in \wf$. The proof is by induction on the structure of $w$.

	For $w= \epsilon$, the claim is trivially fulfilled, as $\geffect{\epsilon}(q) = \{\{q\}\}$ by the definition of string effects.
	
	Let $w = uv$ for $u,v \in \wf$. For the ``only if" direction, it follows from Lemma \ref{lemma:seqcomposition} that there are sets $X_u = \{q_i, \ldots, q_k\} \in \geffect{u}(q)$ and $X_v^1, \ldots, X_v^k$ with $X_v^i \in \geffect{v}(q_i)$ for each $i \in [k]$ and $X = X_v^1 \cup \ldots \cup X_v^k$. By induction, there exist a minimal run $\rho_u$ of $A_C$ starting at $q$ and ending inside $X_u$ and for each $i \in [k]$ a minimal run $\rho_v^i$ on $v$ starting at $q_i$ and ending inside $X_v^i$. From these, we can construct a run $\rho$ of $A_C$ on $w$ by replacing each leaf labelled $q_i$ in $\rho_u$ with the entire run $\rho_v^i$ rooted at $q_i$. Obviously, $\rho$ is a run of $A_C$ starting at $q$ and ending inside $X$, and $\rho$ is minimal because $\rho_u$ and all $\rho_v^i$ are. The ``if" direction is proven analogously.
	
	Let $w = \op{a}v\cl{a}$ for $a \in \Sigma$, $w \in \wf$. Let further $\delta(q, \op{a}) = q'$. For ``only if", Lemma \ref{lemma:hiercomposition} implies that $X \in H_a[\geffect{v},\ceffect](q)$. This means that there is a set $X_v = \{q_1, \ldots, q_k\} \in \geffect{v}(q')$ and sets $X_w^1, \ldots, X_w^k$ such that $X = X_w^1 \cup \ldots \cup X_w^k$ and for each $i \in [k]$ either $X_w^i \in \ceffect(a,p)$ or $X_w^i = \{\delta(q_i, q,\cl{a})\}$. By induction, there exists a minimal run $\rho_v$ of $A_C$ on $v$ starting at $q'$ and ending inside $X_v$. We extend $\rho_v$ to a run $\rho$ on $w$ as follows:  The root of $\rho$ is labelled $q$ and has as its only child the root of a copy of $\rho_v$;  each leaf of this copy labelled $q_i$ has as its children exactly the states in $X_w^i$. Using the definition of $A_C$, it is easy to verify that $\rho$ is indeed a run of $A_C$ on $w$, and it is also clear that $\rho$ starts at $q$ and ends inside $X$. Finally, $\rho$ is minimal because its subrun on $v$ is minimal, and for each $q_i$, the set $X_w^i$ is an inclusion-minimal set fulfilling $\delta_C(q_i, q,\cl{a})$ (for $X_w^i \in \ceffect(a,q)$, this follows from $\ceffect(a,q)$ being normalised). Again, the ``if" part is proven analogously.
\end{proof}

Now we are in the position to prove Proposition \ref{prop:calltoanwa}:

\begin{restate}{Proposition \ref{prop:calltoanwa}}
 There is an algorithm that computes from the call effect
    $\ceffect$ of a game $G$ in polynomial time in $|\ceffect|$ and $|G|$ an ANWA $A_C$ such that $L(A_C)= \safelr(G)$.
\end{restate}

\begin{proof}
The statement follows from Lemma \ref{lemma:anwaconstruction}, as $A_C$ has an accepting run on any string $w \in \wf$ if and only if it has a minimal such run. Obviously, $A_C$ is of polynomial size in the size of $G$ and $\ceffect$ and can be constructed from these in polynomial time
\end{proof}

\subsubsection*{The complexity of ANWAs}

\begin{restate}{Proposition \ref{prop:anwacomplexity}}
  \begin{enumerate}[(a)]
  \item Non-emptiness for ANWAs is $\iiEXPTIME$-complete.
  \item The membership problem for ANWAs is $\PSPACE$-complete. 
  \end{enumerate}
\end{restate}
\begin{proof}
Statement (a) follows easily from \cite{Bozzelli07} where
$\iiEXPTIME$-completeness of Emptiness for alternating visibly
pushdown automata was
shown. The lower bound in that paper only requires finite well-nested
words.

Towards the upper bound in (b),  it is easy to see that an ANWA $\calA=(Q, \Sigma,
\delta, q_0, F)$ on some nested word $w$ can be simulated by an alternating Turing machine with polynomial time bound, hence the classical results from \cite{ChandraKS81} yield a polynomial space upper bound. 

For future reference we note that this computation can be actually be done in polynomial space in $|w|$ and the size $|Q|$ of $\calA$'s set of states, if it can be tested in polynomial space, whether
\begin{itemize}
\item for a given set $X\subs Q\times Q$ of pairs of states, a symbol
  $a$ and a state $q$, whether $X \models \delta(q,a)$, and
\item for a given set $X\subs Q$ of states, a symbol
  $a$ and states $p,q$, whether $X \models \delta(q,p,a)$.
\end{itemize}
The proof of this statement is along the same lines as the proof that alternating polynomial time is contained in polynomial space: The tree of all possible computations has polynomial depth and can be analysed with polynomial space.

The lower bound in (b) is shown by  a reduction from \QBF, that is, the problem to decide whether a quantified Boolean formula evaluates to true. We assume that the input formula for \QBF is of the form $\Phi = Q_1 x_1 \ldots Q_n x_n \varphi(x_1, \ldots x_n)$ with $Q_i \in \{\exists, \forall\}$ and a boolean formula $\varphi$ with $m$ clauses in conjunctive normal form.

The idea behind this reduction is to transform $\Phi$ into an ANWA $A$ and a  nested string $w$
such that $\Phi$ is true if and only $A$ accepts $w$. Actually, $w$ is of a very simple form: $\op{v_1}\cdots\op{v_n}\op{X}\cl{X}\cl{v_n}\cdots\cl{v_1}$.

If the automaton $A$ reads an opening tag $\op{v_i}$, it branches existentially, if $x_i$ is existentially quantified,
and it branches universally, if $x_i$ is universally quantified, thus choosing a truth assignment $\alpha$ for the variables. 
Finally, when it reads $\op{X}$, $A$ branches universally, picking one of the $m$ clauses of $\varphi$ in every branch. When it reads the suffix $\cl{X}\cl{v_n}\cdots\cl{v_1}$ of $w$, $A$ tests that $\alpha$ makes the chosen clause true.

To this end, the automaton $A$ uses three kinds of states: 
\begin{itemize}
	\item \emph{assignment states}, $q_+$ and $q_-$, corresponding to true and false, respectively,
	\item \emph{clause states} $q_j$, for $j\in [m]$, representing the clause chosen from $\varphi$ to be tested for truth and
	\item a starting state $q_0$ and an accepting state $q_F$.
\end{itemize}

For the formal construction, let $\Phi = Q_1 x_1 \ldots Q_n x_n \varphi$ be the input formula for \QBF with $Q_i \in \{\exists, \forall\}$ for all $i \in [n]$ and a quantifier-free boolean formula $\varphi = C_1 \land \ldots \land C_m$ with clauses $C_j$. 
Let $w$ be constructed as above.

The ANWA $A = (Q, \Sigma, \delta, q_0, \{q_F\})$ in normal form is defined as follows:
\begin{itemize}
	\item $Q = \{q_0, q_+, q_-, q_F\} \cup \{q_j \mid j \in [m]\}$;
	\item $\Sigma = \{v_i \mid i \in [n]\} \cup \{X\}$;
	\item For $q\in\{q_0, q_+,q_-\}$ and $i\in [n]$,\\
$\delta(q,\op{v_i}) = 
\begin{cases}
  q_+\lor q_- & \text{if $Q_i=\exists$,}\\
  q_+\land q_- & \text{if $Q_i=\forall$;}\\
\end{cases}
$
\item $\delta(q,\op{X}) = q_1\land\cdots\land q_m$;
\item For $q\in\{q_+,q_-\}$ and $j\in [m]$,\\
$\delta(q_j,q,\cl{X})=
\begin{cases}
  q_F & \text{
    \begin{minipage}{6cm}
      if $x_n$ occurs in $C_j$ and $q=q_+$ or $\neg x_n$ occurs in
      $C_j$ and $q=q_-$,
    \end{minipage}
}\\
  q_j & \text{otherwise.}
\end{cases}
$;
\item For $q\in\{q_+,q_-\}$, $j\in [m]$, and $2 \leq i \leq n$,\\
$\delta(q_j,q,\cl{v_i})=
\begin{cases}
  q_F & \text{
    \begin{minipage}{6cm}
      if $x_{i-1}$ occurs in $C_j$ and $q=q_+$ or $\neg x_{i-1}$ occurs in
      $C_j$ and $q=q_-$,
    \end{minipage}
}\\
  q_j & \text{otherwise.}
\end{cases}
$
\item For all $q \in Q$, $\delta(q, q_0, \cl{v_1}) = q$, and
\item For all $q \in Q$ and $i \in [n]$, $\delta(q_F, q, \cl{v_i}) = q_F$.
\end{itemize}

It remains to be shown that $\Phi$ evaluates to true if and only if $w \in L(A)$.

We first note that $A$ is deterministic on the suffix $\cl{X}\cl{v_n}\cdots\cl{v_1}$ of $w$.
It is not hard to show that, on this suffix, $A$ reaches the accepting
state from state $q_j$, if and only if, the truth assignment $\alpha$
induced by the choices on $\op{v_1}\cdots\op{v_n}$ makes $C_j$
true. Thus, the subrun on the suffix
$\op{X}\cl{X}\cl{v_n}\cdots\cl{v_1}$ of $w$ is accepting, if and only if, $\alpha$ makes all $m$ clauses true. 
Finally, the existential and universal branching of $A$ on
$\op{v_1}\cdots\op{v_n}$ corresponds to the quantification of the
variables of $\Phi$ in the obvious and correct way.
\end{proof}

\begin{restate}{Lemma \ref{lemma:anwaeffects}}
	Given a state $q \in Q$, an alphabet symbol $a \in \Gamma$, and $\cneffect[G]{k}$, for some $k\ge 1$, the call effect 
        $\cneffect[G]{k+1}(a,q)$ can be computed  in
        doubly exponential time in $|G|$.
\end{restate}

\begin{proof}
Let $a\in\Sigma$, $q\in Q$, $X\subs Q$, and $k\ge 0$. We show that,  given $\cneffect[G]{k}$ and a set $X \subs Q$,  it can be decided in
        doubly exponential time in $|Q|$ and polynomial
       time in $|R|$ whether a subset
        of $X$ is in $\in\cneffect[G]{k+1}(a,q)$.

Let $A_C$ be as defined for the proof of Lemma \ref{lemma:anwaconstruction} with $\cneffect[G]{k}$ as its basic Call effect, and let $A$ be its modification with initial state $q$ and set $X$ of accepting states.
$A$ accepts all nested strings $w$ on which there exists a strategy $\Astrat$ for \pone of Call depth at most $k$ such that $\states(q,w,\Astrat) \subs X$.
 	
	Let, for each $a\in \funcsymb$,  $A_a = (Q_a, \Sigma_a, \delta_a, q_{0,a}, F_a)$ be a NWA for $R_a$.
	
By definition,  $X$ has a subset in $\cneffect[G]{k+1}(a,q)$, if \pone has a strategy $\Astrat$ of call depth $k+1$ on $\op{a}\cl{a}$ that plays \Call on $\cl{a}$ and fulfils $\states(q,\op{a}\cl{a},\Astrat) \subs X$. 
Such a strategy $\Astrat$ for \pone exists if and only if \emph{for every word} $w\in R_a$ there is a strategy $\Astrat_w$ of \pone on $w$ with $\states(q,w,\Astrat_w) \subs X$, thus if and only if $R_a\subs L(A)$, equivalently $R_a\cap \widebar{L(A)}=\emptyset$ .

By using a standard product construction and a complementation of an ANWA, the test boils down to a non-emptiness test for an ANWA with a state set of polynomial size in $|G|$ and can thus be done in doubly exponential time thanks to Proposition \ref{prop:anwacomplexity}.\footnote{Note that the transition formulas of this ANWA may be of exponential size in $|G|$; however, the upper bound proof for the complexity of AVPA emptiness testing in \cite{Bozzelli07} still yields only a doubly exponential time complexity in $|G|$ here.}
\end{proof}

\subsubsection*{Adequacy of the fixed-point process}

The following lemma will be used in the proof of Proposition \ref{prop:finitedepth}.
\begin{lemma}\label{lemma:replacementeffects}
For a cfG $G=(\Sigma,\funcsymb,R,T)$ with a deterministic  target NWA $A(T) =
(Q, \Sigma, \delta, q_0, F)$, it holds, for every $a\in\Gamma$ and
$q\in Q$:
  \[
  \ceffect(a,q) =
    \Mix(\{\geffect{w}(q) \mid w\in R_a\}). 
  \]
\end{lemma}

\begin{proof}
	Since both sides of the claimed equation are minimal sets, it suffices by Lemma \ref{lemma:normalised} to show that each element of a set on one side of the equation has a subset on the other side.
  
  ($\covers$):
  Let $a \in \Gamma$, $q \in Q$ and $X \in \ceffect(a,q)$. By definition of $\ceffect$, there exists a strategy $\Astrat \in \stratac$ such that $X = \states(q,\op{a}\cl{a}, \Astrat)$. Again by definition, \pone plays \Call on $\cl{a}$ according to $\Astrat$.

 For every choice $w \in R_a$ with which \ptwo might respond to \pones initial \Call move on $\cl{a}$, there is a sub-strategy $\Astrat_w$ of $\Astrat$ on $w$. For each $w \in R_a$, let $X_w = \states(q,w,\Astrat)$. Obviously, each $X_w$ has a subset in $\geffect{w}(q)$, and therefore the set $X' = \bigcup_{w \in R_a}X_w$ has a subset in $\Mix(\{\geffect{w}(q) \mid w\in R_a\})$. It only remains to be proven that $X' \subs X$, so let $q' \in X'$. Then, by the definition of $X'$, there is some $w \in R_a$, strategy $\Bstrat_w \in \stratb$ and $w' \in \wf$ such that $w' = \word(w, \Astrat_w, \Bstrat_w)$ and $\delta^*(q,w')=q'$. From the way $\Astrat_w$ was defined from $\Astrat$, it follows that $w' \in \words(\op{a}\cl{a}, \Astrat)$, and therefore $q' \in X$.
  
  ($\covered$):
  Let  $a \in \funcsymb$, $q \in Q$ and $X \in \Mix(\{\geffect{w}(q) \mid w\in R_a\})$. Then, for each $w \in R_a$ there exists some set $X_w \in \geffect{w}(q)$ such that $X = \bigcup_{w \in R_a} X_w$. By the definition of $\geffect$, this means that for every $w \in R_a$ there is some strategy $\Astrat_w\in \stratlr$ such that $\states(q,w,\Astrat_w) = X_w$. Let $\Astrat \in \stratac$ be the strategy on $\op{a}\cl{a}$ where \pone plays \Call on $\cl{a}$ and then, if \ptwo picks $w \in R_a$ as a replacement, keeps playing according to $\Astrat_w$ on $w$. By definition of $\ceffect$, the set $X'=\states(q,\op{a}\cl{a},\Astrat)$ has a subset in $\ceffect(q,a)$, and it only remains to be proven that $X' \subs X$. Let therefore $q' \in X'$ Then, there is some strategy $\Bstrat \in \stratb$ and string $w' \in \wf$ such that $w'=\word(\op{a}\cl{a}\Astrat, \Bstrat)$. Since \pones first move according to $\Astrat$ is a \Call on $\cl{a}$, there is some string $w \in R_a$ which \ptwo chooses as a replacement according to $\Bstrat$; by definition of $\Astrat$, it then holds that $q' \in \states(q,w,\Astrat_w) = X_w \subs X$ as was to be proven.
\end{proof}

For the following proof, the \emph{width} of a nested word is the maximum number of children of any node in its corresponding forest. Its \emph{root width} is just the number of trees in its forest.  The \emph{(nesting) depth} of a nested word is the depth of its canonical forest representation.

\begin{restate}{Proposition \ref{prop:finitedepth}}
For every cfG $G$ it holds:  \mbox{$\cneffect{*}=\ceffect$}.
\end{restate}

\begin{proof}
For the proof we construct from a  cfG $G=(\Sigma,\Gamma,R,T)$ a game $G'=(\Sigma,\funcsymb,R',T)$, where $R'$ consists of particular finite sublanguages $R'_a\subs R_a$, for every $a \in \funcsymb$. Then we show
\begin{enumerate}[(a)]
\item \mbox{$\cneffect[G]{*}=\cneffect[G']{*}$},
\item \mbox{$\cneffect[G']{*}=\ceffect[G']$}, and finally
\item \mbox{$\ceffect[G']=\ceffect$}.
\end{enumerate}

To construct $G'$, we first examine the algorithm from the proof of Lemma \ref{lemma:anwaeffects} more closely. For a given state $q \in Q$, alphabet symbol $a \in \Sigma$, state set $X \subs Q$ and effect $\cneffect[G]{k}$, the output of that algorithm depends only on the existence of a single string from $R_a$ -- 
for $a \in \funcsymb$, the algorithm rejects if and only if there is a string in $R_a$ that is \emph{not} accepted by $A$. For each $q \in Q$, $a \in \Sigma$, $X \subs Q$ and $k \geq 1$, let $w(q,a,k,X)$ be one such \emph{witness string} of minimum length, if such a string exists. Obviously, the output of the algorithm from Lemma \ref{lemma:anwaeffects} for input $q,a,X$ and $\cneffect[G]{k}$ does not change if we replace $R_a$ by any subset of $R_a$ containing $w(q,a,k,X)$.

Let $k^*$ be the smallest number with $\cneffect{*} = \cneffect{k^*}$, and let $W_a$ be the set containing all $w(q,a,k,X)$ for all $q \in Q$, $k \leq k^*$ and $X \subs Q$. Furthermore, for each $w \in \wf$, let $v(a,w)$ be a string of minimum length such that $\geffect{w} = \geffect{v(a,w)}$ and $v(a,w) \in R_a$ and let $V_a = \{v(a,w) \mid w \in \wf\}$. Since there are only finitely many different string effects, each set $V_a$ for $a \in \Sigma$ must be finite as well.

The replacement rules $R'$ for $G'$ are now constructed as follows: 
For  each $a \in \funcsymb$, let $R'_a \mydef W_a \cup V_a$. By construction, it holds that $R'_a$ is a finite subset of $R_a$, and an easy induction argument (along with the above considerations) shows that $\cneffect[G']{k} = \cneffect{k}$ for each $k \geq 1$. Along with the definition of $\cneffect[\cdot]{*}$, this proves (a).

For (b) it is sufficient to show that each finite strategy $\Astrat\in\strata[G']$ on a word $w$ has bounded \Call depth. This can be easily established with the help of K\H onig's Lemma. To this end, we consider the strategy tree $T_{\Astrat,w}$ for $\Astrat$ on $w$ where each node is a game position of the form $(p,u,v)$ with a player index $p \in \{\ponea,\ptwoa\}$ and strings $u,v \in (\op{\Sigma} \cup \cl{\Sigma})^*$ and each node corresponding to a game position $\kappa$ has as children the possible follow-up positions $\kappa'$ such that $\kappa \movest \kappa'$ for $\Astrat$ and some counter-strategy $\Bstrat \in \stratb$. Each node of this tree has a finite number of children -- nodes corresponding to positions belonging to \pone have only a single child each (as $\Astrat$ is fixed), and positions in which \ptwo is to replace some $a \in \Sigma$ have one child for each string in $R'_a$.  Thus, the \Call depth of nodes is bounded, as otherwise $T_{\Astrat,w}$ would be a finitely branching tree with branches of arbitrary length, which by K\H onig's Lemma would yield that $T$ has an infinite branch, contradicting the finiteness assumption for $\Astrat$.

Towards (c), we prove the slightly stronger claim that $\geffect{w}(q) = \effect{G',w}(q)$ for all $q \in Q$ and $w \in \wf$. Lemma \ref{lemma:replacementeffects} then implies (c). To this end, we prove that each set in $\effect{G',w}(q)$ has a subset in $\geffect{w}(q)$ and vice versa, which proves the desired equality by Lemma \ref{lemma:normalised}.

One of these directions is almost trivial, as \ptwo simply has no more possible moves in $G'$ than in $G$. Thus, any strategy $\Astrat \in \stratac(G)$ induces a sub-strategy $\Astrat' \in \stratac(G')$ with $\words[G'](w, \Astrat') \subs \words[G](w, \Astrat)$ and therefore also $\states[G'](q, w, \Astrat') \subs \states[G](q, w, \Astrat)$.

For the other direction, let $q \in Q$, $w \in \wf$ and let $\Astrat' \in \stratac(G')$ with $X = \states[G'](q,w,\Astrat') \in \effect{G',w}(q)$. Let $d\mydef \Depth[G'](\Astrat',w)$. This is well-defined as  $\Astrat'$ is finite.
We prove by nested induction over $(d,\text{nesting depth of }w,\text{root width of }w)$ that there exists a strategy $\Astrat$ in $G$ with $\states[G](q,w,\Astrat) \subs X$, which implies that $X$ has a subset in $\geffect{w}(q)$. 

If $d = 0$, then \pone only plays \Read on the entirety of $w$; obviously, this strategy is feasible in $G$ as well and yields the same result. 

If $d>0$, \pone must play \Call on $w$ at some point, and therefore it holds that $w \neq \epsilon$. 

If $w = uv$ for $u,v \in \wf$, let $\Astrat'_u$ be the sub-strategy of $\Astrat'$ on $u$, and let $\{q_1, \ldots, q_k\} = \states[G'](q,u,\Astrat'_u)$. For each $i \in [k]$, let further $\Astrat'_{v,i}$ be a sub-strategy of $\Astrat'$ on $v$ in case the play on $u$ yields some string $u'$ with $\delta^*(q, u')=q_i$. By induction (as $u$ and $v$ have smaller root width than $w$), there exist strategies $\Astrat_u$ on $u$ and $\Astrat_{v,i}$ on $v$ in $G$ such that $\states[G](q,u,\Astrat_u) \subs \{q_1, \ldots, q_k\}$ and $\states[G](q_i ,v,\Astrat_{v,i}) \subs \states[G'](q_i ,v,\Astrat'_{v,i})$. Let $\Astrat$ be the strategy on $uv$ in $G$  where \pone plays according to $\Astrat_u$ on $u$ and according to $\Astrat_{v,i}$ if the play on $u$ yielded a string $u'$ with $\delta^*(q, u')=q_i$. Then, it holds that $\states[G](q, w, \Astrat) \subs \bigcup_{i \in [k]} \states[G'](q_i, v, \Astrat'_{v,i}) \subs X$.

If $w=\op{a}v\cl{a}$ for some $a \in \funcsymb$, $v \in \wf$, let $\delta(q,\op{a}) = (q',p)$, let $\Astrat'_v$ be the sub-strategy of $\Astrat'$ on $v$, and let $\{q_1, \ldots, q_k\} = \states[G'](q', v, \Astrat'_v)$. By induction (as the depth of $v$ is smaller than the depth of $w$), there exists a strategy $\Astrat_v$ on $v$ in $G$ with $\states[G](q', v, \Astrat_v) \subs \states[G'](q', v, \Astrat'_v)$. In the strategy $\Astrat$ on $w$, \pone plays according to $\Astrat_v$ on $v$. The play on $v$ from $q'$ according to $\Astrat$ is bound to reach some state $q_i$ for $i \in [k]$. If there is some string $v_i \in \words[G'](v, \Astrat'_v)$ with $\delta^*(q',v_i) = q_i$ such that \pone would play \Read on $\cl{a}$ according to $\Astrat'$ in $G'$ if the play on $v$ yields $v_i$, then \pone also plays \Read on $\cl{a}$ according to $\Astrat$; obviously, in this case, the resulting state from the play according to $\Astrat$ is in $X$. Otherwise, \pone plays $\Call$ on $\cl{a}$ in $\Astrat$. Let $z \in R_a$ be some arbitrary response for \ptwo to this \Call move in $G$; we now explain how \pone plays on $z$ according to $\Astrat$. 

By construction, the replacement language $R'_a$ in $G'$ contains the string $v(a,z)$, so this string is a valid response for \ptwo to the \Call by \pone on $\cl{a}$ in $G'$. Let $\Astrat'_{v(a,z)}$ be the sub-strategy of $\Astrat'$ if \ptwo chooses this response. As $\Astrat'_{v(a,z)}$ has a \Call depth of at most $d-1$, by induction there exists a strategy $\Astrat_{v(a,z)}$ for \pone on $v(a,z)$ in $G$ with $\states[G](q, v(a,z), \Astrat_{v(a,z)}) \subs \states[G'](q, v(a,z), \Astrat'_{v(a,z)})$. By the definition of $v(a,z)$, it holds that $\geffect{z} = \geffect{v(a,z)}$, which implies that there is a strategy $\Astrat_z$ for \pone on $z$ in $G$ such that $\states[G](q, z, \Astrat_{z}) \subs \states[G](q, v(a,z), \Astrat_{v(a,z)})$. In $\Astrat$, \pone then plays on $z$ according to $\Astrat_z$, and the above set inclusions show that all states resulting from this play are in $X$ as well, which completes the case $w=\op{a}v\cl{a}$ for $a \in \funcsymb$ and concludes the proof.
\end{proof}

\subsection*{Lower bounds}

Similar to Lemma \ref{lemma:anwaconstruction}, where we constructed ANWA from given cfGs to obtain upper complexity bounds, we prove matching lower bounds for by transforming ANWA into cfGs. 

\begin{lemma}\label{lemma:ANWAtocfG}
	There is a polynomial time algorithm that computes, given an ANWA $A$ and a nested word $w$, a  cfG  $G = (\Sigma, \emptyset, \funcsymb, R, T)$ and a nested word $w'$ such that $w\in L(A)$ if and only if \ptwo has a winning strategy on $w'$ in $G$, against all replay-free strategies of \pone. Furthermore, $G$ only depends on $A$ (not on $w$) and can be computed in polynomial time in the size of $A$.
\end{lemma}

\begin{proof}
 Let $A = (Q, \Sigma, q_0, \delta, \{q_F\})$ be an ANWA and $w\in \wf$ a nested word. The idea is to simulate the alternation of $A$ in the game $G$ on $w'$. We design $G$ to only admit replay-free strategies for \pone. To make this possible, we construct $w'$ from $w$ by adding substrings that offer enough ``space'' for this simulation. 

We assume without loss of generality that $\Sigma$ does not contain any symbols from $(Q\times Q) \cup Q \cup \{b, 0,1,{\lor},{\land},\bot,\top\}$.
For any formula $\varphi \in \posbool(Q\times Q) \cup \posbool(Q)$, the \emph{encoding} $\Enc(\varphi)$ is the well-nested string over the alphabet $Q \cup\{\vee, \land\}$ derived from $\varphi$ in the natural way:
\begin{itemize}
\item If $\varphi \in \{\bot,\top\}$, then $\Enc(\varphi) = \op{\varphi}\cl{\varphi}$;
	\item If $\varphi = (q,p) \in Q\times Q$, then $\Enc(\varphi) = \op{(q,p)}\cl{(q,p)}$;
	\item If $\varphi = q \in Q$, then $\Enc(\varphi) = \op{q}\cl{q}$;
	\item If $\varphi = \varphi_1 \lor \varphi_2$, then $\Enc(\varphi) = \op{\lor}\cl{\lor}\op{b}\Enc(\varphi_1) \Enc(\varphi_2)\cl{b}$;	
	\item If $\varphi = \varphi_1 \land \varphi_2$, then $\Enc(\varphi) = \op{\land}\cl{\land}\op{b}\Enc(\varphi_1) \Enc(\varphi_2)\cl{b}$.
\end{itemize}

Let $q_1,\ldots,q_m$ be an enumeration of the states in $Q$. 

Let $\Sigma'$ and $\Sigma''$ be two distinct copies of $\Sigma$ with symbols of the form $a'$ and $a''$, respectively, for every $a\in\Sigma$.

For each $a\in\Sigma $, we define
\begin{itemize}
\item $v(\op{a})\mydef\op{a'}\cl{a'}\Enc(\delta(q_1,\op{a}))\cdots
  \Enc(\delta(q_m,\op{a}))\op{a}$, and
\item   $v(\cl{a})=\op{a''}\cl{a''}\Enc(\delta(q_1,q_1,\cl{a}))\cdots 
  \Enc(\delta(q_m,q_m,\cl{a}))\cl{a}$.
\end{itemize}
We note that in $v(\op{a})$, for each $i\le m$, there is a subword $\Enc(\delta(q_1,\op{a}))$, whereas in $v(\cl{a})$ there is a subword $\Enc(\delta(q_i,q_j,\cl{a}))$, for every $i,j\le m$.  
The string $w'$ is defined as the nested word $v(w)$, that results from $w$ by replacing every tag $\sigma \in \op{\Sigma} \cup \cl{\Sigma}$ with $v(\sigma)$.

As explained above, the purpose of the game $G$ is to simulate the alternation of $A$. We associate existential branching with \ptwo and universal branching with \pone.\footnote{The reader might feel that it would be more natural to associate existential moves to \pone. Why our chosen association is useful will become clear in the proof of Proposition \ref{prop:lowergeneral} given below.} To this end, the replacement languages for $\op{\lor}$ and $\op{\land}$ are as follows.
\begin{itemize}
\item $R_\lor =  \{\op{1}\cl{1},\op{2}\cl{2}\}$;
\item $R_\land =  \{\op{2}\cl{2}\}$;
\end{itemize}
All other symbols should not be replaced in the game, so we set $\funcsymb = \{\lor, \land\}$.

The intention of the construction is that the behaviour of $A$ on $w$ is simulated as follows in the game on $v(w)$ in $G$. Choices corresponding to $\lor$-gates in transitions are taken by \ptwo (and we force \pone to call every symbol $\lor$ as strings containing $\lor$-tags will not be accepted by the target NWA). The choice of $\op{1}\cl{1}$ by \ptwo is interpreted by the choice of the first branch of the formula by $A$ and likewise for $\op{2}\cl{2}$ and the second branch.
 Choices corresponding to $\land$-gates  in transitions are taken by \pone: we interpret $\op{\land}\cl{\land}$ just as $\op{1}\cl{1}$ in the $\lor$-case. Therefore, if \pone reads$\op{\land}\cl{\land}$ this corresponds to choosing the first branch of the formula, if she calls it, she chooses the second branch.

The target automaton follows the choices taken by the two players. At opening tags of the form $\op{(q,p)}$ it interprets $q$ and $p$ as the next horizontal and hierarchical state, respectively. It accepts if it ends in an accepting state or reaches $\op{\top}\cl{\top}$  at some point. If it reaches $\op{\bot}\cl{\bot}$  at some point, it rejects.

It remains to show that indeed $w$ is accepted by $A$ if and only if \ptwo has a winning strategy on $v(w)$ in $G$.

We call a strategy for \pone on $v(w)$ \emph{valid} if \pone plays \Call on every $\lor$ symbol. Since \pone can never win with a strategy that is not valid, we restrict our attention to valid strategies for \pone on $v(w)$.

We will now show that each run of $A$ on $w$ corresponds to some strategy $\Bstrat$ of \ptwo on $v(w)$ in $G$, and that an accepting run on $w$ induces a winning strategy on $v(w)$ and vice versa.

Let $\Bstrat$ be a strategy for \ptwo on $v(w)$, and let $\sigma$ be some tag in $w$. We say that a subformula $\varphi'$ encoded in $v(\sigma)$ is \emph{enabled} according to $\Bstrat$ and some counterstrategy for \pone if the resulting sub-play on $v(\sigma)$ yields a substring of the form $\op{1}\cl{1}\op{b}\Enc(\varphi') \Enc(\psi)\cl{b}$ or $\op{2}\cl{2}\op{b}\Enc(\psi) \Enc(\varphi')\cl{b}$ (for some formula $\psi$). By the construction of $v(\sigma)$, for each $q \in Q$ (and $\gamma \in \Gamma$, if $\sigma \in \cl{\Sigma}$ )the set of all states $q' \in Q$ such that $q'$ might be enabled in the sub-play on $\Enc(\delta(q,\sigma))$ (resp. $\Enc(\delta(q,\gamma,\sigma))$) according to $\Bstrat$ and some valid counter-strategy for \pone satisfies the formula $\delta(q,\sigma)$ (resp. $\delta(q,\gamma,\sigma)$). In this way, the strategy $\Bstrat$ induces a run $\rho$ of $A$ on $w$ such that for each valid counter-strategy of \pone, the resulting rewriting of $v(w)$ corresponds to one path in $\rho$. 

Similarly, a run $\rho$ of $A$ on $w$ induces a strategy $\Bstrat$ for \ptwo on $v(w)$; if, for some tag $\sigma$ in $w$ and state $q \in Q$, $P \subs Q \times Q$ (resp. $P \subs Q$) is the follow-up state set satisfying $\delta(q, \sigma)$ (resp. $\delta(q, p, \sigma)$ for some appropriate $p \in Q$), $\Bstrat$ can be constructed to enable exactly the states from $P$ for all counter-strategies of \pone.

As the target automaton in $G$ accepts a rewriting of $v(w)$ if and only if it encodes a path in a run of $A$ on $w$ ending in an accepting state, the correspondence between runs of $A$ on $w$ and strategies of $\ptwo$ on $v(w)$ in $G$ implies that there exists a winning strategy for \ptwo on $v(w)$ in $G$ if and only if there is an accepting run of $A$ on $w$.
\end{proof}

Using Lemma \ref{lemma:ANWAtocfG}, it is easy to prove our lower bounds.

\begin{restate}{Proposition \ref{prop:lowergeneral}}
 For the class of unrestricted games $\Safelr$ is
  \begin{enumerate}[(a)]
  \item $\iiEXPTIME$-hard with
    bounded replay, and
  \item  $\PSPACE$-hard with
    no replay.
  \end{enumerate}
\end{restate}

\begin{proof}
  The proof that $\Safelr[k](\Gall)$ is $\iiEXPTIME$-hard for all $k \geq 2$ is by a reduction from the emptiness problem for ANWA, which is $\iiEXPTIME$-hard according to Proposition \ref{prop:anwacomplexity}(a). 
  
  Given an ANWA $A$, let $G'$ be the VP-cfG constructed by the algorithm of Lemma \ref{lemma:ANWAtocfG}. Let $G$ be the game with an additional new function symbol $s$ which \ptwo is allowed to rewrite by any string of the form $v(w)$ as defined in the proof of Lemma \ref{lemma:ANWAtocfG}. Then $L(A)$ is non-empty if and only if \ptwo  has a winning strategy on $\op{s}\cl{s}$ in $G$. This yields the desired reduction from emptiness for ANWA to $\Safelr[2](\Gall)$.\\

	$\PSPACE$-hardness of $\Safelr[1](\Gall)$ follows directly from the corresponding hardness result for the ANWA membership problem (Prop. \ref{prop:anwacomplexity} (b)) along with the existence of a polynomial-time reduction proven in Lemma \ref{lemma:ANWAtocfG}.
\end{proof}

\subsection*{Finite replacement languages}

\begin{restate}{Proposition \ref{prop:finitegeneral}}
For the class of unrestricted games with finite replacement languages,  $\Safelr(\calG)$ is
  \begin{enumerate}[(a)]
  \item $\EXPTIME$-complete with    unbounded replay, and
  \item $\PSPACE$-complete with  bounded or without replay.
  \end{enumerate}
\end{restate}

\begin{proof}
As already mentioned in the body of the paper, the lower bounds follow from  Theorem 4.3 in
\cite{MuschollSS06} and the proof of Proposition \ref{prop:lowergeneral}. Thus, only the upper bounds need to be established.

For (a), the non-emptiness test for $R_a\cap \widebar{L(A)}$ can be replaced by a membership test $v\in \widebar{L(A)}$, for each of the finitely many strings $v\in R_a$. This can be done in polynomial space by Proposition \ref{prop:anwacomplexity}. The exponential time upper bound then immediately follows because the number of iterations of the fixpoint process is at most exponential and the final test whether $w$ is accepted by $A_{\ceffect}$ needs only exponential time. 

For (b), a polynomial space algorithm for a bounded number $k$ of replay works basically just as in the general case, by first computing the call effect $\cneffect{k}$ from the input game $G$, then computing from it the ANWA $A_C^k$ from Proposition \ref{prop:calltoanwa} and finally simulating $A_C^k$ on the input string $w$.
The initial call effect,  $\cneffect{1}$, can again be computed in polynomial time. For each $i$, $\cneffect{i+1}$ can be computed from $\cneffect{i}$ in polynomial space and finally, whether $A_C^k$ accepts $w$ can be tested in polynomial space in $|w|$ and the number of states of $A_C^k$, that is, in the number of states of the target automaton of $G$.

Some care is needed though, as the (representation of the) intermediate automata and the resulting automaton  $A_C^k$ can be of superpolynomial (at most exponential) size. However, as usual for space bounded computations, the information about $A_C^k$ and the intermediate automata can be recomputed whenever it is needed. The composition of these constantly many polynomial space computations then yields an overall polynomial space bound. It is crucial here that, as observed in the proof of Proposition \ref{prop:anwacomplexity}, the evaluation of $A_C^k$ is possible in polynomial space in $|w|$ and the number of states of $A_C^k$.
\end{proof}
By a more complicated argument, the upper bound of Proposition \ref{prop:finitegeneral} (a) can even be established in the case where the finite replacement languages are not given explicitly but by NWAs.

\newpage		\section*{Proofs for Section \ref{sec:simple}}

For our upper bounds, we formalise XML Schema\footnote{For more background on formalisations of XML Schema we refer the reader to \cite{MartensNeven+Simple07}.}  target languages by way of simple NWA (as defined in Section \ref{sec:simple}) and use similar techniques as in the upper bound proofs for Section \ref{sec:general}. 
Lower bounds, on the other hand, will generally follow from lower bounds for context-free games on \emph{flat} strings, as defined in \cite{MuschollSS06}.

\subsection*{Upper bounds}

The general structure of the algorithms is the same as in Section \ref{sec:general}. Technically, the two main parts of the proof are to show that SANWAs are suitable (SNWAs can be computed from XML Schemas, Proposition \ref{prop:dtdsimple},  and SANWAs from (simple) game effects, Proposition \ref{prop:calltolanwa}) and to establish the complexity of SANWAs (Propositions \ref{prop:captest} and \ref{prop:lanwacomplexity}).

\subsubsection*{Suitability of simple NWAs}

First off, we prove that simple NWA are at least as expressive as single-type tree grammars. The idea behind this is rather straightforward, as we only need to combine DFAs for each type's content model and, on reading some opening tag $\op{a}$, start some DFA in a sub-computation to check the nested string between $\op{a}$ and the associated $\cl{a}$ for compliance with the content model of some type $X$. Thanks to the single-type property, the type $X$ is uniquely defined by $a$ and the context from which $\op{a}$ was read, so we obtain a deterministic automaton as desired.

\begin{restate}{Proposition \ref{prop:dtdsimple}}
From every  single-type tree grammar $T$, a simple DNWA $A$ can be
computed in polynomial time, such that $L(A)=L(T)$.
\end{restate}

\begin{proof}

	Let $T = (\Sigma, \Delta, S, P, \lambda)$ be a single-type tree grammar. We will construct a SNWA $A$ such that $L(T) = L(A)$.
	
	Due to the single-type property, for each type $X \in \Delta$ and each $a \in \Sigma$, there is at most one type $X'$ in the content model of $X$ with $\lambda(X')=a$. Without loss of generality, assume that there is exactly one such type for each $X$ and $a$ (which can be done by adding a ``dummy type" $X_\bot$ with $r_{X_\bot} = \emptyset$ to $T$), and denote this type by $\nu(X,a)$.
	
	For each $X \in \Delta$, let $A_X = (P_X, \Delta, \delta_X, p_{0,X}, F_X)$ be a DFA deciding $L(r_X)$ (which can be computed from the \emph{deterministic} regular expression $r_X$ in polynomial time). Assume w.l.o.g. that all $P_X, P_Y$ are disjoint for $X \neq Y$. Then, the SNWA $A = (Q, \Sigma, \delta, (p_0, 0), \{(p_f,0)\})$ is defined as follows:
	\begin{itemize}
	\item $Q = \{\bot\} \cup P \times \Delta'$, with 
		\begin{itemize}
		\item $P = \{p_0, p_f \} \cup \bigcup_{X \in \Delta} P_X$ and
		\item $\Delta' = \Delta \cup \{0\}$, with $0 \notin \Delta$
		\end{itemize}
	\item $\delta$ is defined by
		\begin{itemize}
		\item $\delta((p_0,0), \op{\lambda(S)}) = (q_{0,S},S)$,
		\item $\delta((p, X), \op{a}) = (p_{0, \nu(X,a)}, \nu(X,a))$ for each $a \in \Sigma$, $p \in P$, $X \in \Delta$,
		\item $\delta(q,q',\cl{a})$ is defined by $t$ below as per the definition of SNWA,		
		\end{itemize}
	\item $\Floc(a) = \bigcup_{X \in \Delta: \lambda(X)=a} (F_X \times \{X\})$, and
	\item $t$ is defined by
		\begin{itemize}
		\item $t((p_0,0), \lambda(S))=(q_f,0)$ and
		\item $t((p,X), a) = (\delta_X(p, a),X)$ for each $a \in \Sigma$, $p \in P$, $X \in \Delta$.
		\end{itemize}	
	\end{itemize}
	
	To show that $L(T) = L(A)$, it suffices to show that for every $w \in \wf$ and $X \in \Delta$, it holds that $\op{\lambda(X)} w \cl{\lambda(X)} \in L(X)$ if and only if $\delta^*((q_{0,X},X), w) \in \Floc(\lambda(X))$, where $L(X)$ is defined like $L(T)$ with root type $X$. The claimed equality then follows with $L(T) = L(S)$.
\end{proof}

\begin{proposition}\label{prop:calltolanwa}
  There is an algorithm that computes from the call effect
    $\ceffect$ of a simple game $G$ in polynomial time in $|\ceffect|$ and $|G|$
    a SANWA $A_C$ such that $L(A_C)= \safelr(G)$.
 \end{proposition}

\begin{proof}
 We construct $A^\ell_C$ almost as the automaton $A_C$  in the proof of
Proposition \ref{prop:calltoanwa}. However, as they are mimicking games, the
alternating transitions in $A_C$ occur at closing tags, whereas the
definition of simple ANWAs requires that  alternating transitions occur
only at opening tags. Thus we slightly adapt the construction as follows.

Let $A(T) = (Q, \Sigma, \delta, q_0, F_0)$ be a SNWA in normal form and let $P, \Delta,\Floc,t,\bot$
witness the simplicity of $A(T)$.
With $q_? \notin Q$, we let $A^\ell_C \mydef ((Q \cup \{q_?\}), \Sigma, \delta^\ell_C, q_0, F_0)$, where
$\delta^\ell_C$ is defined as follows
	\begin{itemize}
	\item For every $q\in Q$ and $a\in \Sigma$, where $q' = \delta(q,\op{a})$,\\
 \[
          \delta^\ell_C(q,\op{a}) \mydef \left((q',t(q,a)) \land (q',q_?)\right) \lor \bigvee_{X
            \in C(q)} \bigwedge_{r \in X} (q',r).
\]
	\item  For every $q,q'\in Q$ and $a\in \Sigma$, $\delta^\ell_C(q,q', \cl{a})$ is defined via a target state function $t_C$ as per the definition of SANWA.
	\end{itemize}
	
	The target state function $t_C$ and final state function $F_{\text{loc,C}}$ witnessing the simplicity of $A^\ell_C$ are defined by $t_C(q,a)\mydef q$ and $F_{\text{loc,C}}(a)\mydef\Floc(a)$, respectively. This automaton obviously fulfils both simplicity conditions, by construction and the simplicity of $A(T)$. The correctness of the automaton is proven analogously to the proof of Lemma \ref{lemma:anwaconstruction}.
\end{proof}

\subsubsection*{Complexity of simple ANWAs}

To prove the upper bound in Proposition \ref{prop:lanwacomplexity} (a), i.e., that non-emptiness for SANWAs is in \PSPACE, we start off by proving a somewhat stronger result: That the problem of determining, given a NWA $A$ and a SANWA $B$, whether there is a nested word accepted by both $A$ and $B$, is in $\PSPACE$. The standard approach for proving a result of this sort (a product construction between two NWA or two SANWA) is generally not feasible here, as SANWA are less expressive than NWA (so $A$ cannot in general be transformed into a SANWA) and transforming $B$
 into a NWA might incur a doubly exponential blow-up in size. Therefore, a $\PSPACE$ algorithm has to be constructed especially for this problem and uses the following  pumping property for strings in $L(A)\cap L(B)$.

As in the previous section, the \emph{width} of a nested word is the maximum number of children of any node in its corresponding forest. Its \emph{root width} is just the number of trees in its forest.  The \emph{depth} of a nested word is the depth of its canonical forest representation.

	\begin{lemma}\label{lemma:shortstrings}
		Let $A = (Q_A, \Sigma, \delta_A, q_{0,A}, F_A)$ be a NWA 
		and let $B = (Q_B, \Sigma, \delta_B, q_{0,B}, F_B)$ be a SANWA 
 with type alphabet $\Delta$, final state function $\Floc$, target state function $t$ and test state $q_? \in Q$. Then $L(A)\cap L(B)\not=\emptyset$ if and only if there exists a string in $L(A)\cap L(B)$ of width at most $2^{|Q_B|} \cdot |\Sigma| \cdot |Q_A|$ and depth at most $3(|\Sigma|+1)|Q_A|^2|\Delta|$.
	\end{lemma}
	
\begin{proof}
	The ``if'' direction is trivial. For ``only if'', assume for the sake of contradiction that $L(A)\cap L(B)\not=\emptyset$, but there is no string in $L(A)\cap L(B)$ fulfilling the claimed upper bounds on both width and depth.

	First, we observe that for all words $w \in L(B)$, all nodes of any depth $i$ in an arbitrary accepting run of $B$ on $w$ contain only linear states of the same type, i.e. if $\rho = (D, \lambda)$ is an accepting run of $B$ on $w$, and $x, y \in D$ with $|x| = |y| = i$ for any $i \in \mathbb{N}$, and if $\linstate{\lambda(x)} = (p,X)$ and $\linstate{\lambda(x)} = (p', Y)$, then $X=Y$. This can be proven by a simple induction on $i$.
	
	In the remainder of this proof, if $\rho$ is a run of $B$ on some string $w$ and $\rho'$ is a sub-run of $\rho$ on a nested substring $w'$ of $w$, we call $\rho'$ \emph{successful} if all leaves of $\rho'$ are accepting with respect to the context of $w'$, i.e. if all leaves of $\rho'$ are in $F$ in case $w'=w$, or if all leaves of $\rho'$ are in $\Floc(a)$ in case $\op{a}w'\cl{a}$ is a substring of $w$. Note that due to the definition of runs, all test subruns of $\rho'$ (i.e. subruns starting with horizontal state $q_?$) have to accept. Furthermore, by the above observation, all subtrees of $\rho'$ immediately below its root start from the same state, as that state is uniquely given by the tags enclosing $w'$ and the root type of $\rho'$.
	
	First off, let $w \in L(A) \cap L(B)$ be a string of width greater than $2^{|Q_B|} \cdot |\Sigma| \cdot |Q_A|$ and minimal length. We now prove that $L(A) \cap L(B)$ contains a string shorter than $w$, in contradiction to the assumed minimality.

	Let $w'$ be a maximum-length nested substring of $w$ with root width greater than $2^{|Q_B|} \cdot |\Sigma| \cdot |Q_A|$. Let $\rho$ be an accepting run of $B$ on $w$ and $\rho'$ its sub-run on $w'$. Similarly, since $A$ may also be viewed as an ANWA, there is an accepting run $\pi$ of $A$ on $w$ in which each non-leaf node has only a single child. Let $\pi'$ be the sub-run of $\pi$ on $w'$. For $k=1,..,|w'|$, let $\layer[k](w') \in \Sigma \times ((\Pot(Q_B) \times Q_A) \cup (\Pot(Q_B^2) \times Q_A^2))$ such that if $w'_k = \op{a}$, $(q_1,p_1), .., (q_\ell, p_\ell)$ are all pairs of states at depth $k$ in $\rho'$ and $(q,p)$ is the state pair of depth $k$ in $\pi'$, then $\layer[k](w')=(a, \{(q_1, p_1),..,(q_\ell,p_\ell)\}, (q,p))$, and if $w'_k = \cl{a}$, $(q_1), .., (q_\ell)$ are all states at depth $k$ in $\rho'$ and $q$ is the state at depth $k$ in $\pi'$, then $\layer[k](w')=(a, \{q_1,..,q_\ell\},q)$. As the root width of $w'$ is greater than $|\Sigma \times \Pot(Q_B) \times Q_A|$, there are numbers $i < j < |w'|$ such that $\layer[i](w')=\layer[j](w')$ and the substrings $w'_1..w'_i$ and $w'_1..w'_j$ (and therefore also $w'_{i+1}\ldots w'_j$) are well-nested. The claim, then, is that there are accepting runs of $A$ and $B$ on the string $\tilde{w}$ derived from $w$ by deleting $w'_{i+1}\ldots w'_j$ from $w'$.
	
\skipproof{	Let $\layer[i](w')=\layer[j](w') = (a, \{q_1,..,q_\ell\},q)$ for some $a \in \Sigma$, $q_1, \ldots, q_\ell \in Q_B$ and $q \in Q_A$. An accepting run for $A$ on $\tilde{w}$ can easily be obtained from $\pi$ by cutting out the sub-run on $w'_{i+1}\ldots, w'_j$. For $B$, we construct an accepting run $\tilde{\rho}$ on $\tilde{w}$ by constructing a sub-run $\tilde{\rho}'$ of $B$ on $w'_1 \ldots w'_i w'_{j+1} \ldots w'_{|w'|}$ from $\rho'$ that is successful if $\rho'$ is successful, and then replacing $\rho'$ with $\tilde{\rho}'$ in $\rho$.
	
	For each $i \in [\ell]$, denote by $\rho'_i$ the sub-run of $\rho'$ on $w'_{j+1} \ldots w'_{|w'|}$ starting at $q_i$. Then it is easy to see that the tree $\tilde{\rho}'$ obtained from $\rho'$ by replacing each sub-run on $w'_{i+1} \ldots w'_{|w'|}$ starting at $q_i$ with $\rho'_i$ is indeed a run of $B$ on $w'_1 \ldots w'_i w'_{j+1} \ldots w'_{|w'|}$. Since the states at the leaves of $\tilde{\rho}'$ form a subset of the states at the leaves of $\rho'$, it is also clear that $\tilde{\rho}'$ is successful if $\rho$ is. It follows that replacing $\rho'$ with $\tilde{\rho}'$ in $\rho$ yields an accepting run of $B$ on $\tilde{w}$. With $\tilde{w}$, we have therefore constructed a string in $L(A) \cap L(B)$ that is strictly shorter than $w$, which yields the desired contradiction. $L(A) \cap L(B)$ must therefore contain a string of width at most $2^{|Q_B|} \cdot |\Sigma| \cdot |Q_A|$.
}	

	Assume now, again for the sake of contradiction, that there is no string in $L(A) \cap L(B)$ of width at most $2^{|Q_B|} \cdot |\Sigma| \cdot |Q_A|$ and depth at most $3(|\Sigma|+1)|Q_A|^2|\Delta|$. By the above part of the proof, this means that all strings fulfilling the requirement on width must be of a depth exceeding $3(|\Sigma|+1)|Q_A|^2|\Delta|$. Let $w$ be such a string of minimal length, and let $\rho$ be an accepting run of $B$ and $\pi$ an accepting run of $A$ on $w$. 
	
	As the nesting depth of $w$ is greater than $3(|\Sigma|+1)|Q_A|^2|\Delta|$, there exist
well-nested strings $w'$ and $w''$ such that for some $a \in \Sigma$, 
	\begin{itemize}
		\item $\op{a}w'\cl{a}$ is a substring of $w$, 
		\item $\op{a}w''\cl{a}$ is a substring of $w'$, 
		\item all sub-runs of $\rho$ on $w'$ and $w''$ start from the same state $q_a \in Q_B$
		\item either all sub-runs of $\rho$ on $w'$ and $w''$ are unsuccessful or there exist successful runs in $\rho$ on both $w'$ and $w''$, and 
		\item the states of $A$ according to $\pi$ before and after reading $\op{a}w'\cl{a}$ are the same as those before and after reading $\op{a}w''\cl{a}$.		
	\end{itemize}
	
	The claim is that both $A$ and $B$ have accepting runs on the string $\tilde{w}$ derived from $w$ by replacing $w'$ with $w''$. As $w''$ is a proper substring of $w'$, proving this claim yields the desired contradiction to the minimal length of $w$ and thus the claim of Lemma \ref{lemma:shortstrings}
\end{proof}

\begin{proposition}\label{prop:captest}
  There is an alternating algorithm that tests in polynomial time whether, for an
  NWA $A$ and a SANWA $B$ it holds $L(A)\cap L(B)\not=\emptyset$. 
\end{proposition}

\begin{proof}
	We formulate the claimed algorithm as a game for two players, whom we will call \playerA and \playerE to avoid confusion with the players for context-free games. This game will always terminate after at most polynomially many rounds, so an alternating polynomial-time algorithm can easily be constructed from it by branching nondeterministically (resp. universally) for the moves for \playerE (resp. \playerA) and accepting the input if and only if \playerE wins.
	
	We will construct the game such that that \playerE has a winning strategy on input NWA $A=(Q_A,\Sigma, \delta_A, q_{0,A}, F_A)$ and SANWA $B=(Q_B,\Sigma, \delta_B, q_{0,B}, F_B)$ with final state function $\Floc$, test state $q_?$ and target state function $t$ if and only if $L(A) \cap L(B) \neq \emptyset$. \playerEs goal in this game is to prove that there is a string that is accepted by both $A$ and $B$ without writing down that string explicitly; by Lemma \ref{lemma:shortstrings}, it does suffice to examine strings of at most exponential width and polynomial depth, but such a string can still not be explicitly spelled out using only polynomial space. We therefore represent a string implicitly by the behaviour it induces in $A$ and $B$.
	
	Game positions for \playerE consist of two states $p_1, p_2 \in Q_A$, a function $S : Q_B \rightarrow \Pot(Q_B)$ and two numbers $c,n \geq 0$, and the game is constructed in such a way that \playerE has a winning strategy from position $(q_1, q_2, S,c,n)$ if and only if there is a string $w$ of root width at most $2^c$ and nesting depth at most $n$ such that $q_1 \reach[w]_A q_2$,
	 and for every $q \in Q_B$ there is a run of $B$ on $w$ beginning in $q$ and ending inside $S(q)$. We write $c_0$ for $|Q_B| \log(|\Sigma| \cdot |Q_A|)$ and $n_0$ for $3(|\Sigma|+1)|Q_A|^2|\Delta|$.
Lemma \ref{lemma:shortstrings} then guarantees that $L(A) \cap L(B)$ is nonempty if and only if there is a state $q_f \in F_A$ and a function $S$ with $S(q_{0,B}) \subs F_B$ such that \playerE has a winning strategy from position $(q_{0,A}, q_f, S,c_0,n_0)$.
		
	In any position $(q_1, q_2, S,c,n)$, \playerE has the following options:
	\begin{itemize}
	\item If $c>0$, she may choose to play a \emph{concatenation round}, asserting that $w=v_1 v_2$ for strings $v_1, v_2 \in \wf$ whose root width is at most half that of $w$. In this case, she chooses two functions $S_1, S_2 : Q_B \rightarrow \Pot(Q_B)$, corresponding to strings $v_1, v_2$ as above and an ``in-between" state $q'\in Q_A$. The functions $S_1$ and $S_2$ have to fulfil the condition that for each $q \in Q_B$ it holds that $S(q) = \bigcup_{p \in S_1(q)}S_2(p)$; if $S_1$ and $S_2$ do not fulfil this condition, \playerA wins. Otherwise, \playerA has a choice of which part of \playerEs assertion he wants to contest, so he may choose as a follow-up position either $(q_1, q', S_1, c-1, n)$ or $(q', q_2, S_2, c-1, n)$.
	\item If $n>0$, \playerE may choose to play a \emph{nesting round} (with $a \in \Sigma$), asserting that $w=\op{a}v\cl{a}$ for some $v \in \wf$. To this end, she first chooses an alphabet symbol $a$ and a function $S'$ corresponding to $v$ as above, as well as states $q'_1, q'_2,p \in Q_A$ such that $\delta_A(q_1, \op{a}) = (q'_1, p)$ and $\delta_A(q'_2,p,\cl{a}) = q_2$ (if no such states exist, \playerA wins immediately). Next, \playerA chooses some $q \in Q_B$ on which to contest \playerEs claim. In response, \playerE picks a state $p' \in Q_B$ and a set of states $P=\{p_1,..,p_k\} \subs Q_B$ such that $(\{p'\} \times P) \models \delta_B(q, \op{a})$ and $S(q) = \bigcup_{p \in P \setminus \{q_?\}}\{t(p,a)\}$. If she cannot choose such a set, \playerA wins.
	
	If \playerA has not won by this point, he has to contest \playerEs claim that there is a string $v$ such that $B$ has a successful run on $v$. If $q_? \in P$ and $S'(p') \not \subs F(a)$, the string $v$ claimed by \playerE fails the test subrun mandated by $B$ branching with $q_?$, so in this case, \playerA wins. Otherwise, the game continues from position $(q'_1, q'_2, S',c_0, n-1)$, as the root width of the  substring $v$ is bounded by $2^{c_0}$.
	\item \playerE may choose to \emph{solve (with $a \in \Sigma$)}, asserting that $w=\op{a}\cl{a}$. In this case, she chooses a symbol $a \in \Sigma$. Similar to a nesting round, \playerA then picks a state $q \in Q_B$ on which to contest \playerEs claim, to which \playerE responds by choosing a state $p \in Q_B$ and a set $P \subs Q_B$. The game then ends and a winner is determined. \playerE wins if and only if all of the following conditions are fulfilled:
		\begin{enumerate}[(a)]
		\item There are states $p', q' \in Q_A$ such that $\delta_A(q_1, \op{a}) = (q',p')$ and \mbox{$\delta_A(q',p',\cl{a}) = q_2$};
		\item $(\{p\} \times P) \models \delta_B(q, \op{a})$;
		\item $S(q) = \bigcup_{p \in P \setminus \{q_?\}}\{t(p,a)\}$;
		\item If $q_? \in P$, then $p \in F(a)$.
		\end{enumerate}
	\item \playerE may choose to \emph{solve with $\epsilon$}, asserting that $w = \epsilon$. In this case, the game ends and \playerE wins if and only if $q_1=q_2$ and for each $q \in Q_B$ it holds that $S(q) = \{q\}$.
    \end{itemize}
	
	Since each round that does not end the game decreases either the number of remaining nesting or concatenation rounds and the number of remaining concatenation rounds only increases at the end of a nesting round, the total number of rounds starting from $(q_{0,A}, q_f, S,c_0,n_0)$ is bounded by $c_0n_0$, which is polynomial in the size of $A$ and $B$. It is easy to see that each choice by \playerE or \playerA requires only a polynomial-size certificate, and that each check for winning conditions is computable in polynomial time. Therefore, an alternating algorithm checking whether \playerE has a winning strategy on this game (as described above) has a polynomial upper bound on its running time. It remains to be shown that this algorithm indeed tests $A$ and $B$ for intersection emptiness, i.e. that \playerE has a winning strategy from $(q_{0,A}, q_f, S,c_0,n_0)$ for some $q_f \in F_A$ if and only if $L(A) \cap L(B) \neq \emptyset$.
	
	To prove this claim, we show that the following statements are equivalent:
	\begin{enumerate}[(1)]
	\item \playerE has a winning strategy from position $(q_1, q_2, S, c,n)$;
	\item There is a string $w \in \wf$ of width at most $2^{|Q_B|}|\Sigma||Q_A|$, root width at most $2^c$ and depth at most $n$ such that there is a run of $A$ on $w$ from $q_1$ to $q_2$, and for each $q \in Q$, there is a successful run of $B$ on $w$ from $q$ ending inside $S(q)$.
	\end{enumerate}
	
	$(1) \Rightarrow (2)$: Assume \playerE has a winning strategy $\Astrat$ from position $(q_1, q_2, S, c,n)$. We prove (2) by induction on the structure of $\Astrat$.
	
	If \playerE solves with $\epsilon$ as her first move according to $\Astrat$, the string $w=\epsilon$ obviously fulfils the claim of (2).	
	
	If \playerEs first move according to $\Astrat$ is to solve with some $a \in \Sigma$, then $w=\op{a}\cl{a}$ fulfils the claim of (2). Since $c,n \geq 0$, $w$ fulfils the desired upper bounds on nesting depth and width; winning condition (a) ensures the existence of a run of $A$; and as for each $q \in Q$ that \playerA chooses, \playerE can respond with a set of horizontal states compliant with the transition formulae of $B$ according to winning conditions (b) to (d), the desired runs of $B$ on $w$ exist as well.
	
	If \playerE begins with a concatenation round according to $\Astrat$, it follows that there exist a state $q' \in Q_A$ and functions $S_1, S_2 : Q_B \rightarrow \Pot(Q_B)$ such that for each $q \in Q_B$ it holds that $S(q) = \bigcup_{p \in S_1(q)}S_2(p)$ and \playerE has a winning strategy on both $(q_1, q', S_1, c-1, n)$ and $(q', q_2, S_2, c-1, n)$. By induction, this implies that there are strings $v_1, v_2 \in \wf$ of width at most $2^{|Q_B|}|\Sigma||Q_A|$, root width at most $2^{c-1}$ and depth at most $n$ for which there exist appropriate runs of $A$ and $B$; it is easy to see that $w = v_1 v_2$ fulfils the width and depth requirements of the claim, and that the claimed runs of $A$ and $B$ on $w$ can be constructed by combining those on $v_1$ and $v_2$.
	
	If \playerE starts by playing a nesting round with some $a \in \Sigma$, there exists a function $S'$ as well as states $q'_1, q'_2,p \in Q_A$ such that $\delta_A(q_1, \op{a}) = (q'_1, p)$ and $\delta_A(q'_2,p,\cl{a}) = q_2$. Furthermore, for each $q \in Q_B$, there is a state $p \in Q_B$ and a set $P \subs Q_B$ such that  $(\{p\} \times P) \models \delta_B(q, \op{a})$ and $S(q) = \bigcup_{p \in P \setminus \{q_?\}}\{t(p,a)\}$, and if $q_? \in P$ then $S'(p) \subs F(a)$. Finally, \playerE has a winning strategy starting from position $(q'_1, q'_2, S', c_0, n-1)$.
	
	By induction, there exists a string $v$ of width at most $2^{|Q_B|}|\Sigma||Q_A|$ and depth at most $n-1$ for which there exist appropriate runs of $A$ and $B$; the string $w = \op{a}v\cl{a}$ therefore fulfils the claimed restrictions on depth and width. Again, it is easy to see that a run of $A$ on $w$ can be constructed from the one on $v$. To construct the desired runs of $B$ on $w$, denote the run on $v$ starting at $p$ and ending in $S'(p)$ by $\rho$ and let $q \in Q_B$. A run on $w$ starting at $q$ is then constructed as follows: The root node, labelled $q$, has $\{p\} \times P$ as the set of labels of its children. Each of these nodes $(p,p')$ is the root of a copy of $\rho$, whose leaves are all inside $S'(p)$; if $p'=q_?$, the leaves of the corresponding copy of $\rho$ have no further children; otherwise, their only child is labelled with the state $t(p',a)$. Using the above properties and the definition for SANWA semantics, it is easy to verify that the tree thus constructed is indeed a run of $B$ on $w$ starting at $q$ and ending inside $S(q)$.
	
	$(2) \Rightarrow (1)$: This part of the proof is by an induction on the structure of $w$ analogous to the above proof of $(1) \Rightarrow (2)$. 
\end{proof}

\begin{restate}{Proposition \ref{prop:lanwacomplexity}}
  \begin{enumerate}[(a)]
  \item Non-emptiness for SANWA is $\PSPACE$-complete.
  \item The membership problem for SANWA is decidable in polynomial
    time.
  \end{enumerate}
\end{restate}

\begin{proof}
That non-emptiness for SANWAs is in \PSPACE follows directly from Proposition \ref{prop:captest}, as alternating polynomial time equals polynomial space.

	$\PSPACE$-hardness can be proven by a simple reduction (with a constant-sized NWA $A$ accepting $\wf$) from the nonemptiness problem for SANWA, which in turn is $\PSPACE$-hard by reduction from the nonemptiness problem for alternating finite automata, interpreting flat strings $w_1 \ldots w_n \in \Sigma^*$ as nested strings $\op{w_1}\cl{w_1}\ldots\op{w_n}\cl{w_n} \in \wf$ of nesting depth 0 (and vice versa). It is then quite easy to construct from an AFA $B'$ a SANWA $B$ such that $L(B') \neq \emptyset$ if and only if $B$ accepts some nested string of depth 0.

Together, statement (a) follows.

To show (b),  that the membership problem for SANWAs can be decided in polynomial time,   it suffices to show that the problem can be decided
  by an alternating Turing machine with
  logarithmic space. The computation of a SANWA $A$ can be easily simulated by an
  alternating Turing machine $M$. To this end,
  the TM $M$ could branch existentially and universally, just as $A$. In
  particular, on a word $w$ it would have exactly one run for each run of $A$
  on $w$. However, such a na\:{i}ve simulation would need to remember the
  stack contents to compute $t(p,a)$ at the next closing tag
  $\cl{a}$, and thus the space required would be proportional to the nesting
  depth of the input word.
  
To achieve a logarithmic space bound, we can modify $M$ as follows. Whenever a
transition at an opening tag $\op{a}$ yields a pair $(q,p)$ with
$p\not=q_?$, the computation branches universally into two
subcomputations: one moves directly to the corresponding closing tag
$\cl{a}$ and continues after that from state $t(p,a)$. The other
proceeds as $A$ on the current subword but does \emph{not} need to
remember $p$. Whenever such a computation reaches a closing tag it
accepts. Test subruns, starting from a pair $(q,q_?)$ are simulated
slightly different: they remember the nesting depth of the opening tag
$\op{a}$ and behave at the corresponding closing tag just as $A$
would. However, if a test subrun starts a test-subsubrun the latter
only needs to remember the new nesting depth, as it can stop when the
subsubrun has finished. 

The correspondence between runs of $A$ and the ATM can be shown by
induction on the nesting depth of the input word $w$. In particular,
the ATM accepts just if $A$ does.\\

More formally, we claim that Algorithm \ref{alg:lanwa-mc} evaluates a SANWA $A=(Q, \Sigma, \delta, q_0, F)$ with local acceptance function $\Floc$, test state $q_?$ and target state function $t$ on a nested word $w = w_1 \ldots w_n\in \wf$ (with $w_i \in \op{\Sigma} \cup \cl{\Sigma}$ for each $i \in [n]$). To this end, it keeps track of a current state $q \in Q$ of $A$, two indices $i$ and $j$ denoting the starting and ending position in $w$ of the substring to be verified in its current run, and an index $f \in \Sigma \uplus \{0\}$ that tracks whether the current string is to be verified against the accepting states of $A$ ($f=0$) or some $\Floc(a)$ ($f=a$). To simplify notation for the former case, we let $\Floc(0) \mydef F$.

 	\begin{algorithm}[h]
   \caption{Verify($A, w$)}
   \label{alg:lanwa-mc}
   \begin{algorithmic}[1]
   	\STATE $q \gets q_0$
   	\STATE $i \gets 1$
   	\STATE $j \gets n$
   	\STATE $f \gets 0$
   	\WHILE{$i \leq j$}
   		\STATE //$i$ always denotes the position of an opening tag
		\STATE Choose alternatingly $(q', p)$ according to $\delta(q, w_i)$
		\IF{$p \neq q_?$}
			\STATE $i \gets$ (position of closing tag associated with $w_i$) + 1
			\STATE $q \gets t(p,w_i)$
		\ELSE
			\STATE //$p = q_?$; start test subrun:
			\STATE $q \gets q'$
			\IF{$w_{i+1}\in \op{\Sigma}$}
				\STATE $i \gets i+1$				
				\STATE $j \gets$ (position of closing tag associated with $w_i$)			
			\ELSE
				\STATE //$w_{i+1}$ is the associated closing tag of $w_i$; end test subrun and exit loop to test for acceptance.
				\STATE $j \gets i$
				\STATE $i \gets i+1$
			\ENDIF
		\ENDIF
   	\ENDWHILE
   	\IF{$q \in \Floc(f)$}
   		\STATE Accept
   	\ELSE
   		\STATE Reject
   	\ENDIF
   \end{algorithmic}
 \end{algorithm}	

	We first elaborate on how to execute line 7 of Algorithm \ref{alg:lanwa-mc} in alternating logarithmic space. Assume that each transition function $\delta(q,\op{a})$ in $A$ is given in prefix notation, i.e. formulas are of the form (i) $\land(\varphi_1, \varphi_2)$ or (ii) $\lor(\varphi_1, \varphi_2)$ or (iii) $(q',p)$. In case (i), the algorithm guesses universally whether to branch into $\varphi_1$ or $\varphi_2$, in case (ii) this choice is existential, and in case (iii), a result is fixed. Clearly, this is feasible in alternating logarithmic space and equivalent to first choosing existentially a set $P \subs Q^2$ with $P \models \delta(q,\op{a})$ and then universally picking a tuple $(q',p) \in P$.

	It is also easy to see that Algorithm \ref{alg:lanwa-mc} terminates (as the value of $i$ increases in each iteration of the loop in line 5 while $j$ only ever decreases) and requires only logarithmic space.
	
	It remains to be proven that Algorithm \ref{alg:lanwa-mc} is correct. We do this by proving that, for any nested word $w$, state $q \in Q$, $f \in \Sigma \cup \{0\}$ and indices $i,j$ such that $w_i \ldots w_j$ is a well-nested string, lines 5-24 of Algorithm \ref{alg:lanwa-mc} accept in an alternating fashion if and only if there is a run of $A$ on $w_i \ldots w_j$ starting at $q$ and ending inside $\Floc(f)$. The proof is by induction on the structure of $w$ and uses as its crucial component the above insight that picking a follow-up state tuple from $\delta(q, w_i)$ in line 7 is equivalent to universally selecting a child of a depth $i$ node labelled $q$ in some run of $A$, and that each existential strategy for the alternating execution of Algorithm \ref{alg:lanwa-mc} corresponds to a single run of $A$ in this way. 
\end{proof}

\subsection*{Lower bounds}

All of our lower bounds for simple games follow from lower bounds for cfGs on flat strings, that is, games on strings as defined in \cite{MuschollSS06}, with target and replacement languages represented by deterministic regular expressions.
Lower bound results for 
replay-free games and bounded replay with finite replacement languages and target languages represented as DFAs were already proven in \cite{MuschollSS06}. They can be transferred to games with target languages described by deterministic regular expressions. As an entirely new result compared to \cite{MuschollSS06}, we prove here a lower bound for 
bounded replay (actually, \Call depth 2 suffices) and later sketch how these results carry over to nested word cfGs.\footnote{Note that our $\PSPACE$ lower bound for bounded replay is not in conflict with the corresponding $\PTIME$ upper bound in \cite{MuschollSS06}. This is because the $\PTIME$ upper bound given there required replacement languages to be finite, whereas we consider here replacement languages given by arbitrary deterministic regular expressions, which may be infinite.}

Intuitively, a regular expression is deterministic, if each of its positions can be matched uniquely with a symbol of the regular expression, without lookahead. Formally let, for a regular expression $r$, $D(r)$ be the expression, in which the $i$-th symbol $\sigma$ of $r$ is replaced by $(\sigma,i)$, e.g. $D((a+b)^*a)=((a,1)+(b,2))^*(a,3)$. We call $r$ is \emph{deterministic}, if there do not exist strings $w,v,v'$, symbol $\sigma$ and numbers $i,j$ such that $w(\sigma,i)v\in L(D(r))$, $w(\sigma,j)v'\in L(D(r))$ and $i\not=j$.

\begin{lemma}\label{lemma:drepspace}
	For the class of games on flat strings with target and replacement languages specified by deterministic regular expressions, $\Safelr$ is $\PSPACE$-hard with bounded replay of Call depth 2.
\end{lemma}

\begin{proof}
        We prove this 
        by reduction from the complement of the problem \algprobname{Corridor Tiling}: Given a set $U$ of \emph{tiles}, relations $V,H \subs U \times U$ of \emph{vertical} and \emph{horizontal constraints}, \emph{initial} and \emph{final} tiles $u_i, u_f \in U$ and a number $n$ (represented in unary), is there a correct tiling of width $n$ and arbitrary height that starts with $u_i$, ends with $u_f$ and violates none of the (vertical or horizontal) constraints.
	
	Formally, a \emph{tiling} of width $n$ and height $m$ is a mapping $t:[n]\times[m]\to U$. A tiling $t$ is \emph{valid}  if
  \begin{itemize}
  \item $t(0,0)=u_i$,
  \item $t(n,m)=u_f$,
  \item for every $i\in [n-1]$ and $j\in [m]$,
    $(t(i,j),t(i+1,j))\in H$, and
  \item for every $i\in [n]$ and $j\in [m-1]$,
    $(t(i,j),t(i,j+1))\in V$.
  \end{itemize}
\algprobname{Corridor Tiling} asks whether an instance $\mathcal{I}=(U,u_i,u_f,V,H,n)$ has a valid tiling of width $n$. It is well known that this problem is $\PSPACE$-complete (see, e.g., \cite{Chlebus86} for a slightly different definition of tilings). Since $\PSPACE$ is closed under complementation, the complement of \algprobname{Corridor Tiling} is complete for $\PSPACE$ as well.

We give here a reduction from the complement of \algprobname{Corridor Tiling} to $\Safelr$.

The reduction constructs, given an instance $\mathcal{I}=(U,u_i,u_f,V,H,n)$ for \algprobname{Corridor Tiling}, a game $G=(\Sigma,\funcsymb,R,T)$ and a symbol $s$ from $\funcsymb$  such that \pone has a winning strategy on $s$ if and only if $\mathcal{I}$ does \emph{not} have a valid corridor tiling. The basic idea is that, after \pones first \Call move on $s$, \ptwo will answer with an encoding $w$ of a valid corridor tiling, if one exists. With her \Call moves of depth 2, \pone may then try to flag inconsistencies (i.e. constraint violations) in the tiling given by \ptwo; finally, the target automaton should accept a tiling if \pone did indeed point out an actual inconsistency.

The game $G$ is over an alphabet $\Sigma$ which is obtained by the union of $U$ with a set $\hat{U}$ of disjoint copies $\hat{u}$ of all elements $u \in U$ and the set $\{s, ?_h, ?_v, !_h, !_v, \#\}$ for some $s,?_h, ?_v, !_h, !_v, \# \notin U$. We set $\funcsymb=U \cup \{s, ?_h, ?_v\}$ The replacement and target languages are described below.

A \emph{tiling candidate} (for $\mathcal{I}$) is a string of the form $(((U ?_v ?_h)^n \#)^*$, whose length-$n$ blocks of elements from $U$ are supposed to be interpreted as lines of a tiling, with \emph{protest symbols} $?_v, ?_h$ after each tile and a \emph{line separator} symbol $\#$ at the end of each line. The replacement language $R_s$ consists of all tiling candidates $v$ such that $u_i$ is the first symbol of $v$. It is easy to see that $R_s$ can be described by a DRE of polynomial size in $|\mathcal{I}|$. The other replacement languages are very simple: $R_u = \{\hat{u}\}$ for each $u \in U$, $R_{?_h} = \{!_h\}$ and $R_{?_v} = \{!_v\}$.

The construction of the target language $T$ is best motivated by sketching how plays can proceed on the input string $s$. First, \pone should be forced to play \Call on $s$ and allow \ptwo to actually give a candidate for a valid tiling.
Therefore, $s \notin T$.

By the definition of $R_s$, \ptwo responds to this \Call with a tiling candidate which already begins with the correct tile.
It is now \pones task to flag an error in this tiling, i.e. either
\begin{itemize}
\item two tiles separated by $(n-1)$ tiles with corresponding protest symbols and one line separator symbol (a potential vertical error), or
\item two tiles separated by exactly $2$ protest symbols (a potential horizontal error), or
\item a single tile at the end of $v$ (a potential incorrect final tile).
\end{itemize} 
To flag any tiles, \pone plays \Call on them, forcing \ptwo to replace any called tile $x$ by a \emph{marked tile} $\hat{x}$. If the marked tiles indeed make up an error, we want \pone to win, so the DRE for $T$ should describe such strings. If, on the other hand, \pone tries to cheat by marking too few or too many tiles, or tiles that do not make up an error, she should lose the game.

To allow easy DRE-based checking of the three types of errors mentioned above, \pone also has to specify the \emph{type} of error right after the first tile she flagged; in case of a horizontal (vertical) error, she has to \Call $?_h$ ($?_v$) to have it replaced with $!_h$ ($!_v$). This basically ``tells" the target DRE what sort of error to check for. An incorrect final tile does not need its own type of protest symbol, because (as we will see) a flagged inconsistency of this sort can be recognised by a DRE ``as is".

To construct the target language DRE, we first define some abbreviations:
\begin{itemize}
\item For any set $S=\{s_1, \ldots s_k\}$ and REs $\alpha_s$ for each $s \in S$, $\bigoplus_{s \in S} \alpha_s$ stands for the RE $\alpha_{s_1} + \ldots + \alpha_{s_k}$; 
\item $U'$ denotes the DRE $\bigoplus_{u \in U} u ?_v ?_h$, and $(U'+\#)^k$ the $k$-fold repetition of $(U'+\#)$;
\item for each $u \in U$, $V_u \mydef\ !_v ?_h (U' + \#)^{n} (\bigoplus_{(u,u') \notin V} \hat{u'}) ?_v ?_h (U'+\#)^*$;
\item for each $u \in U$, $H_u \mydef\ !_h (\bigoplus_{(u,u') \notin V} \hat{u'}) ?_v ?_h (U'+\#)^*$;
\end{itemize}

It is easy to verify that for each $u \in U$, $V_u$, and $ H_u$ are DREs of polynomial size in $|\mathcal{I}|$. Intuitively, $V_u$ ($H_u$) describes all suffixes immediately to the right of $\hat{u}$ ($\hat{u} ?_v$) in tilings where \pone has correctly flagged a vertical (horizontal) error starting with $u$. 

The target language $T$ can now be described by the DRE 
\[
(U'+\#)^* \big(\hat{u_f}(V_{u_f}+ ?_v H_{u_f}) +  \bigoplus_{u \in U\setminus \{u_f\}} \hat{u}(V_u + ?_v( H_u + ?_h \#)\big), 
\]
which is also of polynomial size in $|\mathcal{I}|$.
It is easy to see that if a valid tiling exists, \ptwo can simply win
the game by providing it in the first move. Therefore, in this case,
\pone does \emph{not} have a winning strategy. On the other hand, if
no tiling exists, \ptwo can only give a tiling candidate with 
at least one (vertical, horizontal or final tile) error in his first move and \pone can win by marking one
such error.
\end{proof}

The following two results can be shown by careful adaptation of the corresponding lower bound proofs in \cite{MuschollSS06}.

\begin{lemma}\label{lemma:dreptime}
	For the class of games on flat strings with target and replacement languages specified by deterministic regular expressions, $\Safelr$ is $\PTIME$-hard  (under logspace reductions) without replay.
\end{lemma}

\begin{lemma}\label{lemma:dreexptime}
	For the class of games on flat strings with target and replacement languages specified by deterministic regular expressions, $\Safelr$ is $\EXPTIME$-hard with unlimited replay.
\end{lemma}

\begin{restate}{Proposition \ref{prop:lowerdtd}}
 For the class of games with target languages specified by DTDs, $\Safelr$ is
  \begin{enumerate}[(a)]
  \item $\EXPTIME$-hard with
    unrestricted replay,
  \item  $\PSPACE$-hard with
    bounded replay, and
  \item $\PTIME$-hard  (under logspace-reductions) without replay
  \end{enumerate}
\end{restate}

\begin{proof}
All lower bounds follow by the same reduction from corresponding lower bounds for flat cfGs, which were just given as Lemma \ref{lemma:dreexptime}, Lemma \ref{lemma:drepspace} and Lemma \ref{lemma:dreptime}. 
	
	The idea for the reduction from flat cfGs to simple (nested) cfGs is as follows: All input and replacement strings $w=w_1 \ldots w_n \in \Sigma^*$ are replaced by $\nw{w} = \op{w_1}\cl{w_1}  \ldots \op{w_n}\cl{w_n} \in \wf$; to this end a target DFA $A(T)$ is simulated by a SNWA in normal form with an extra state $q_n$ such that $\delta(q, \op{a})=q_n$ for each $q$ and $a$, $\Floc(a) = q_n$ for each $a$, and $t(q,a)$ is the transition function of $A(T)$. Replacement NFAs are similarly transformed into NWAs.
\end{proof}

Using the reduction from the proof of Proposition \ref{prop:lowerdtd}, Theorem \ref{theo:simple} also yields the following result, which we will need in later proofs:

\begin{corollary}\label{coro:flatptime}
	For the class of games on flat strings with target languages specified by DFAs, $\Safelr$ is $\PTIME$-complete without replay.
\end{corollary}

 \begin{restate}{Proposition \ref{prop:simplefinite}}
For the class of games with target languages specified by XML Schemas and
explicitly enumerated finite replacement languages, $\Safelr$ is
  \begin{enumerate}[(a)]
  \item $\EXPTIME$-complete with
    unrestricted replay, and
  \item $\PTIME$-complete  (under logspace-reductions)  with bounded replay
    or without  replay.
  \end{enumerate}
The same results hold for DTDs in place of XML Schemas.
\end{restate}

\begin{proof}
  The upper bound in (a) follows from Theorem \ref{theo:simple}. The lower bounds in (a) and (b) follow from Lemma \ref{lemma:dreexptime} and Lemma \ref{lemma:dreptime}, respectively, as there the replacement rules are finite. It thus only remains to show the upper bound in (b).

The proof is quite similar to the proof of the upper bound in  Proposition \ref{prop:lanwacomplexity} (b). It combines an alternating logspace-computation, that simulates all plays of the game on the input string, with universal branching to divide, at each opening tag $\op{a}$, the processing of the remaining word into the processing of the subword until the corresponding closing tag $\cl{a}$ and the processing of the remaining word after that $\cl{a}$. 

We first describe how the game on an input string $w$ can be simulated by  an alternating logspace-computation. This part of the proof is very similar to the proof of the upper bound of Theorem 5.8 in \cite{MuschollSS06}. Let $k$ be the bound on the replay depth. We consider the equivalent version of cfGs in which \pone decides already when she reads an opening tag $\op{a}$, whether she wants \ptwo to rewrite a subword $u=\op{a}\cdots\cl{a}$. 

The idea is that the choices of \pone and \ptwo are simulated by existential and universal branching of the algorithm in the obvious fashion. However, if \pone calls an opening tag $\op{a}$  at some position $i$ and \ptwo replaces the corresponding subword $u=\op{a}\cdots\cl{a}$ of $w$ by a word $v$ from $R_a$ then the algorithm does not actually replace $u$ but rather stores the information that $u$ has been replaced by a pointer to $i$ and another pointer to $v$ (which is stored in the representation of $G$). As the replay depth is bounded by $k$, at each time at most $k$ such pairs of pointers are active, consuming at most $\bigO(\log (|G|))$ many bits. The test whether the resulting word (of each branch) is accepted by the target automaton $T$ is integrated into this branching process as follows. Each process maintains a current linear state $p$ reflecting the state of $T$ in the unique computation on the prefix of the current string, that is, if the current game configuration is $(\ponea,w_1,w_2)$, the current state is the one obtained by $T$ after reading $w_1$.   
Whenever \pone reads an opening tag $\op{a}$, the computation universally branches into two subcomputations.  The first subcomputation checks whether \pone has a winning strategy in the game on the subword between $\op{a}$ and its corresponding closing tag. 
The other subcomputation continues after that closing tag in the state determined by the target state function. \pone can only win if both subcomputations accept. Each subcomputation may recursively branch in the same way. When a subcomputation reaches a closing tag $\cl{a}$ it accepts if the current linear state is in $\Floc(a)$ and rejects otherwise. It is not hard to see that this algorithm has an accepting run on a word $w$ if and only if \pone has a winning strategy on $w$. As the algorithm only uses logarithmic space it witnesses the desired \PTIME upper bound.
\end{proof}

\newpage		\section*{Proofs for Section \ref{sec:other}}

In this section, we give proofs for our results concerning parameter validation and games with insertion stated in Section \ref{sec:other}.

\subsection*{Validation of parameters}
In this subsection, we consider cfGs with \emph{parameter validation}, i.e. games of the form $G = (\Sigma, \Gamma, R, V, T)$ which have an additional \emph{validity relation} $V \subs \Gamma \times \wf$. We will generally assume each \emph{validation language} $V_a \mydef \{w \in \wf \mid (a, w) \in V\}$ (for $a \in \Gamma$) to be a nonempty nested word language conforming to some specification (e.g. NWA, DTD or XML Schema). The semantics of such games is similar to the general semantics for cfGs, except for the fact that, in a configuration $(\ponea,u\op{a}v,\cl{a}w)$, \pone is only able to play \Call on $\cl{a}$ if it holds that $\op{a}v\cl{a} \in V_a$. Note that, while it isn't strictly necessary to pass the outermost $a$ to $V_a$ along with $v$, we still do so in order to easier describe $V_a$ as a language of trees with root node labelled $a$.

As mentioned in Section \ref{sec:other}, we restrict our attention to games without replay, as we are seeking to identify tractable cases, and $\Safelr$ is already $\PSPACE$-hard for bounded-replay games with target languages specified by DTDs \emph{without} parameter validation.

\subsubsection*{Upper bounds}

First off, we prove tractability for a restricted class of validation cfGs. As notation used in the proof, we say that a function symbol $g$ is ``from $V_f$'', if $V_g=V_f$.

\nc{\proset}[2][\sigma]{\ensuremath{\calP_{#1}(#2)}}
\nc{\lab}{\ensuremath{\text{label}}}

\begin{restate}{Theorem \ref{theo:validationptime}}
For the class of games with validation with a bounded number of validation DTDs and target languages specified by DTDs, $\Safelr$  is in  $\PTIME$ without replay.
\end{restate}

\begin{proof}(sketch)
	The basic proof idea for this result follows a similar approach to that used in \cite{MiloAABN05}: Going through the input string (interpreted as a tree) in a bottom-up fashion, we check for each node's child string whether it (and the subtree below it) can be rewritten to fit the target and verification languages in a replay-free manner. This allows us to tell whether \pone is able to safely play \Read or \Call on the node whose child string we just examined, and possibly on ancestor nodes as well. In this manner, deciding $\safelr(G)$ basically boils down to performing a polynomial number of safe rewritability tests for replay-free games on flat strings, which are each feasible in polynomial time by Corollary \ref{coro:flatptime}.
	
 For the sake of simple presentation, we identify trees and their nested word linearisations throughout this proof.
	
	As described above, our goal is to subsequently remove subtrees in a bottom-up manner and only consider flat strings of leaf node labels. More precisely, each removal step replaces a subtree of depth one, that is, a node $v$ whose children are all leaves, by a single node with a label that contains all relevant information about its (former) subtree with respect to the game. 
If, for instance, the subtree below a node $v$ with function symbol $f$ cannot be rewritten to conform to the corresponding part of some  DTD $V_f$, this information will be encoded into the label of $v$ and
\pone will never be able to play \Call on $v$ or any of its ancestors with a function symbol from $V_f$, no matter her rewriting capabilities on other parts of the input tree.

Let $t$ be the tree representing some well-nested rooted\footnote{For simplicity, we do not consider non-rooted words in this proof. They can be handled similarly.} word $w$. By $\lab(v)$ we denote the label of a node $v$.
By $S$ we denote the set $\{T,V_1,\ldots,V_d\}$ of schemas of the game. The \emph{profile} $P(t') \subs S$ of a tree $t'$ is the set of schemas for which $t'$ is valid. We first consider subgames on subtrees $t_v$ rooted at some node $v$ with label $a$. With each replay-free strategy $\sigma$ on $t_v$ that does \emph{not} play \Call on $v$ itself, we associate the \emph{profile set} $\proset{t_v}$ of profiles $P$, for which \ptwo has a counterstrategy yielding a tree $t'$ with $P=P(t')$. The \emph{dossier} $\calD(v)$ of $v$ is the set of all sets $X$, for which there is a strategy $\sigma$ of \pone such that 
$\proset{t_v}\subs X$. In our words, $\calD(v)$ is the closure of the set of all sets $\proset{t_v}$ under taking supersets.\footnote{The reason why we do not aim just at the set of all sets $\proset{t_v}$ will become clearer below.}

In the bottom-up computation mentioned above, we plan to replace the subtree below each node $v$ with label $a$ and change $v$'s label to $(a,\proset{t_v})$. Once, this process reaches the root $\Root(t)$ of the tree, it can be instantly decided whether \pone has a wining strategy on $w$. Indeed, this is the case if and only if $\calD(\Root(t))$ contains a profile set $\calP$, such that every profile $P\in\calP$ contains the target schema $T$.

To illustrate the above definitions, we consider the special case $d=1$, that is, besides the target schema $T$ there is only one validation schema $V$. In this case, there are four possible profiles of trees: $\{V,T\}$, $\{V\}$, $\{T\}$, $\emptyset$. As an example, a tree has profile $\{V\}$ if it is valid with respect to $V$ but not with respect to $T$. 

The four different profiles yield $2^4=16$ possible profile sets and $2^{16}=65536$ candidate dossiers. 
However, only the following six cases need to be distinguished: 

 	\begin{itemize}
	\item $\{\{V,T\}\}\in\calD$: \pone has a strategy that guarantees to yield a tree $t'$ that is valid with respect to both schemas;
        \item $\{\{T\}\}\in\calD$ and $\{\{V\}\}\in\calD$: \pone has a strategy that guarantees a tree $t'$ in $T$ and a strategy that guarantees a tree in $t''$ in $V$, but neither $t'$ nor $t''$ is valid with respect to the other schema;
	\item $\{\{T\}\}\in\calD$, but $\{\{V\}\}\not\in\calD$: \pone has a strategy that guarantees a tree $t'$ in $T$, but no strategy that guarantees a tree in $t''$ in $V$;
	\item $\{\{V\}\}\in\calD$, but $\{\{T\}\}\not\in\calD$: \pone has a strategy that guarantees a tree $t''$ in $V$, but no strategy that guarantees a tree in $t'$ in $T$;
        \item $\{\{V\},\{T\}\}\in\calD$: \pone has a strategy that guarantees to yield a tree that is either in $T$ or in $V$, but she can not enforce either of the two;
	\item $\calD=\{\{\emptyset\}\}$: no matter how \pone plays, \ptwo can always enforce a tree that is invalid for both $T$ and $V$.
	\end{itemize}
In all lower cases, we assume that  none of the upper cases applies.

We now start with the detailed description of the algorithm. We assume\footnote{As content models are given by \emph{deterministic} regular expressions, these DFAs can be computed efficiently.} that all content models of DTDs are given by DFAs.  

As stated above, the algorithm works in a bottom-up fashion. First, for all leaf nodes, their dossier is computed. As there is no actual subgame on a leaf node $v$ (that does not play \Call on that node), each such dossier is just $\{\{P(t_v)\}\}$. In this case, $P(t_v)$ is just the set of schemas in which the (original) label of $v$ is allowed at a leaf node.

The key step that the algorithm performs is to compute the dossier of a node $v$ with children $u_1,\ldots,u_m$ all of whose dossiers are already given. The idea is to compute $\calD(v)$ with the help of replay-free games on flat strings, whose winning problem can be decided thanks to Corollary \ref{coro:flatptime}.

For these flat games, the algorithm needs to compute, in a preprocessing phase that only depends on $G$, flat replacement sets $R'_f$, for every function symbol $f\in\Gamma$. As replacement strings represent strings in which no further \Call moves are possible, the labels of their positions do not include dossiers but rather the profile of the actual tree that they represent. 

Each set $R'_f$ can be computed as follows. Let $L_f$ denote the content model of $f$ in $V_f$ (represented by some DFA $A_f$). For each symbol $a$ occurring in $L_f$, let $\Sigma_{f,a}$ be the set of all pairs $(a,P)$, such that there is a tree $t'$ with profile $P$ and root label $a$ that is valid with respect to $R_f$. For each $f$, $a$ and $P$, it can be decided in polynomial time whether $(a,P)\in\Sigma_{f,a}$ by constructing a deterministic tree automaton that accepts all trees that are valid with respect to $R_f$  and the schemas in $P$, and invalid with respect to the schemas in $S\setminus P$. As $d$ is fixed, this amounts to an emptiness test for the polynomial-size product of $d+1$ deterministic tree automata. It follows that all sets $\Sigma_{f,a}$ can be computed in time polynomial in the size of $G$.\footnote{Since the number of validation schemas, and thus also $|S|$, is fixed, the fact that we need superpolynomial time in $|S|$ is of no consequence here.}

The set $R'_f$ consists of all strings over $\bigcup_{a\in\Sigma} \Sigma_{f,a}$ whose $\Sigma$-projection is in $L_f$. Given the sets  $\Sigma_{f,a}$, a DFA for $R'_f$ can be easily (and efficiently) computed.

Now, with the schemas $R'_f$ at hand, we describe the computation of $\calD(v)$ from  $u_1,\ldots,u_m$ and their dossiers in more detail.

For a dossier $\calD=\{\calP_1,\ldots,\calP_\ell\}$ and a symbol $a$, let $s(a,\calD)$ denote the string\footnote{The order of the profile sets in  $s(a,\calD)$ is inessential. We can assume just some ordering of profile sets.} \mbox{$(a,\calD)(a,\calP_1)\cdots(a,\calP_\ell) \cdot \#_a$}.

The idea behind the construction of the flat game is as follows.

The original game on a tree $t_z$ with root label $g$ (where $z$ is a child of the current root node $v$) can be viewed as follows: \pone chooses a strategy for the first phase of the game before the closing tag $\cl{g}$ of $z$ is reached. This strategy corresponds to some profile set $\calP_i\in\calD(z)$. By choosing a counterstrategy for this subgame, \ptwo basically picks  a profile $P\in \calP_i$. Then \pone decides whether she plays \Call at   $\cl{g}$ (subject to validity with respect to $V_g$) and \ptwo replaces $z$, in case she plays \Call.

In the flat game on $(g,\calD)(g,\calP_1)\cdots(g,\calP_\ell) \#_g$ this is mimicked as follows: \pone chooses her  strategy by playing \Call at $(g,\calP_i)$. \ptwo replaces $(g,\calP_i)$ by some pair $(g,P)$ with $P\in P_i$. So far the games exactly mimicks the original game \emph{before} reaching $\cl{g}$. If $P$ allows \pone to play \Call at $\cl{g}$ (that is, if $V_g\in P$), she can call the follow-up symbol $\#_g$ which is then replaced by \ptwo with a string from $R'_g$. The case that \pone cheats by playing \Call although $V_g\not\in P$ can be easily detected by the target automaton (whose construction will be explained soon, otherwise).

For each of the $2^{2^{d+1}}$ possible profile sets $\calQ$, the algorithm determines the winner for a particular replay-free game on the string $s(\lab(u_1),\calD(u_1)),\ldots,s(\lab(u_m),\calD(u_m)))$ with replacement sets
\begin{itemize}
\item $R'_a$, for every symbol $\#_a$ and
\item   $\{(a,P_1),\ldots,(a,P_j)\}$, for each symbol $(a,\calP)$ with
  $\calP=\{(a,P_1),\ldots,(a,P_j)\}$.
\end{itemize}
 It only remains to specify the target language of the game, which, of course, depends on $\calQ$. The DFA $A_\calQ$ for the target language for profile set $\calQ$ has to determine whether \pone has a winning strategy in the (original) subgame on $t_v$ that yields a tree with a profile in $\calQ$. 

To this end, $A_\calQ$ ignores all symbols that do not represent actual subtrees in the original game, that is,
\begin{itemize}
\item all symbols $(g,\calD)$, as they only indicate the beginning of a substring for some node;
\item all symbols $(a,\calP)$ with profile sets $\calP$ as they correspond to strategy options for \pone that she did not choose; and
\item all symbols $\#_g$ as they represent cases in which \pone played \Read and the respective subtree is represented by the symbol $(g,P)$, chosen by \ptwo;
\end{itemize}
We call all other symbols \emph{relevant}.

Thus, $A_\calQ$ accepts all strings $y$ resulting from the game, for which the subsequence $y'$ of relevant symbols is consistent with some profile $P\in\calQ$. That is, if\footnote{As we did not require that $\calD(v)$ consists exactly of all profile sets $\proset{t_v}$, we do not need to ensure anything for profiles not in $\calQ$.}
\begin{itemize}
\item for all  symbols $(q,P')$ of $y'$ it holds $P\subs P'$ and,
\item for each
  schema $D\in P$ the $\Sigma$-projection of $y'$ is in (the language of) $D$.
\end{itemize}
As $d$ is fixed, $A_\calQ$ is of polynomial size.

This completes the construction of the flat game and thus of the algorithm.

Each of the bottom-up reduction steps amounts to a (large but) constant number of tests whether \pone has a winning strategy in a flat game without replay and therefore can be done in overall polynomial time.

It is not too difficult but tedious to verify that the algorithm is also correct.
\end{proof}

\subsubsection*{Lower bounds}

\newcommand{\Gvals}[1][]{\ensuremath{\calG_{\text{val,simple}}^{#1}}}
\newcommand{\Gvald}[1][]{\ensuremath{\calG_{\text{val,DTD}}^{#1}}}

In this subsection, we prove lower bounds for less restricted classes of validation cfGs. We prove the lower bounds of Theorem \ref{theo:validationlower} as single results in the order in which they were stated in Section \ref{sec:other}: from most expressive to least expressive target, replacement and validation languages.

\begin{theorem}\label{theo:nwavalidationlower}
	For the class of validation games with target, validation and replacement languages specified by DNWAs, $\Safelr$ is $\EXPTIME$-hard without replay. This lower bound already holds for games with one single function symbol.
\end{theorem}

\begin{proof}
	We show $\EXPTIME$-hardness by reduction from the intersection emptiness problem for deterministic nested word automata: Given $n$ DNWAs $A_1, \ldots, A_n$, does it hold that $L(A_1) \cap \ldots \cap L(A_n) = \emptyset$? That this problem is $\EXPTIME$-hard follows directly from the $\EXPTIME$-hardness of the intersection emptiness problem for deterministic top-down finite tree automata \cite{seidl94}.
	
	Given DNWAs $A_1, \ldots, A_n$ over an alphabet $\Sigma$, we construct a game $G$ and input string $w$ such that \pone has a winning strategy on $w$ in $G$ if and only if there is no string $v \in \wf$ accepted by all $n$ automata. The game $G$ uses the alphabet $\Sigma \cup \{s, t\}$, with $s,t \notin \Sigma$, and $t$ being the only function symbol of $G$.
	
	The input string is $w=\op{t}^{n+1}\op{s}\cl{s}\cl{t}^{n+1}$, i.e. the tree linearised by $w$ is simply a path of length $n+2$ whose $n+1$ non-leaf nodes are labelled $t$ and whose leaf is labelled $s$. According to $G$, play on $w$ should proceed as follows: First, \pone plays \Call on the first $\cl{t}$ in $w$ (i.e. the innermost $t$). We emphasize that $t$ is the only function symbol and is therefore used for two different purposes in this proof.

 \ptwo replies to this call by providing some string $v \in \wf$; if possible, \ptwo will want to choose as $v$ a string contained in the intersection of all $L(A_i)$ for $i \in [n]$. \pone, in turn, will try to show that there is some $i \in [n]$ such that $v \notin L(A_i)$; she does so by playing \Call on the $i$-th remaining $\cl{t}$ in the rewritten string $\op{t}^n v \cl{t}^n$. The validation language for $t$ will ensure that this \Call is only possible if $v$ is indeed not in $L(A_i)$. \ptwo can reply to such a \Call by \pone with an arbitrary string in $\wf$. However, the actual choice of this string is inconsequential as all is needed for \pone to win is that there are less than $n$ occurrences of $\cl{t}$ in the resulting string. 
	
	More formally, the game $G$ over alphabet $\Sigma \cup \{s, t\}$ with $\Gamma=\{t\}$ is defined with the replacement language $R_t = \wf$ and validation language $V_t = \{\op{t}\op{s}\cl{s}\cl{t}\} \cup \{\op{t}^{i} v \cl{t}^i \mid v \in \wf \setminus L(A_i), i \in [n] \}$. The target language is $T = \{\op{t}^i v \cl{t}^i \mid v \in \wf, 0 \leq i < n\}$.
	
	$G$ can be efficiently computed from $A_1, \ldots, A_n$, as DNWAs can be complemented in polynomial time, and given DNWAs for $\wf \setminus L(A_i)$ for each $i \in [n]$, DNWAs for $V_t$ and $T$ can easily be constructed.
	
	Clearly, any strategy $\Astrat$ for \pone on $w$ in $G$ that does not \Call the first $\cl{t}$ cannot be a winning strategy, as the target language does not contain any $s$ tags and the only part of the validation language containing $s$ tags only ever applies to the innermost $t$. From there, it is straightforward to prove that \pone has a winning strategy on $w$ if and only if \ptwo can not respond to this first \Call with a string that is contained in the intersection of all $L(A_i)$ for $i \in [n]$, i.e. iff $L(A_1) \cap \ldots \cap L(A_n) = \emptyset$.
\end{proof}

\begin{theorem}\label{theo:simplevalidationlower}
	For the class of validation games with target, validation and replacement languages specified by XML Schemas, $\Safelr$ is $\PSPACE$-hard. This lower bound already holds for games with one single function symbol, whose replacement and target language are given by DTDs and whose replacement language is finite.
\end{theorem}

\begin{proof} (sketch)
	We prove this claim by giving a reduction from the problem \QBF of determining for a given quantified Boolean formula $\Phi$, whether that $\Phi$ is true. We assume that the input formula is of the form $\Phi = Q_1 x_1 \ldots Q_n x_n \varphi(x_1, \ldots x_n)$ with $Q_i \in \{\exists, \forall\}$ for all $i \in [n]$ and a Boolean formula $\varphi = C_1 \lor \ldots \lor C_m$ with $m$ clauses in \emph{disjunctive} normal form. Without loss of generality, we further assume that no clause contains both $x_i$ and $\neg x_i$ for any $i \in [n]$.
	
	We construct from $\Phi$ a validation game $G$ with a single function symbol $f$, and an input string $w$ such that \pone has a winning strategy on $w$ in $G$ if and only if $\Phi$ is true. We first sketch the construction and the manner in which play proceeds according to $G$ before giving a formal construction. For the sake of simpler presentation, we identify trees and their nested word linearisations. 
	
	The input string consists of a path of $m$ \emph{clause nodes} each labelled $f$.
	 Below the final clause node,
	 $w$ consists of a "spine" of \emph{backbone nodes} with labels $b_1$ to $b_{n+1}$, where each $b_i$ has as its left child a \emph{variable node} 
	 labelled $f$, and as its right child a node  labelled $b_{i+1}$. The leaf node $b_{n+1}$ terminates this chain.
	 
	 As should be obvious from this description, nodes labelled $f$ in this string serve different purposes, depending on their placement in $w$. This is reflected by the single-type tree grammar for the validation language $V_f$ having several different types for nodes which may be labelled $f$. In principle, the clause nodes may be assigned types from $\{C_1, \ldots, C_m\}$, while the variable node child of each node labelled $b_i$ for some $i$ will (usually) be typed as $x_i$. Additionally, the tree grammar for $V_f$ may assign to each node labelled $b_i$ a type from $\{b_i^1, \ldots, b_i^m\}$. The exact purpose of these types will become clear in the rest of the proof.
	
	Play on the input string $w$ proceeds as follows: In a left-to-right order, the variable nodes of type $x_1$ to $x_n$ are the first to be played on (in the same order as the variables $x_1, \ldots, x_n$ are quantified in $\Phi$). A play rewriting these nodes establishes an assignment $\alpha$ of truth values to the variables $x_1, \ldots, x_n$, with \pone choosing assignments for existentially quantified variables and \ptwo choosing for universally quantified variables.
	
	Afterwards, \pone is supposed to select one clause $C_i$ that evaluates to "true" under $\alpha$ by playing \Call on the 
	$i$-th clause node from the bottom, with \ptwo replacing it and thus truncating the input tree to end in a leaf after 
	a path of $f$ nodes. The validation language will ensure that \pone is only allowed to play \Call on the $i$-th clause node from the bottom if $\alpha$ indeed satisfies $C_i$. If \pone manages to play \Call on any clause node, she wins the game, otherwise she loses.
	
	We sketch in some more detail how \pone and \ptwo construct a variable assignment before giving formal details on the construction. For universally quantified variables $x_i$, the matter is simple: No validation (or target) language will be able to match type $x_i$ in this position, so \pone is forced to play \Call on it, giving \ptwo the opportunity to replace it with $0$ or $1$ (which is then interpreted as setting $x_i$ to be true resp.\ false under $\alpha$). For existentially quantified variables $x_i$, the binary choice of setting $x_i$ to true or false is modelled by \pones choice whether or not to \Call the symbol $f$ of type $x_i$: The replacement language for type $x_i$ also has to be $\{0,1\}$ (as there is only the single function symbol $f$ used as a label for all function types), so in this case an uncalled $x_i$ is interpreted as setting $x_i$ to be true under $\alpha$, while both $0$ and $1$ will be interpreted as setting $x_i$ to be false. Note that the replacement language $R_f = \{\op{0}\cl{0}, \op{1}\cl{1}\}$ thus constructed is finite and definable by a DTD.
	
	Formally, the game $G$ is over the alphabet $\Sigma = \{b_1, \ldots, b_{n+1}, 0, 1, f\}$ with a single function symbol $\Gamma = \{f\}$ and replacement language $R_f = \{\op{0}\cl{0}, \op{1}\cl{1}\}$.
		
	The target language consists of all strings linearising paths containing only non-leaf nodes
	labelled $f$ and ending in a leaf node labelled $0$ or $1$. Again, it is easy to see that this language can be represented by a DTD.
	
	Variable nodes should always allow \pone to \Call them on the input string described above, so $\epsilon \in V_{f}$. Furthermore, the part of the validation language used at each clause node of 
	distance $i$ from the bottom $b_1$ node should accept exactly those subtrees encoding satisfying assignments for $C_i$. 

	Altogether, we can give a tree grammar $T_f$ to define the schema for $V_f$. This grammar uses the label alphabet $\Sigma$ (as above), type alphabet $\Delta = \{x_i, b^j_i, C_j \mid i \in [n], j \in [m|\} \cup \{b_{n+1}, 0, 1,x\}$ with a labelling function $\lambda$ mapping all types that are also symbols in $\Sigma$ to themselves, all $b_i^j$ (for $j \in [m]$) to $b_i$ and all other types to $f$, and the following productions (with $C_1$ being the start symbol):
	
	\begin{itemize}
	\item $C_1 \rightarrow C_2 + b_1^1 + \epsilon$\footnote{This $\epsilon$ rule accommodates the special case of a variable node being called.}
	\item $C_j \rightarrow C_{j+1} + b_1^j$ for $2 \leq j < m$,
	\item $C_m \rightarrow b^m_1$,
	\item for all $i \in [n], j \in [m]$:
		\begin{itemize}
		\item $b_i^j \rightarrow x_i b^j_{i+1}$ if $x_i$ is positive in $C_j$ and existentially quantified in $\Phi$,
		\item $b_i^j \rightarrow 1 b^j_{i+1}$ if $x_i$ is positive in $C_j$ and universally quantified in $\Phi$,
		\item $b_i^j \rightarrow (0+1) b^j_{i+1}$ if $x_i$ is negative in $C_j$ and existentially quantified in $\Phi$,
		\item $b_i^j \rightarrow 0 b_{i+1}^j$ if $x_i$ is negative in $C_j$ and universally quantified in $\Phi$,
		\item $b_i^j \rightarrow (x_i + 0 + 1) b_{i+1}^j$ if $x_i$ is not in $C_j$
		\end{itemize}
	\item $b^j_{n+1} \rightarrow \epsilon$.
	\end{itemize}

	It is clear to see that this grammar is indeed single-type, and that all of its content models are specified by deterministic regular expressions.

	We can now explain in detail the exact purpose of the types defined above: When \pone and \ptwo construct a variable assignment, variable nodes can be matched to type $C_1$ in the above grammar and (as they are leaves) accepted with child string $\epsilon$. After the variable assignment has been constructed, if \pone calls the $j$-th clause node from the bottom, that node is matched to type $C_1$, with subsequent child clause nodes being matched to clause types with increasing clause numbers. The bottom clause node is matched to $C_j$. Since that node's child is labelled $b_1$, it has to be matched to $b_1^j$, and this upper index $j$ is "carried down" through the backbone nodes, making certain that each backbone node "knows" which clause is to be checked. The sub-grammars for each type $b_i^j$ then takes care of checking whether the truth assignment constructed by \pone and \ptwo indeed fulfils clause $C_j$.
		
	The correctness of this correctness is proven as follows: Each play on the variable nodes induces an assignment to the variables $x_1 \ldots x_n$ compliant with their quantification in $\Phi$ (and vice versa), and the subtree starting at 
	the $i$-th clause node from the bottom is in $V_{f}$ if and only if it has been rewritten to correspond to a variable assignment satisfying $C_i$. This directly implies that \pone has a winning strategy on $w$ in $G$ if and only if $\Phi$ is true, which concludes the reduction.
\end{proof}

As seen in the proof of Theorem \ref{theo:validationptime}, the running time of the algorithm we give for deciding $\Safelr$ with target, replacement and and verification DTDs grows superpolynomially in the parameter $d$, i.e. the number of function symbols. The following result shows it is unlikely that one can avoid such behaviour.

\begin{theorem}\label{theo:validationpspacelower}
	  For the class of games with validation, $\Safelr$ (without replay) is $\PSPACE$-hard, for games with an unbounded number of validation DTDs and replacement and target languages specified by DTDs.
\end{theorem}

\begin{proof} 
	This follows from the proof of Theorem \ref{theo:simplevalidationlower}, with slight modifications. As in that proof, we show $\PSPACE$-hardness by reduction from \QBF, with the quantor-free part of the input formula in disjunctive normal form.
	
	First off, note that each single-type tree grammar may be seen as a DTD over its type alphabet. More precisely, if $T = (\Sigma, \Delta, S, P, \lambda)$ is a single-type tree grammar, then the tree grammar $T' = (\Delta, \Delta, S, P, \id_\Delta)$ (where $\id_\Delta$ is the identity function on $\Delta$ mapping each type to itself) is local. We use this fact to construct from the verification language $V_f$ given in the proof of Theorem \ref{theo:simplevalidationlower} several validation DTDs, with the number of function symbols (and corresponding validation languages) constructed in the reduction from \QBF growing with the number of clauses and variables of the input formula.
	
	The input string $w$ is similar to the one from the proof of Theorem \ref{theo:validationpspacelower}, consisting of a path of $m$ clause nodes and, below them, a subtree made up of variable and backbone nodes. Other than in that proof, however, the clause nodes are already labelled $C_1$ to $C_m$ (with $C_1$ labelling the topmost node, directly below the root). The variable nodes are already labelled $X_1$ to $x_n$ right from the start.
	
	The target language is almost the same as in the proof of Theorem \ref{theo:validationpspacelower} (accounting, however, for clause node labels), and the replacement language is identical to the one given there. The validation languages for each $x_i$ simply consists of a singleton node labelled $x_i$.
	
	Validation languages for each $C_j$ are obtained from the tree grammar $T_f$ for $V_f$ given in the proof of Theorem \ref{theo:validationpspacelower} as the sub-grammars of $T_f$ starting at $C_j$, with the only difference being that each $b_i^j$ is simply replaced by $b_i$. This is because the purpose of the upper index $j$ in that proof was carrying the clause number selected by \pone down through the variable assignment subtree. This technique is no longer necessary here, due to the fact that the validation language for each $C_j$ is separate, which means that the clause to be checked is inherent to its corresponding validation DTD and thus already "known" to all of its variables.
	
	The correctness of this construction is shown as in the proof of Theorem \ref{theo:validationpspacelower}.
\end{proof}

\subsection*{Insertion rules}
We consider here cfGs with insertion instead of replacement rules, i.e. games of the form $G=(\Sigma, \funcsymb, I, T)$ where the insertion relation $I \subs \funcsymb \times \wf$ takes the place of the replacement relation $R$ from cfGs as defined in Section \ref{sec:prelim}. The semantics is similar to standard cfGs, except for the definition of follow-up configurations after a \Call move by \pone. We recall that we consider three different semantics here: the \emph{general setting}, where \pone may play another subgame on the substring she just called; the \emph{weak replay} setting, where \pone only gets to play on the newly inserted substring after a \Call; and the setting \emph{without replay}, where the play proceeds to the right of a newly inserted substring without modifying it.

Our restriction to games having \emph{only} insertion rules is primarily to simplify the presentation of our proofs. It is relatively easy (if tedious) to prove that games can be extended to contain both replacement and insertion rules without changing the complexity of $\Safelr$, as long as appropriate semantics for insertion and replacement rules are chosen. 

We generally assume $G=(\Sigma,\funcsymb,I,T)$ to be an insertion game with target language $T$ represented by a DNWA $A(T)$ and insertion languages $I_a$ represented by an arbitrary NWA for each $a \in \funcsymb$.

We restate  Proposition \ref{theo:insertion} for easier reference.

\begin{restate}{Proposition \ref{theo:insertion}}
  For the class of games with insertion semantics, target DNWAs and replacement NWAs, $\Safelr$ is
  \begin{enumerate}[(a)]
  \item undecidable in general;
  \item $\iiEXPTIME$-complete for games with weak replay; and
  \item $PSPACE$-complete for games without replay.
  \end{enumerate}
\end{restate}

Before proving Proposition \ref{theo:insertion}, we prove two auxiliary results showing a strong correspondence between replacement games and insertion games (with appropriate semantics). Recall that for a replacement game $G$, $\safelr(G)$ denotes the set of all winning strings for \pone in $G$ with unbounded replay and $\safelr[1](G)$ without replay; similarly, we denote the winning set for \pone in an insertion game $G'$ by $\safelr(G')$ in the general setting, by $\safelr[1+](G')$ with weak replay, and by $\safelr[1](G)$ without replay.

\begin{lemma}\label{lemma:replacement2insertion}
	There exists a polynomial-time algorithm that, given a replacement cfG $G = (\Sigma,\funcsymb,R,T)$ and nested word $w \in \wf$, outputs an insertion cfG $G' = (\Sigma',\funcsymb,I,T')$ and word $w' \in \wf[\Sigma']$ such that
        \begin{itemize}
        \item $w \in \safelr(G) \Leftrightarrow w' \in
          \safelr[1+](G')$, and
        \item $w \in \safelr[1](G) \Leftrightarrow w'
          \in \safelr[1](G')$.
        \end{itemize}

\end{lemma}

\begin{proof}
The main observation we need is that replacement in cfGs is generally very localised, i.e. a \Call on $\cl{a}$ in a string of the form $u \op{a}v\cl{a}$ only affects $\op{a}v\cl{a}$, the shortest well-nested suffix of the current string up to $\cl{a}$.

The obvious idea behind the proof is to simulate replacement rules with insertion rules. The crucial insight for this simulation is that, while we cannot delete the rooted suffix $ w = \op{a}v\cl{a}$ from a current string, the new target automaton $A(T')$ can ``undo" the effect of $w$ on $A(T)$ by reverting it to the state it had before reading $w$. To this end, $A(T')$ simulates $A(T)$, all the while memorising (in its state) a ``fallback state" that $A(T)$ was in before beginning to read $w$. That way, $A(T')$ can always revert its simulation of $A(T)$ to the point before $w$ was read, effectively making $A(T)$ ``forget" $w$ and thus simulating a replacement of $w$.

In this way, it is easy to simulate deletion of suffixes that would be replaced in $G$, so we only need some way of knowing \emph{when} such a deletion should take place. To this end, we encapsulate replacement strings $u$ for $G$ within \emph{backspace} tags as $\op{b}u\cl{b}$ (with $b \notin \Sigma$). Now, when the automaton $A'$ reads a $\op{b}$, it knows that what follows after is supposed to be a replacement string, so it ``forgets" the last rooted suffix of the current string, jumps back to the last fallback state and continues simulating $A(T)$ on $u$, re-setting its fallback state along the way as necessary.

Formally, let $A(T) = (Q, \Sigma, \delta, q_0, F)$ be a DNWA in normal form for $T$ and let $b \notin \Sigma$. We define $G'$ as follows:

\begin{itemize}
	\item $\Sigma' = \Sigma \cup \{b\}$
	\item $I_a = \{\op{b}u\cl{b} \mid u \in R_a\}$ for all $a \in \funcsymb$ and
	\item $T' = L(A')$ for the DNWA $A'$ defined below. 
\end{itemize}

The automaton $A' = (Q', \Sigma \cup \{b\}, \delta', q'_0, F')$ is defined by

\begin{itemize}
	\item $Q' = Q \times Q$;
	\item $q'_0 = (q_0, q_0)$;
	\item $F' = F \times Q$ and
	\item $\delta'((p,q), \op{a}) = (\delta(p,\op{a}),p)$ for all $a \in \Sigma$,
	\item $\delta'((p,q), \op{b}) = (q,q)$,
	\item $\delta'((p,q), (p',q'), \cl{a}) = (\delta(p,p', \cl{a}),q)$ for all $a \in \Sigma$ and
	\item $\delta'((p,q),(p',q'), \cl{b}) = (p,q)$.
\end{itemize}

In keeping with the above intuition, $A'$ tracks in its state $(p,q)$ a \emph{current} state $p$ and a \emph{fallback} state $q$ of $A(T)$. When $A'$ reads an $\op{a}$ (resp. $\cl{a}$), it knows that the rooted string immediately to the left of $\op{a}$ has \emph{not} been replaced, so it simulates a step of $A(T)$ to obtain a new current state and sets the new fallback state to be the state $A$ had immediately before reading $\op{a}$ (respectively the $\op{a}$ associated with the current $\cl{a}$).

On reading $\op{b}$, $A'$ knows that the last minimal nested string has been replaced in $G$, so it returns its simulation of $A(T)$ to the fall-back state and simulates $A(T)$ on the replacement string following after $\op{b}$ from there. On $\cl{b}$, neither the current nor fallback state changes, as the last rooted string to the right of $\cl{b}$ may be considered the last minimal suffix of the current word in the replacement game.

If, after reading a string and simulating $A(T)$ on it as described above, $A(T)$ accepts (i.e. the current state of $A'$ is in $F$), $A'$ accepts as $A(T)$ would.
\end{proof}

\begin{lemma}\label{lemma:insertion2replacement}
	There exists a polynomial-time algorithm that, given an insertion cfG $G = (\Sigma,\funcsymb,I,T)$ and nested word $w \in \wf$, outputs a replacement game $G' = (\Sigma',\funcsymb,R,T')$ and word $w' \in \wf[\Sigma']$ such that
        \begin{itemize}
        \item $w \in \safelr[1+](G) \Leftrightarrow w' \in
          \safelr(G')$, and
        \item $w \in \safelr[1](G) \Leftrightarrow w' \in
          \safelr[1](G')$.
        \end{itemize}

\end{lemma}

\begin{proof}
The basic idea behind simulating insertion games using replacement games is to replace every subword $\op{a}v\cl{a}$ of $w$ by $\op{a}\mu(v)\cl{a}\op{a'}\cl{a'}$ in $w'$ (where $a'$ is a new ``copy'' of $a$) and to simulate the insertion of a new substring to the right of $\op{a}v\cl{a}$ by the replacement of $\op{a'}\cl{a'}$. We refer to the additional substrings of the form $\op{a'}\cl{a'}$ as  ``anchors". 

To this end, we need to ensure that (a) no non-anchor substring ever gets replaced, and (b) each replacement string contains new anchors for further insertions. For part (a), we add extra symbols to the input alphabet, while part (b) is done through the transformation from $w\in \wf$ to $w' \in \wf[\Sigma']$ hinted at in the claim's statement.

More formally, we set $\Sigma' = \Sigma \cup \{a' \mid a \in \Sigma\}$, i.e. we add a second disjoint copy of $\Sigma$ to itself. Strings will generally be transformed using a function $\mu: \wf \rightarrow \wf[\Sigma']$ defined inductively by
\begin{itemize}
	\item $\mu(\epsilon) = \epsilon$
	\item $\mu(uv) = \mu(u) \mu(v)$ for all $u,v \in \wf$ and
	\item $\mu(\op{a}v\cl{a}) = \op{a}\mu(v)\cl{a}\op{a'}\cl{a'}$ for all $a \in \Sigma$, $v \in \wf$.
\end{itemize}

The target language of $G'$ is defined as $T' = \{\mu(w) \mid w \in T\}$; it is easy to see that a DNWA for $T'$ can be constructed from $A(T)$ by simply ignoring symbols from $\Sigma' \setminus \Sigma$.

The set of function symbols in $G'$ is just $\{a'\mid a\in \funcsymb\}$, and the replacement languages are defined by $$R_{a'} = \{\mu(w) \mid w \in R_a\}.$$ Again, it is easy to see that automata for each $R_{a'}$ can be computed from those for $R_a$ in polynomial time.

Finally, the input string gets transformed (in polynomial time) via $\mu$ as well: $w' = \mu(w)$.

\end{proof}

\begin{proofof}{Proposition \ref{theo:insertion}}
	Parts (b) and (c) follow directly from Lemmas \ref{lemma:insertion2replacement}, \ref{lemma:replacement2insertion} as well as Proposition \ref{prop:lowergeneral}. All that remains to be proven is therefore the undecidability of $\Safelr$ in the general setting.

Intuitively this holds because, on a string of the form $\op{a}v\cl{a}$, jumping back to the start after calling $\cl{a}$ effectively allows \pone to play arbitrarily many left-to-right passes on $v$, thereby enabling her to simulate \emph{any} (not just L2R-) strategy on $v$. We utilise this fact to give a reduction from the algorithmic problem to find out whether for a flat string $w$, and a context-free game $G$  with flat regular replacement and target languages and the ability of \pone to freely select positions (not only from left-to-right), \pone has a winning strategy. It was shown  in \cite{MuschollSS06} that this problem is undecidable. For precise definitions of these games we refer to \cite{MuschollSS06}. 

For the reduction, we construct a cfG $G_n$ from a given input flat cfG $G_r = (\Sigma,\funcsymb,R,T)$ with target language DFA $A = (Q, \Sigma, \delta, q_0, F)$ for $T$. The idea is to simulate an arbitrary strategy for \pone on some flat string $w \in \Sigma^*$ in $G_r$ by means of a L2R strategy on the nested word $\nw{w} \in \wf$ derived from $w \in \Sigma^*$ by replacing each symbol $a$ in $w$ with $\op{a}\cl{a}$.
	
	We make use of the relatively simple observation that an arbitrary strategy of \pone on $w$ (in which \pone may freely choose which position in the current string to \Call next) can easily be simulated by an unbounded number of left-to-right passes over the current string using only the moves \Read (which moves the current position within the string one step to the right), \Call (which does not change the current position) and additional \emph{left-step} (\LS) moves, which reset the current position to 0 once the end of the current string has been reached (cf. \cite{BjorklundSSK13}).
	
	The idea for the reduction, now, is to transform the input string $w$ into a string of the form $\op{r}\nw{w}\cl{r}$ (for some $r \notin \Sigma$), simulate each left-to-right pass for \pone on $w$ appropriately on $\nw{w}$ and then use a \Call on $\cl{r}$ to simulate a \LS move, appending some irrelevant ``tail" $\op{t}\cl{t}$ (for $t \notin \Sigma$) to the current nested string in the process.
	
	The only minor conceptual difficulty is how to simulate a left-to-right pass of \pone on $\nw{w}$ using insertion rules, as context-free games with non-nested regular languages are defined using only replacement in \cite{MuschollSS06}. This can be done with a similar technique as described in the proof of Lemma \ref{lemma:replacement2insertion} -- replacement strings $v$ from some replacement language $R_a \subs \Sigma^*$ are transformed into nested strings as above and encapsulated in ``backspace" tags as $\op{b}\nw{v}\cl{b}$ (for $b \notin \Sigma$); on reading an opening $\op{b}$, the target DNWA for $G_n$ ``forgets" the last nested string before the $\op{b}$ by restoring a fallback state of $A$. The only difference to the proof of Lemma \ref{lemma:replacement2insertion} is that here, the target DNWA for $G_n$ merely has to simulate a DFA, not a DNWA.
	\end{proofof}

\bibliographystyle{plain}
\bibliography{nwgames}

\begin{thebibliography}{10}

\bibitem{AbiteboulBM08}
Serge Abiteboul, Omar Benjelloun, and Tova Milo.
\newblock The {A}ctive {XML} project: an overview.
\newblock {\em VLDB J.}, 17(5):1019--1040, 2008.

\bibitem{AbiteboulMB05}
Serge Abiteboul, Tova Milo, and Omar Benjelloun.
\newblock Regular rewriting of active {XML} and unambiguity.
\newblock In {\em PODS}, pages 295--303, 2005.

\bibitem{AlurM09}
Rajeev Alur and P.~Madhusudan.
\newblock Adding nesting structure to words.
\newblock {\em J. ACM}, 56(3), 2009.

\bibitem{BjorklundSSK13}
Henrik Bj{\"o}rklund, Martin Schuster, Thomas Schwentick, and Joscha Kulbatzki.
\newblock On optimum left-to-right strategies for active context-free games.
\newblock In {\em Joint 2013 EDBT/ICDT Conferences, ICDT '13 Proceedings,
  Genoa, Italy, March 18-22, 2013}, pages 105--116, 2013.

\bibitem{Bozzelli07}
Laura Bozzelli.
\newblock Alternating automata and a temporal fixpoint calculus for visibly
  pushdown languages.
\newblock In {\em CONCUR- Concurrency Theory, 18th International Conference},
  pages 476--491, 2007.

\bibitem{ChandraKS81}
A.~K. Chandra, D.~Kozen, and L.~J. Stockmeyer.
\newblock Alternation.
\newblock {\em Journal of the ACM}, 28(1):114--133, 1981.

\bibitem{Chlebus86}
B.~S. Chlebus.
\newblock Domino-tiling games.
\newblock {\em Journal of Computer and System Sciences}, 32(3):374--392, 1986.

\bibitem{GraedelTW02}
E.~Gr\"{a}del, W.~Thomas, and T.~Wilke, editors.
\newblock {\em Automata, Logics, and Infinite Games. A Guide to Current
  Research}.
\newblock Springer, 2002.

\bibitem{Kaiser09}
Lukasz Kaiser.
\newblock Synthesis for structure rewriting systems.
\newblock In Rastislav Kr{\'a}lovic and Damian Niwinski, editors, {\em MFCS},
  volume 5734 of {\em Lecture Notes in Computer Science}, pages 415--426.
  Springer, 2009.

\bibitem{MartensNeven+Simple07}
Wim Martens, Frank Neven, and Thomas Schwentick.
\newblock Simple off the shelf abstractions for {XML} schema.
\newblock {\em SIGMOD Record}, 36(3):15--22, 2007.

\bibitem{MartensNSB06}
Wim Martens, Frank Neven, Thomas Schwentick, and Geert~Jan Bex.
\newblock Expressiveness and complexity of {XML} schema.
\newblock {\em ACM Trans. Database Syst.}, 31(3):770--813, 2006.

\bibitem{MiloAABN05}
Tova Milo, Serge Abiteboul, Bernd Amann, Omar Benjelloun, and Frederic~Dang
  Ngoc.
\newblock Exchanging intensional {XML} data.
\newblock {\em ACM Trans. Database Syst.}, 30(1):1--40, 2005.

\bibitem{MurataLMK05}
Makoto Murata, Dongwon Lee, Murali Mani, and Kohsuke Kawaguchi.
\newblock Taxonomy of {XML} schema languages using formal language theory.
\newblock {\em ACM Trans. Internet Techn.}, 5(4):660--704, 2005.

\bibitem{MuschollSS06}
Anca Muscholl, Thomas Schwentick, and Luc Segoufin.
\newblock Active context-free games.
\newblock {\em Theory Comput. Syst.}, 39(1):237--276, 2006.

\bibitem{PaulyP03a}
Marc Pauly and Rohit Parikh.
\newblock Game logic - an overview.
\newblock {\em Studia Logica}, 75(2):165--182, 2003.

\bibitem{seidl94}
H.~Seidl.
\newblock Haskell overloading is {DEXPTIME}-complete.
\newblock {\em Information Processing Letters}, 52(2):57--60, 1994.

\bibitem{Benthem03}
Johan van Benthem.
\newblock Logic games are complete for game logics.
\newblock {\em Studia Logica}, 75(2):183--203, 2003.

\bibitem{Waldmann02}
Johannes Waldmann.
\newblock Rewrite games.
\newblock In Sophie Tison, editor, {\em RTA}, volume 2378 of {\em Lecture Notes
  in Computer Science}, pages 144--158. Springer, 2002.

\end{thebibliography}

\end{document}